\documentclass[12pt,british,ngerman,english]{scrbook}
\usepackage{lmodern}
\usepackage[T1]{fontenc}
\usepackage[latin9]{inputenc}
\usepackage[a4paper]{geometry}
\usepackage{fancyhdr}
\usepackage{color}
\usepackage{babel}
\usepackage{array}
\usepackage{prettyref}
\usepackage{wrapfig}
\usepackage{textcomp}
\usepackage{multirow}
\usepackage{amsthm}
\usepackage{amsmath}
\usepackage{amssymb}
\usepackage{graphicx}
\usepackage{setspace}
\usepackage[numbers]{natbib}
\usepackage[unicode=true,pdfusetitle,
 bookmarks=true,bookmarksnumbered=false,bookmarksopen=false,
 breaklinks=false,pdfborder={0 0 0},backref=page,colorlinks=true]
 {hyperref}
\hypersetup{
 linkcolor=blue, urlcolor=blue, citecolor=blue}

\makeatletter

\let\pr@chap=\pr@cha
\newcommand{\noun}[1]{\textsc{#1}}
\providecommand{\tabularnewline}{\\}

\numberwithin{equation}{section}
\numberwithin{figure}{section}
\usepackage{enumitem}		
\theoremstyle{plain}
\newtheorem{thm}{\protect\theoremname}[section]
\theoremstyle{definition}
\newtheorem{defn}[thm]{\protect\definitionname}
\theoremstyle{plain}
\newtheorem{lem}[thm]{\protect\lemmaname}
\ifx\proof\undefined
\newenvironment{proof}[1][\protect\proofname]{\par
\normalfont\topsep6\p@\@plus6\p@\relax
\trivlist
\itemindent\parindent
\item[\hskip\labelsep\scshape #1]\ignorespaces
}{%
\endtrivlist\@endpefalse
}
\providecommand{\proofname}{Proof}
\fi
\theoremstyle{remark}
\newtheorem{rem}[thm]{\protect\remarkname}
\theoremstyle{plain}
\newtheorem{cor}[thm]{\protect\corollaryname}
\theoremstyle{definition}
\newtheorem{example}[thm]{\protect\examplename}
\theoremstyle{plain}
\newtheorem{prop}[thm]{\protect\propositionname}

\@ifundefined{date}{}{\date{}}

\newrefformat{cor}{Corollary~\ref{#1}}
\newrefformat{pro}{Proposition~\ref{#1}}
\newrefformat{def}{Definition~\ref{#1}}
\newrefformat{sub}{Subsection~\ref{#1}}
\newrefformat{exa}{Example~\ref{#1}}
\newrefformat{rem}{Remark~\ref{#1}}

\usepackage{url}

\setlist{itemsep=0pt}
\setlist[1]{leftmargin=2em}
\setlist[2]{leftmargin=2em}

\usepackage[font=small, width=.9\textwidth]{caption}

\newlength{\LyXMinipageIndent}
\usepackage[format=plain,labelfont=bf]{caption}
\DeclareCaptionLabelSeparator{sp}{ \;}
\newtheoremstyle{plain}{1em}{1em}{\itshape}{}{\bf}{}{.5em}{}
\newtheoremstyle{definition}{1em}{1em}{}{}{\bf}{}{.5em}{}

\makeatletter
\let\ps@plain\ps@empty
\makeatother

\usepackage[table]{xcolor}
\RequirePackage{booktabs}
\RequirePackage{tikz}
\usetikzlibrary{arrows,automata,positioning,shadows,patterns,decorations.pathreplacing}

\newcommand{\cloneFont}[1]{\mathsf{#1}}
\newcommand{\CloneBF}{\protect\ensuremath{\cloneFont{BF}}}
\newcommand{\CloneM}{\protect\ensuremath{\cloneFont{M}}}
\newcommand{\CloneL}{\protect\ensuremath{\cloneFont{L}}}
\newcommand{\CloneR}{\protect\ensuremath{\cloneFont{R}}}
\newcommand{\CloneD}{\protect\ensuremath{\cloneFont{D}}}
\newcommand{\CloneN}{\protect\ensuremath{\cloneFont{N}}}
\newcommand{\CloneS}{\protect\ensuremath{\cloneFont{S}}}
\newcommand{\CloneV}{\protect\ensuremath{\cloneFont{V}}}
\newcommand{\CloneE}{\protect\ensuremath{\cloneFont{E}}}
\newcommand{\CloneI}{\protect\ensuremath{\cloneFont{I}}}

\definecolor{g1}{gray}{0.7}
\definecolor{g3}{gray}{0.85}

\colorlet{WSAT}{g1}
\colorlet{nWSAT}{g3}

\tikzset{
WSAT/.style={fill=WSAT},
nWSAT/.style={fill=nWSAT},
}

\makeatother

\addto\captionsbritish{\renewcommand{\corollaryname}{Corollary}}
\addto\captionsbritish{\renewcommand{\definitionname}{Definition}}
\addto\captionsbritish{\renewcommand{\examplename}{Example}}
\addto\captionsbritish{\renewcommand{\lemmaname}{Lemma}}
\addto\captionsbritish{\renewcommand{\propositionname}{Proposition}}
\addto\captionsbritish{\renewcommand{\remarkname}{Remark}}
\addto\captionsbritish{\renewcommand{\theoremname}{Theorem}}
\addto\captionsenglish{\renewcommand{\corollaryname}{Corollary}}
\addto\captionsenglish{\renewcommand{\definitionname}{Definition}}
\addto\captionsenglish{\renewcommand{\examplename}{Example}}
\addto\captionsenglish{\renewcommand{\lemmaname}{Lemma}}
\addto\captionsenglish{\renewcommand{\propositionname}{Proposition}}
\addto\captionsenglish{\renewcommand{\remarkname}{Remark}}
\addto\captionsenglish{\renewcommand{\theoremname}{Theorem}}
\addto\captionsngerman{\renewcommand{\corollaryname}{Korollar}}
\addto\captionsngerman{\renewcommand{\definitionname}{Definition}}
\addto\captionsngerman{\renewcommand{\examplename}{Beispiel}}
\addto\captionsngerman{\renewcommand{\lemmaname}{Lemma}}
\addto\captionsngerman{\renewcommand{\propositionname}{Satz}}
\addto\captionsngerman{\renewcommand{\remarkname}{Bemerkung}}
\addto\captionsngerman{\renewcommand{\theoremname}{Theorem}}
\providecommand{\corollaryname}{Corollary}
\providecommand{\definitionname}{Definition}
\providecommand{\examplename}{Example}
\providecommand{\lemmaname}{Lemma}
\providecommand{\propositionname}{Proposition}
\providecommand{\remarkname}{Remark}
\providecommand{\theoremname}{Theorem}

\begin{document}
\begin{onehalfspace}
\noindent \begin{center}
\pagenumbering{roman}\thispagestyle{empty}\foreignlanguage{ngerman}{\vspace*{1cm}
}\textsf{\textbf{\Huge{}Connectivity of Boolean Satisfiability}}\textsf{\vspace{3.5cm}
}\foreignlanguage{ngerman}{\textsf{\large{}}}\\
\foreignlanguage{ngerman}{\textsf{\large{}Von der}\textsf{\vspace{0.6cm}
}\textsf{\Large{}}}\\
\foreignlanguage{ngerman}{\textsf{\Large{}Fakultät für Elektrotechnik
und Informatik}\textsf{\vspace{0.3cm}
}\textsf{\Large{}}}\\
\foreignlanguage{ngerman}{\textsf{\Large{}der Gottfried Wilhelm Leibniz
Universität Hannover}\textsf{ \vspace{0.6cm}
}}\\
\foreignlanguage{ngerman}{\textsf{\large{}zur Erlangung des Grades}\textsf{
\vspace{0.6cm}
}\textsf{\Large{}}}\\
\foreignlanguage{ngerman}{\textsf{\Large{}Doktor der Naturwissenschaften}\textsf{\vspace{0.3cm}
}\textsf{\Large{}}}\\
\foreignlanguage{ngerman}{\textsf{\Large{}Dr. rer. nat.}\textsf{ \vspace{0.6cm}
}\textsf{\large{}}}\\
\foreignlanguage{ngerman}{\textsf{\large{}genehmigte Dissertation}\textsf{\vspace{2cm}
}}\\
\foreignlanguage{ngerman}{\textsf{\large{}von}\textsf{\vspace{0.6cm}
}\textsf{\Large{}}}\\
\foreignlanguage{ngerman}{\textsf{\LARGE{}Dipl.-Physiker $\;$Konrad
W. Schwerdtfeger}\textsf{\vspace{0.6cm}
}\textsf{\large{}}}\\
\foreignlanguage{ngerman}{\textsf{\large{}geboren am 6. August 1985
in Hildesheim}}\textsf{\vspace{3.5cm}
}\textsf{\LARGE{}}\\
\textsf{\LARGE{}2015}\newpage{}
\par\end{center}
\end{onehalfspace}

\begin{flushleft}
\thispagestyle{empty}\phantom{}\vspace{19cm}

\par\end{flushleft}

\begin{flushleft}
\begin{tabular}{ll}
Referent: & Heribert Vollmer, Leibniz Universität Hannover\tabularnewline
Korrefferent: & Olaf Beyersdorff, University of Leeds\tabularnewline
Tag der Promotion: & \tabularnewline
\end{tabular}\newpage{}\thispagestyle{empty}\phantom{}\textsf{\vspace{1cm}
}
\par\end{flushleft}

\selectlanguage{ngerman}%
\begin{center}
Für meinen Vater\\
Gerhard Schwerdtfeger\\
1926 - 2014
\par\end{center}

\selectlanguage{english}%
\begin{center}
\vspace{13cm}

\par\end{center}

\selectlanguage{ngerman}%
\begin{center}
Mein herzlicher Dank gilt meinem Doktorvater Heribert Vollmer\\
 für seine Unterstützung bei der Arbeit an dieser Dissertation.
\par\end{center}

\selectlanguage{english}%
\begin{center}
\textsf{\vspace{2cm}
}\phantom{}
\par\end{center}

\begin{center}
\newpage{}
\par\end{center}

\begin{center}
\thispagestyle{empty}\phantom{}
\par\end{center}

\begin{center}
\textsf{\vspace{3.5cm}
}
\par\end{center}

\begin{center}
\emph{The first principle is that you must not fool yourself, }\\
\emph{and you are the easiest person to fool.}
\par\end{center}

\begin{center}
{\small{}\hspace*{38ex}}Richard Feynman
\par\end{center}

\begin{center}
\vspace{10cm}

\par\end{center}

\begin{center}
\newpage{}
\par\end{center}

\section*{\thispagestyle{empty}Zusammenfassung}

\selectlanguage{ngerman}%
In dieser Dissertation befassen wir uns mit der Lösungsraum-Struktur
Boolescher Erfüllbarkeits-Probleme, aus Sicht der theoretischen Informatik,
insbesondere der Komplexitätstheorie.

Wir betrachten den \emph{Lösungs-Graphen} Boolescher Formeln; dieser
Graph hat als Knoten die Lösungen der Formel, und zwei Lösungen sind
verbunden wenn sie sich in der Belegung genau einer Variablen unterscheiden.
Für diesen implizit definierten Graphen untersuchen wir dann das Erreichbarkeitsproblem
und das Zusammenhangsproblem.

Die erste systematische Untersuchung der Lösungs-Graphen Boolescher
Constraint-Satisfaction-Probleme wurde 2006 von Gopalan et al.\ durchgeführt,
motiviert hauptsächlich von Forschung für Erfüllbarkeits-Algorithmen.
Insbesondere untersuchten sie CNF\textsubscript{\selectlanguage{english}%
C\selectlanguage{ngerman}%
}($\mathcal{S}$)-Formeln, d.h. Konjunktionen von Bedingungen, welche
sich aus dem Einsetzen von Variablen und Konstanten in Boolesche Relationen
einer endlichen Menge $\mathcal{S}$ ergeben.

Gopalan et al.\ bewiesen eine Dichotomie für die Komplexität des
Erreichbarkeitsproblems: Entweder ist es in Polynomialzeit lösbar
oder PSPACE-vollständig, Damit übereinstimmend fanden sie auch eine
strukturelle Dichotomie: Der maximale Durchmesser der Zusammenhangskomponenten
ist entweder linear in der Zahl der Variablen, oder er kann exponentiell
sein, Weiterhin vermuteten sie eine Trichotomie für das Zusammenhangsproblem:
entweder ist es in P, coNP-vollständig oder PSPACE-vollständig. Zusammen
mit Makino et al.\ bewiesen sie schon Teile dieser Trichotomie.

Auf diesen Arbeiten aufbauend vervollständigen wir hier den Beweis
der Trichotomie, und korrigieren auch einen kleineren Fehler von Gopalan
et al, was in einer leichten Verschiebung der Grenzen resultiert.

Anschließend untersuchen wir zwei wichtige Varianten: CNF($\mathcal{S}$)-Formeln
ohne Konstanten, und partiell quantifizierte Formeln. In beiden Fällen
beweisen wir für das Erreichbarkeitsproblem und den Durchmesser Dichotomien
analog jener für CNF\textsubscript{\selectlanguage{english}%
C\selectlanguage{ngerman}%
}($\mathcal{S}$)-Formeln. Für das Zusammenhangsproblem zeigen wir
eine Trichotomie im Fall quantifizierter Formeln, während wir im Fall
der Formeln ohne Konstanten Fragmente identifizieren in denen das
Problem in P, coNP-vollständig, und PSPACE-vollständig ist.

Schließlich betrachten wir die Zusammenhangs-Fragen für $B$-Formeln,
d.h. geschachtelte Formeln, aufgebaut aus Junktoren einer endlichen
Menge $B$, und für $B$-Circuits, d.h. Boolesche Schaltkreise, aufgebaut
aus Gattern einer festen Menge $B$. Hier nutzen wir Emil Post's Klassifikation
aller geschlossener Klassen Boolescher Funktionen. Wir beweisen eine
gemeinsame Dichotomie für das Erreichbarkeitsproblem, das Zusammenhangsproblem
und den Durchmesser: Auf der einen Seite sind beide Probleme in P
und der Durchmesser ist linear, während auf der anderen Seite die
Probleme PSPACE-vollständig sind und der Durchmesser exponentiell
sein kann. Für partiell quantifizierte $B$-Formeln zeigen wir eine
analoge Dichotomie.\\

\paragraph*{Schlagworte}

Komplexität \foreignlanguage{english}{$\cdot$} Erfüllbarkeit \foreignlanguage{english}{$\cdot$
}Zusammenhang in Graphen \foreignlanguage{english}{$\cdot$ }Boolesche
CSPs \foreignlanguage{english}{$\cdot$} Boolesche Schaltkreise \foreignlanguage{english}{$\cdot$}
Post'scher Verband \foreignlanguage{english}{$\cdot$} Dichotomien 

\newpage{}

\selectlanguage{english}%

\section*{Abstract}

In this thesis we are concerned with the solution-space structure
of Boolean satisfiability problems, from the view of theoretical computer
science, especially complexity theory.

We consider the \emph{solution graph} of Boolean formulas; this is
the graph where the vertices are the solutions of the formula, and
two solutions are connected iff they differ in the assignment of exactly
one variable. For this implicitly defined graph, we then study the
$st$-connectivity and connectivity problems.

The first systematic study of the solution graphs of Boolean constraint
satisfaction problems was done in 2006 by Gopalan et al., motivated
mainly by research for satisfiability algorithms. In particular, they
considered CNF\textsubscript{C}($\mathcal{S}$)-formulas, which are
conjunctions of constraints that arise from inserting variables and
constants in relations of some finite set $\mathcal{S}$.

Gopalan et al.\ proved a computational dichotomy for the $st$-connectivity
problem, asserting that it is either solvable in polynomial time or
PSPACE-complete, and an aligned structural dichotomy, asserting that
the maximal diameter of connected components is either linear in the
number of variables, or can be exponential. Further, they conjectured
a trichotomy for the connectivity problem: That it is either in P,
coNP-complete, or PSPACE-complete. Together with Makino et al., they
already proved parts of this trichotomy.

Building on this work, we here complete the proof of the trichotomy,
and also correct a minor mistake of Gopalan et al., which leads to
slight shifts of the boundaries.

We then investigate two important variants: CNF($\mathcal{S}$)-formulas
without constants, and partially quantified formulas. In both cases,
we prove dichotomies for $st$-connectivity and the diameter analogous
to the ones for CNF\textsubscript{C}($\mathcal{S}$)-formulas. For
for the connectivity problem, we show a trichotomy in the case of
quantified formulas, while in the case of formulas without constants,
we identify fragments where the problem is in P, where it is coNP-complete,
and where it is PSPACE-complete. 

Finally, we consider the connectivity issues for $B$-formulas, which
are arbitrarily nested formulas built from some fixed set $B$ of
connectives, and for $B$-circuits, which are Boolean circuits where
the gates are from some finite set $B$. Here, we make use of Emil
Post's classification of all closed classes of Boolean functions.
We prove a common dichotomy for both connectivity problems and the
diameter: on one side, both problems are in P and the diameter is
linear, while on the other, the problems are PSPACE-complete and the
diameter can be exponential. For partially quantified $B$-formulas,
we show an analogous dichotomy.\\

\paragraph*{Keywords}

Computational complexity$\cdot$ Boolean satisfiability $\cdot$ Graph
connectivity $\cdot$ Boolean CSPs $\cdot$ Boolean circuits $\cdot$
Post's lattice $\cdot$ Dichotomy theorems

\let\lstd\clearpage \let\lstdd\cleardoublepage \let\clearpage\relax \tableofcontents{}\textsf{\vspace{1.5cm}
}\let\cleardoublepage\relax

\listoffigures
\foreignlanguage{ngerman}{\textsf{\vspace{0.6cm}
}}

\listoftables

\mainmatter\let\clearpage\lstd \let\cleardoublepage\lstdd

\chapter{Introduction}

\section{Boolean Satisfiability and~Solution~Space~Connectivity}

The Boolean satisfiability problem (SAT) asks for a propositional
formula if there is an assignment to the variables such that it evaluates
to true. It is of great importance in many areas of theoretical and
applied computer science: In complexity theory, it was one of the
first problems proven to be NP-complete, and still is the most important
standard problem for reductions. In propositional logic, many important
reasoning problems can be reduced to SAT, e.g. checking entailment:
For any two sentences $\alpha$ and $\beta$, $\alpha\models\beta$
if and only if $\alpha\wedge\overline{\beta}$ is unsatisfiable. These
connections are used for example in artificial intelligence for reasoning,
planning, and automated theorem proving, and in electronic design
automation (EDA) for formal equivalence checking. 

SAT is only the most basic version of a multitude of related problems,
asking questions about a relation given by some short description.
In one direction, one may look at constraint satisfaction problems
over higher domains, or at multi-valued logics. In another direction,
one may consider other tasks like enumerating all solutions, counting
the solutions, checking the equivalence of formulas or circuits, or
finding the optimal solution according to some measure. In this thesis,
we follow the second direction and focus on the solution-space structure:
For a formula $\phi$, we consider the \emph{solution graph} $G(\phi)$,
where the vertices are the solutions, and two solutions are connected
iff they differ in the assignment of exactly one variable. For this
implicitly defined graph, we then study the connectivity and $st$-connectivity
problems.

Since any propositional formula over $n$ variables defines an $n$-ary
Boolean relation $R$, i.e.\ a subset of $\{0,1\}^{n}$, another
way to think of the solution graph is the subgraph of the $n$-dimensional
hypercube graph induced by the vectors in $R$. The figures below
give an impression of how solution graphs may look like.

\begin{figure}[!h]
\begin{centering}
\includegraphics[scale=0.8]{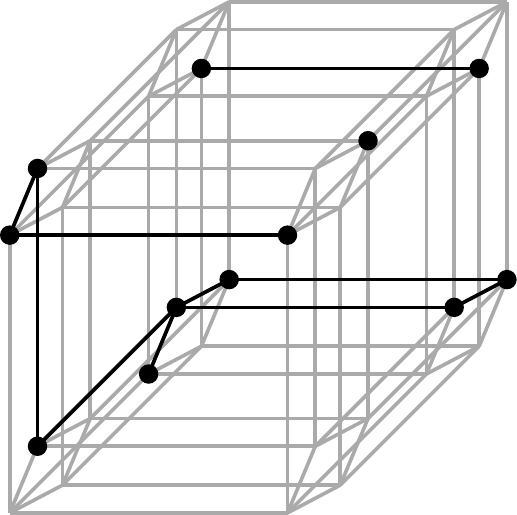}$\qquad$\includegraphics[scale=0.41]{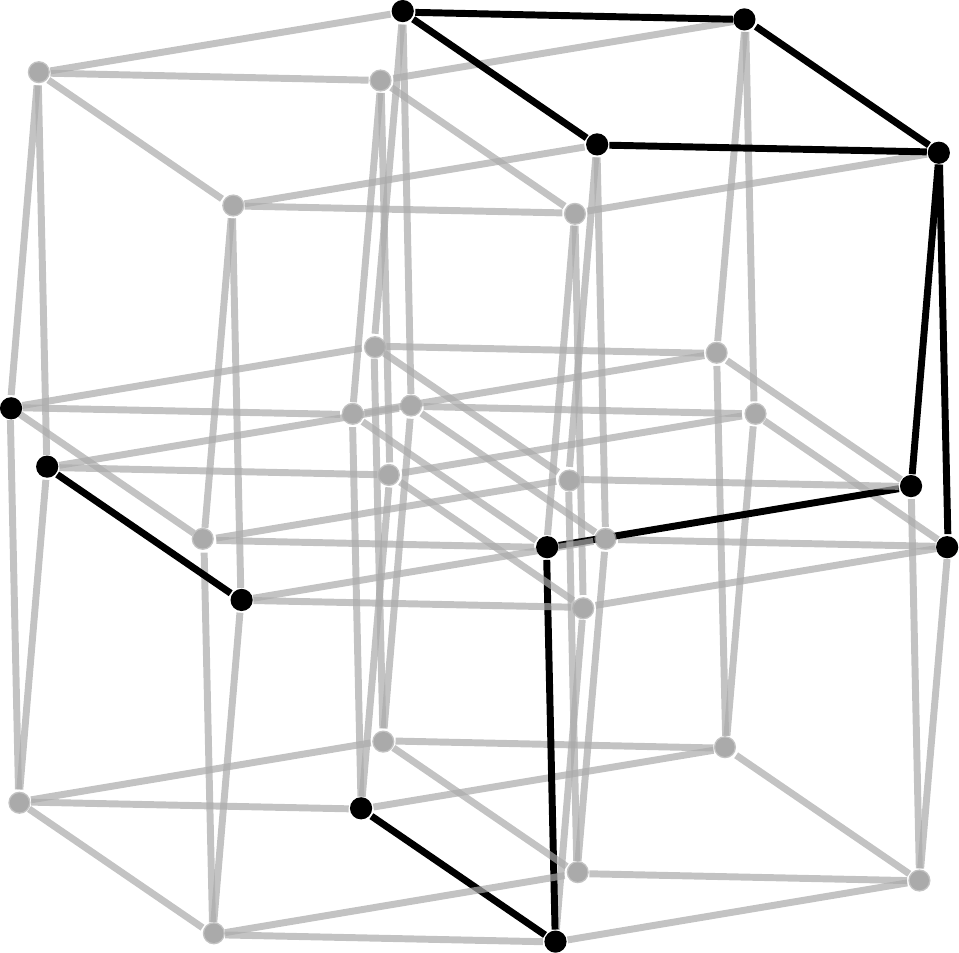}$\qquad$\includegraphics[scale=0.5]{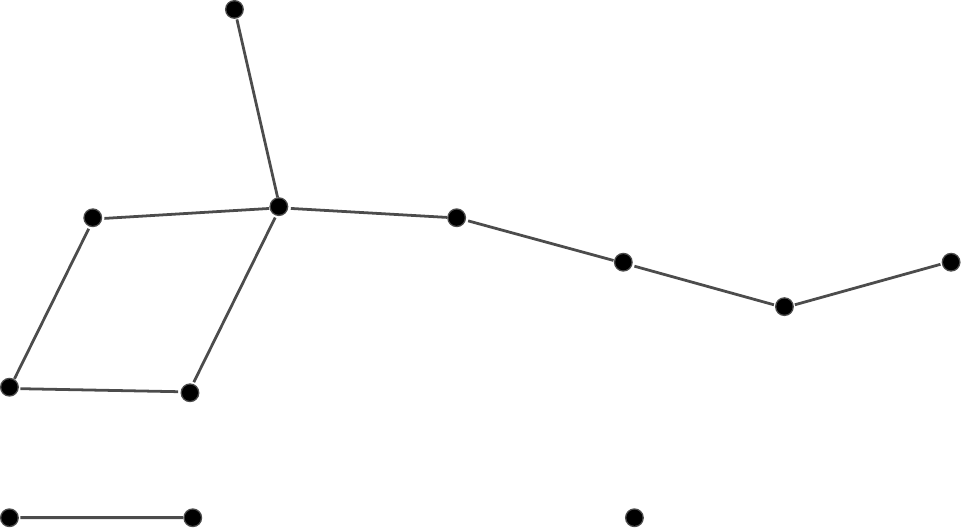}
\par\end{centering}

\protect\caption[Depictions of the subgraph of the 5-dimensional hypercube graph induced
by a typical random Boolean relation with 12 elements.]{\emph{}Depictions of the subgraph of the 5-dimensional hypercube
graph induced by a typical random Boolean relation with 12 elements.
Left: highlighted on an orthographic hypercube projection by our \noun{SatConn}-tool.
Center: highlighted on a ``Spectral Embedding'' of the hypercube
graph by \noun{Mathematica}. Right: the sole subgraph, arranged by
\noun{Mathematica.}}
\end{figure}

\begin{figure}[!h]
\begin{centering}
\includegraphics[scale=0.42]{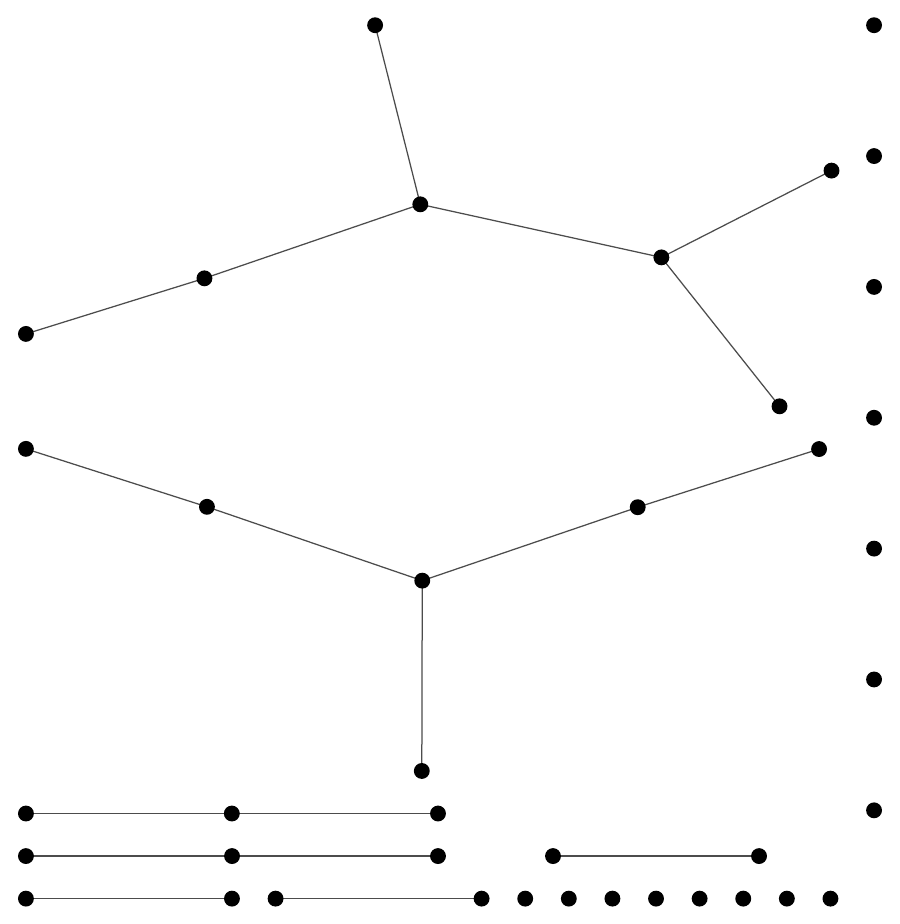}$\qquad$\includegraphics[scale=0.48]{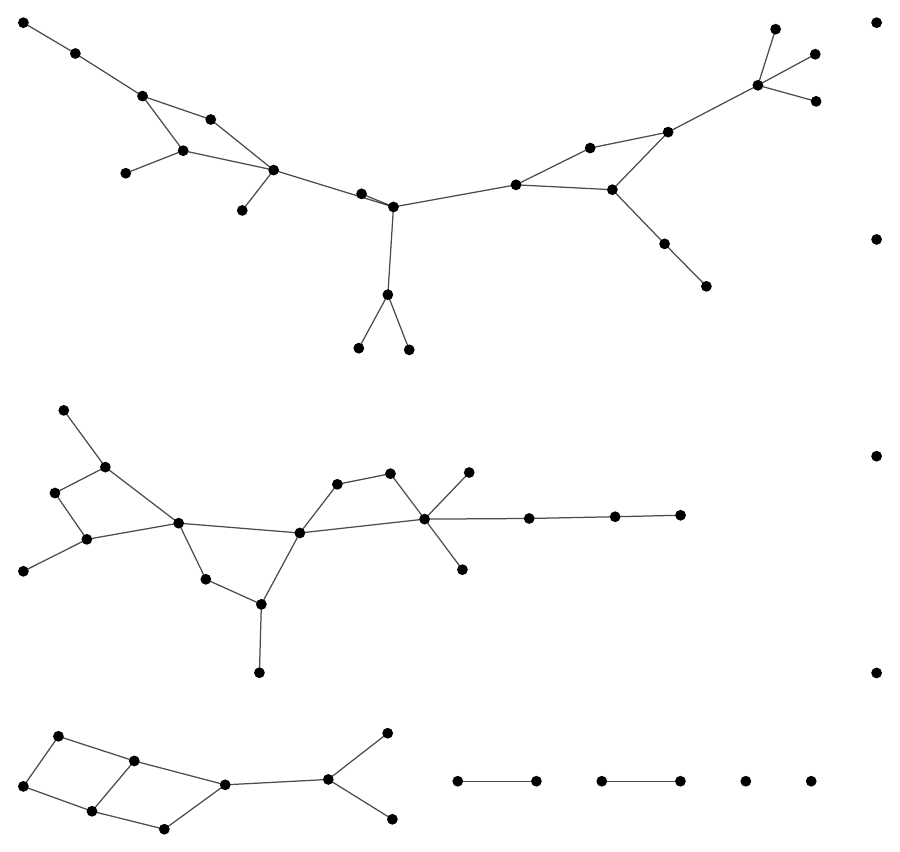}$\qquad$\includegraphics[scale=0.6]{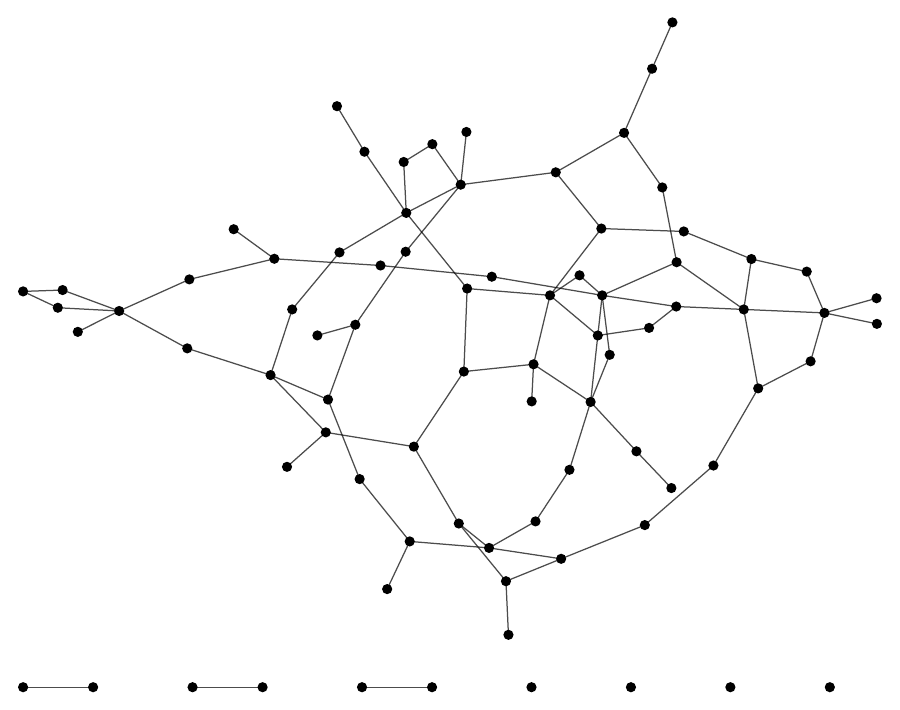}
\par\end{centering}

\protect\caption[Subgraphs of the 8-dimensional hypercube graph induced by typical
random relations]{\emph{}Subgraphs of the 8-dimensional hypercube graph (with 256 vertices)
induced by typical random relations with 40, 60 and 80 elements, arranged
by \noun{Mathematica.}}
\end{figure}

Our perspective is mainly from complexity theory: As it was done for
SAT and many of the related problems, we classify restrictions of
the connectivity problems by their worst-case complexity. Along the
way, we will also examine structural properties of the solution graph,
and devise efficient algorithms for solving the connectivity problems.

Besides the usual propositional formulas with the connectives $\wedge$,
$\vee$ and $\neg$, there are many alternative representations of
Boolean relations; we will consider the following:
\begin{itemize}
\item \emph{Boolean constraint satisfaction problems} (\emph{Boolean CSP}s,
here \emph{CSPs} for short), specifically

\begin{itemize}
\item \emph{CNF\textsubscript{C}($\mathcal{S}$)-formulas}, i.e.\ conjunctions
of constraints that arise from inserting variables and constants in
relations of some finite set $\mathcal{S}$,
\item \emph{CNF($\mathcal{S}$)-formulas}, where no constants may be used,
\end{itemize}
\item \emph{$B$-formulas}, i.e.\ arbitrarily nested formulas built from
some finite set $B$ of connectives,
\item \emph{$B$-circuits}, i.e.\ Boolean circuits where the gates are
from some finite set $B$.
\end{itemize}
For CNF\textsubscript{C}($\mathcal{S}$)-formulas and $B$-formulas,
we also consider versions with quantifiers.

\section{Relevance of Solution Space Connectivity}

A direct application of $st$-connectivity in solution graphs are
\emph{reconfiguration problems}, that arise when we wish to find a
step-by-step transformation between two feasible solutions of a problem,
such that all intermediate results are also feasible. Recently, the
reconfiguration versions of many problems such as \noun{Independent-Set},
\noun{Vertex-Cover}, \noun{Set-Cover} \noun{Graph-$k$-Coloring},
\noun{Shortest-Path} have been studied (see e.g. \citep{recon,short}).

The connectivity properties of solution graphs are also of relevance
to the problem of \emph{structure identification}, where one is given
a relation explicitly and seeks a short representation of some kind
(see e.g. \citep{creignou2008structure}); this problem is important
especially for learning in artificial intelligence.

Further, a better understanding of the solution space structure promises
advancement of SAT algorithms: It has been discovered that the solution
space connectivity is strongly correlated to the performance of standard
satisfiability algorithms like WalkSAT and DPLL on random instances:
As one approaches the\emph{ satisfiability threshold} (the ratio of
constraints to variables at which random $k$-CNF-formulas become
unsatisfiable for $k\geq3$) from below, the solution space (with
the connectivity defined as above) fractures, and the performance
of the algorithms deteriorates \citep{mezard2005clustering,maneva2007new}.
These insights mainly came from statistical physics, and lead to the
development of the \emph{survey propagation algorithm}, which has
superior performance on random instances \citep{maneva2007new}.

While current SAT solvers normally accept only CNF-formulas as input,
in EDA the instances mostly derive from digital circuit descriptions
\citep{csat2}, and although many such instances can easily be encoded
in CNF, the original structural information, such as signal ordering,
gate orientation and logic paths, is lost, or at least obscured. Since
exactly this information can be very helpful for solving these instances,
considerable effort has been made recently to develop satisfiability
solvers that work with the circuit description directly \citep{csat2},
which have far better performance in EDA applications, or to restore
the circuit structure from CNF \citep{extracting}. This is a reason
for us to study the solution space also for Boolean circuits.

\section{Related Work, Prior Publications, and this Thesis}

Research has focused on the solution space structure only quite recently:
Complexity results for the connectivity problems in the solution graphs
of CSPs have first been obtained in 2006 by P.~Gopalan, P.~G.~Kolaitis,
E.~Maneva, and C.~H.~Papadimitriou \citep{Gopalan:2006,gop}. In
particular, they investigated CNF\textsubscript{C}($\mathcal{S}$)-formulas
and studied
\begin{itemize}
\item the $st$-connectivity problem \noun{st-Conn}\textsubscript{C}($\mathcal{S}$),
that asks for a CNF\textsubscript{C}($\mathcal{S}$)-formula $\phi$
and two solutions $\boldsymbol{s}$ and $\boldsymbol{t}$ whether
there a path from $\boldsymbol{s}$ to $\boldsymbol{t}$ in $G(\phi)$,
\item the connectivity problem \noun{Conn}\textsubscript{C}($\mathcal{S}$),
that asks for a CNF\textsubscript{C}($\mathcal{S}$)-formula $\phi$
whether $G(\phi)$ is connected,
\end{itemize}
and
\begin{itemize}
\item the maximal diameter of any connected component of $G(\phi)$ for
a CNF\textsubscript{C}($\mathcal{S}$)-formula $\phi$, where the
diameter of a component is the maximal shortest-path distance between
any two vectors in that component.
\end{itemize}
They found a common structural and computational dichotomy: On one
side, the maximal diameter is linear in the number of variables, $st$-connectivity
is in P and connectivity is in coNP, while on the other side, the
diameter can be exponential, and both problems are PSPACE-complete.
Moreover, they conjectured a trichotomy for connectivity: That it
is in P, coNP-complete, or PSPACE-complete. Together with Makino et
al.\ \citep{Makino:2007:BCP:1768142.1768162}, they already proved
parts of this trichotomy.

In \citep{csp}, we completed the proof of the trichotomy, and also
corrected minor mistakes in \citep{gop}, which lead to a slight shift
of the boundaries towards the hard side. So for CNF\textsubscript{C}($\mathcal{S}$)-formulas,
we now have a quite complete picture, which we present in \prettyref{chap:2}.
In \citep{csp}, we explained in detail the mistakes of Gopalan et
al.\ and their implications, here we just give the correct statement
and proofs.

In \prettyref{chap:3}, we investigate two important variants: CNF($\mathcal{S}$)-formulas
without constants, and partially quantified CNF\textsubscript{C}($\mathcal{S}$)-formulas.
In both cases, we prove a dichotomy for $st$-connectivity and the
diameter analogous to the one for CNF\textsubscript{C}($\mathcal{S}$)-formulas.
For for the connectivity problem, we prove a trichotomy in the case
of quantified formulas, while in the case of formulas without constants,
we have no complete classification, but identify fragments where the
problem is in P, where it is coNP-complete, and where it is PSPACE-complete.
Of this chapter, only a preprint with preliminary results appeared
on ArXiv \citep{schwerdtfeger2014connectivity}.

Finally, in \prettyref{chap:4}, we look at $B$-formulas and $B$-circuits.
Here, we find a common dichotomy for the diameter and both connectivity
problems: on one side, the diameter is linear and both problems are
in P, while on the other, the diameter can be exponential, and the
problems are PSPACE-complete. For quantified $B$-formulas, we prove
an analogous dichotomy. The work in this chapter has been published
in \citep{csr}.

\section{Associated Software}

As part of the research for this thesis, several programs were written,
some of which may be useful for future work on related problems. All
software is written in Java (version 8) and provided in the \noun{SatConn
}package at \href{https://github.com/konradws/SatConn}{https://github.com/konradws/SatConn},
including a graphical tool to draw the solution graphs on hypercube
projections, used for several graphics in this thesis.

After downloading the complete repository, the folder can be opened
resp. imported in \noun{Netbeans} or \noun{Eclipse} as a project.
The graphical tool is also provided as executable (\textsf{SatConnTool.jar}).

The most useful functions are declared \texttt{public} and equipped
with Javadoc comments, where helpful. The \texttt{main}-functions
provide usage examples and can be executed by running the respective
file.

\section{General Preliminaries}

\paragraph*{Prerequisites }

We assume familiarity with some basic concepts from theoretical computer
science, especially complexity theory, and its mathematical foundations:
\begin{itemize}
\item From mathematics, we require propositional logic, and basics about
graphs, hypergraphs, and lattices,
\item From theoretical computer science, we require Turing machines, the
common complexity classes P, NP, coNP, PSPACE, and polynomial-time
reductions.
\end{itemize}

\paragraph{Notation}

We use $\boldsymbol{a},\boldsymbol{b},\ldots$ or $\boldsymbol{a}^{1},\boldsymbol{a}^{2},\ldots$
to denote vectors of Boolean values and $\boldsymbol{x},\boldsymbol{y},\ldots$
or $\boldsymbol{x}^{1},\boldsymbol{x}^{2},\ldots$ to denote vectors
of variables, $\boldsymbol{a}=(a_{1},a_{2},\ldots)$ and $\boldsymbol{x}=(x_{1},x_{2},\ldots)$.

$\phi[\boldsymbol{x}^{i}/\boldsymbol{a}]$ denotes the formula resulting
from $\phi$ by substituting the constants $a_{j}$ for the variables
$x_{j}^{i}$.

The symbol $\leq_{m}^{p}$ is used for polynomial-time many-one reductions.

\paragraph*{Central concepts}

In the following definition, we formally introduce some concepts related
to solution space connectivity in general. At the beginning of the
next chapter, we define notions specific to CSPs. A reader only interested
in \emph{$B$-}formulas and \emph{$B$-}circuits may read \prettyref{sec:ge}
after the next definition, and then skip to \prettyref{chap:4}.
\begin{defn}
An $n$-ary\emph{ Boolean relation} (or\emph{ logical relation, relation}
for short) is a subset of $\{0,1\}^{n}$ for some integer $n\ge1$.

The set of solutions of a propositional formula $\phi$ over $n$
variables defines in a natural way an $n$-ary relation $[\phi]$,
where the variables are taken in lexicographic order. We will often
identify the formula $\phi$ with the relation it defines and omit
the brackets.

The \emph{solution graph} $G(\phi)$ of $\phi$ then is the subgraph
of the $n$-dimensional hypercube graph induced by the vectors in
$[\phi]$. We will also refer to $G(R)$ for any logical relation
$R$ (not necessarily defined by a formula).

The \emph{Hamming weight }$|\boldsymbol{a}|$ of a Boolean vector
\textbf{$\boldsymbol{a}$} is the number of 1's in \textbf{$\boldsymbol{a}$}.
For two vectors \textbf{$\boldsymbol{a}$} and $\boldsymbol{b}$,
the \emph{Hamming distance }$|\boldsymbol{a}-\boldsymbol{b}|$ is
is the number of positions in which they differ\emph{.}

If \textbf{$\boldsymbol{a}$} and $\boldsymbol{b}$ are solutions
of $\phi$ and lie in the same connected component (\emph{component}
for short) of $G(\phi)$, we write $d_{\phi}(\boldsymbol{a},\boldsymbol{b})$
to denote the shortest-path distance between \textbf{$\boldsymbol{a}$}
and $\boldsymbol{b}$. The \emph{diameter} \emph{of a component} is
the maximal shortest-path distance between any two vectors in that
component. The \emph{diameter of }$G(\phi)$\emph{ }is the maximal
diameter of any of its connected components.
\end{defn}

\chapter{\label{chap:2}Connectivity of Constraints}

We start our investigation with\emph{ }constraint satisfaction problems.
A constraint is a tuple of variables together with a Boolean relation,
restricting the assignment of the variables. A CSP then is the question
whether there is an assignment to all variables of a set of constraints
such that all constraints are satisfied.

\section{Preliminaries}

\subsection{\label{sub:CSPs}CNF-Formulas and Schaefer's Framework}

In line with Gopalan et al., we define CSPs by CNF($\mathcal{S}$)-formulas,\emph{
}which were introduced in 1978 by Thomas Schaefer as a generalization
of CNF (conjunctive normal form) formulas \citep{Schaefer:1978:CSP:800133.804350}.
\begin{defn}
\label{def:cnf}A \emph{CNF-formula} is a propositional formula of
the form $C_{1}\wedge\cdots\wedge C_{m}$ ($1\leq m<\infty$), where
each $C_{i}$ is a \emph{clause}, that is, a finite disjunction of
\emph{literals} (variables or negated variables). A\emph{ $k$-CNF-formula}
($k\geq1$) is a CNF-formula where each $C_{i}$ has at most $k$
literals. A \emph{Horn (dual Horn)} formula is a CNF-formula where
each $C_{i}$ has at most one positive (negative) literal.
\end{defn}

\begin{defn}
\label{def:cnf-1}For a finite set of relations $\mathcal{S}$, a\emph{
CNF\textsubscript{C}($\mathcal{S}$)-formula} over a set of variables
$V$ is a finite conjunction $C_{1}\wedge\cdots\wedge C_{m}$, where
each $C_{i}$ is a \emph{constraint application} (\emph{constraint
}for short), i.e., an expression of the form $R(\xi_{1},\ldots,\xi_{k})$,
with a $k$-ary relation $R\in\mathcal{S}$, and each $\xi_{j}$ is
a variable from $V$ or one of the constants \{0, 1\}. A\emph{ CNF($\mathcal{S}$)-formula
}is a\emph{ }CNF\textsubscript{C}($\mathcal{S}$)-formula where each
$\xi_{j}$ is a variable in $V$, not a constant.

By $\mathrm{Var}(C_{i})$, we denote the set of variables occurring
in $\xi_{1},\ldots,\xi_{k}$. With the\emph{ relation corresponding
to $C_{i}$} we mean the relation $[R(\xi_{1},\ldots,\xi_{k})]$ (that
may be different from $R$ by substitution of constants, and identification
and permutation of variables).

A \emph{$k$-clause} is a disjunction of $k$ variables or negated
variables. For $0\leq i\leq k$, let $D_{i}$ be the corresponding
to the $k$-clause whose first $i$ literals are negated, and let
$S_{k}=\{D_{0},\ldots,D_{k}\}$, e.g., $\mathcal{S}_{3}=\{[x\vee y\vee z],\:[\overline{x}\vee y\vee z],\:[\overline{x}\vee\overline{y}\vee z],\:[\overline{x}\vee\overline{y}\vee\overline{z}]\}$.
Thus, CNF($S_{k}$) is the collection of $k$-CNF-formulas.
\end{defn}
Thomas Schaefer introduced CNF($\mathcal{S}$)-formulas for expressing
variants of Boolean satisfiability; in his dichotomy theorem, Schaefer
then classified the complexity of the satisfiability problem for CNF\textsubscript{C}($\mathcal{S}$)-\emph{
}and\emph{ }CNF($\mathcal{S}$)-formulas \citep{Schaefer:1978:CSP:800133.804350};
we will do so here for the connectivity problems. We use the following
notation:
\begin{itemize}
\item \emph{\noun{Sat}}($\mathcal{S}$) for the \emph{satisfiability problem:}
Given a CNF($\mathcal{S}$)-formula $\phi$, is $\phi$ satisfiable?
\item \emph{\noun{st-}}\noun{Conn}($\mathcal{S}$) for the\emph{ $st$-connectivity
problem}\emph{\noun{:}}\noun{ }Given a CNF($\mathcal{S}$)-formula
$\phi$ and two solutions $\boldsymbol{s}$ and $\boldsymbol{t}$,
is there a path from $\boldsymbol{s}$ to $\boldsymbol{t}$ in $G(\phi)$?
\item \noun{Conn}($\mathcal{S}$) for the\emph{ connectivity problem:} Given
a CNF($\mathcal{S}$)-formula $\phi$, is $G(\phi)$ connected? (if
$\phi$ is unsatisfiable, we consider $G(\phi)$ connected)
\end{itemize}
The respective problems for CNF\textsubscript{C}($\mathcal{S}$)-formulas
are marked with the subscript \textsubscript{C}. Note that Gopalan
et al.\ considered the case with constants, but omitted the \textsubscript{C}.

\subsection{Classes of Relations }

In the following definition, we introduce the types of relations needed
for the classifications. Some are already familiar from Schaefer's
dichotomy theorem\emph{\noun{,}} some were introduced by Gopalan et
al., and the ones starting with ``safely'' we defined in \citep{csp}
to account for the shift of the boundaries resulting from Gopalan
et al.'s mistake; \emph{IHSB }stands for\emph{ }``implicative hitting
set-bounded'' and was introduced in \citep{Creignou:2001:CCB:377810}.
\begin{defn}
\label{def:typ}Let $R$ be an $n$-ary logical relation.
\begin{itemize}
\item $R$ is\emph{ 0-valid (1-valid)} if $0^{n}\in R$ ($1^{n}\in R$).
\item $R$ is \emph{complementive} if for every vector $(a_{1},\ldots,a_{n})\in R$,
also $(a_{1}\oplus1,\ldots,a_{n}\oplus1)\in R$.
\item $R$ is \emph{bijunctive} if it is the set of solutions of a 2-CNF-formula.
\item $R$ is \emph{Horn (dual Horn)} if it is the set of solutions of a
Horn (dual Horn) formula.
\item $R$ is \emph{affine} if it is the set of solutions of a formula $x_{i_{1}}\oplus\ldots\oplus x_{i_{m}}\oplus c$
with $i_{1},\ldots,i_{m}\in\{1,\ldots,n\}$ and $c\in\{0,1\}$.
\item $R$ is \emph{componentwise bijunctive} if every connected component
of $G(R)$ is a bijunctive relation. $R$ is \emph{safely componentwise
bijunctive} if $R$ and every relation $R'$ obtained from $R$ by
identification of variables is componentwise bijunctive.
\item $R$ is \emph{OR-free} (\emph{NAND-free}) if the relation OR = $\left\{ 01,10,11\right\} $
(NAND = $\left\{ 00,01,10\right\} $) cannot be obtained from $R$
by substitution of constants. $R$ is \emph{safely OR-free} (\emph{safely
NAND-free}) if $R$ and every relation $R'$ obtained from $R$ by
identification of variables is OR-free (NAND-free).
\item $R$ is \emph{IHSB$-$ (IHSB$+$)} if it is the set of solutions of
a Horn (dual Horn) formula in which all clauses with more than 2 literals
have only negative literals (only positive literals).
\item $R$ is \emph{componentwise IHSB$-$ (componentwise IHSB$+$)} if
every connected component of $G(R)$ is\emph{ }IHSB$-$ (IHSB$+$).
$R$ is \emph{safely} \emph{componentwise IHSB$-$ (safely componentwise
IHSB$+$) }if $R$ and every relation $R'$ obtained from $R$ by
identification of variables is componentwise\emph{ }IHSB$-$ (componentwise
IHSB$+$).
\end{itemize}
\end{defn}
If one is given the relation explicitly (as a set of vectors), the
properties 0-valid, 1-valid, complementive, OR-free\emph{ }and\emph{
}NAND-free can be checked straightforward, while bijunctive, Horn,
dual Horn, affine, IHSB$-$ and IHSB$+$ can be checked by \emph{closure
properties}:
\begin{defn}
A relation $R$ is \emph{closed} under some $n$-ary operation $f$
iff the vector obtained by the coordinate-wise application of $f$
to any $m$ vectors from $R$ is again in $R$, i.e., if 
\[
\boldsymbol{a}^{1},\ldots,\boldsymbol{a}^{m}\in R\Longrightarrow\left(f(a_{1}^{1},\ldots,a_{1}^{m}),\ldots,f(a_{n}^{1},\ldots,a_{n}^{m})\right)\in R.
\]
\pagebreak{}\end{defn}
\begin{lem}
\label{lem:clos}A relation $R$ is
\begin{enumerate}
\item bijunctive, iff it is closed under the ternary majority operation\\
 maj($x,y,z$)=$\left(x\vee y\right)\wedge\left(y\vee z\right)\wedge\left(z\vee x\right)$,
\item Horn (dual Horn), iff it is closed under $\wedge$ (under $\vee$,
resp.),
\item affine,\emph{ }iff it is closed under $x\oplus y\oplus z$,
\item IHSB$-$ (IHSB$+$), iff it is closed under $x\wedge(y\vee z)$ (under
$x\vee(y\wedge z)$, resp.).
\end{enumerate}
\end{lem}
\begin{proof}
1. See \citep[Lemma 4.9]{Creignou:2001:CCB:377810}.

2. See \citep[Lemma 4.8]{Creignou:2001:CCB:377810}.

3. See \citep[Lemma 4.10]{Creignou:2001:CCB:377810}.

4. This can be verified using the Galois correspondence between closed
sets of relations and closed sets of Boolean functions (see \citep{bohler2005bases}):
From the table (Fig. 1) in \citep{bohler2005bases}, we find that
the IHSB$-$ relations are a base of the co-clone INV($\mathsf{S}_{10}$),
and the IHSB$+$ ones a base of INV($\mathsf{S}_{00}$), and from
the table (Figure 1) in \citep{bloc}, we see that $x\wedge(y\vee z)$
and $x\vee(y\wedge z)$ are bases of the clones $\mathsf{S}_{10}$
and $\mathsf{S}_{00}$, resp.\end{proof}
\begin{rem}
The class \texttt{Check} of \noun{SatConn} provides functions to check
the properties of \prettyref{def:typ}, and the class \texttt{Clones}
provides functions to calculate the clone and co-clone of a relation.
\end{rem}
The closure properties carry over from a relation to its connected
components, as shown by Gopalan et al.:
\begin{lem}
\label{lem:l41}\emph{\citep[Lemma 4.1]{gop}} ~If a logical relation
$R$ is closed under an operation $\alpha:\{0,1\}^{k}\rightarrow\{0,1\}$
such that $\alpha(1,\ldots,1)=1$ and $\alpha(0,\ldots,0)=0$ (a.k.a.
an idempotent operation), then every connected component of $G(R)$
is closed under $\alpha$.
\end{lem}

\subsection{Classes of Sets of Relations}

The classes in the following definition demarcate the structural and
computational boundaries for the solution graphs of CNF\textsubscript{C}($\mathcal{S}$)-formulas.
\begin{defn}
\label{def:cpss}A set $\mathcal{S}$ of logical relations is \emph{safely
tight} if at least one of the following conditions holds:
\begin{enumerate}
\item Every relation in $\mathcal{S}$ is safely componentwise bijunctive.
\item Every relation in $\mathcal{S}$ is safely OR-free .
\item Every relation in $\mathcal{S}$ is safely NAND-free.
\end{enumerate}
A set $\mathcal{S}$ of logical relations is \emph{Schaefer} if at
least one of the following conditions holds:
\begin{enumerate}
\item Every relation in $\mathcal{S}$ is bijunctive.
\item Every relation in $\mathcal{S}$ is Horn.
\item Every relation in $\mathcal{S}$ is dual Horn.
\item Every relation in $\mathcal{S}$ is affine.\pagebreak{}
\end{enumerate}
A set $\mathcal{S}$ of logical relations is \emph{CPSS} if at least
one of the following conditions holds:
\begin{enumerate}
\item Every relation in $\mathcal{S}$ is bijunctive.
\item Every relation in $\mathcal{S}$ is Horn and safely componentwise
IHSB$-$.
\item Every relation in $\mathcal{S}$ is dual Horn and safely componentwise
IHSB$+$.
\item Every relation in $\mathcal{S}$ is affine.
\end{enumerate}
A single logical relation $R$ is safely tight, Schaefer, or CPSS,
if $\{R\}$ has that property. Vice versa, we say that a set $\mathcal{S}$
of logical relations has one of the properties from \prettyref{def:typ}
if every relation in $\mathcal{S}$ has that property, e.g., $\mathcal{S}$
is 0-valid if every relation in $\mathcal{S}$ is 0-valid.
\end{defn}
The term \emph{tight} was introduced by Gopalan et al.\ because of
the structural properties of the formulas built from tight (actually,
only safely tight) relations, see \prettyref{lem:l43} and \prettyref{lem:l45}.
We introduced the CPSS class in \citep{csp}; CPSS stands for ``constraint-projection
separating Schaefer'', which will become clear in \prettyref{sec:cpss}
from \prettyref{def:dcp}, \prettyref{lem:prjs} and \prettyref{lem:sc}.

From the definition we see that every CPSS set of relations is also
Schaefer, and we can show that it also holds that every Schaefer set
is safely tight, by modifying a lemma of Gopalan et al.:
\begin{lem}
\label{lem:l42}\emph{\citep[modified from][Lemma 4.2]{gop} }~Let
$R$ be a logical relation.
\begin{enumerate}
\item If $R$ is bijunctive, then it is safely componentwise bijunctive.
\item If $R$ is Horn, then it is safely OR-free.
\item If $R$ is dual Horn, then it is safely NAND-free.
\item If $R$ is affine, then it is safely componentwise bijunctive, safely
OR-free, and safely NAND-free.
\end{enumerate}
\end{lem}
\begin{proof}
We first note that
\begin{itemize}[label= ({*})]
\item  any relation obtained from a bijunctive (Horn, dual Horn, affine)
one by identification of variables is itself bijunctive (Horn, dual
Horn, affine),
\end{itemize}
which is obvious from the definitions.

If $R$ is bijunctive, it is closed under maj, which is idempotent,
so by Lemma \ref{lem:l41}, $R$ is also componentwise bijunctive,
and by ({*}), it is safely componentwise bijunctive as well.

The cases of Horn and dual Horn are symmetric. Suppose a $r$-ary
Horn relation $R$ is not OR-free. Then there exist $i,j\in\{1,\dots,r\}$
and constants $t_{1},\dots,t_{r}\in\{0,1\}$ such that the relation
$R(t_{1},\dots,t_{i-1},x,t_{i+1},\dots,t_{j-1},y,t_{j+1},\dots,t_{r})$
on variables $x$ and $y$ is equivalent to $x\vee y$, i.e. 
\[
R(t_{1},\dots,t_{i-1},x,t_{i+1},\dots,t_{j-1},y,t_{j+1},\dots,t_{r})=\{01,11,10\}.
\]
Thus the tuples $\boldsymbol{t^{00}},\boldsymbol{t^{01}}\boldsymbol{t^{10}},\boldsymbol{t^{11}}$
defined by $(t_{i}^{ab},t_{j}^{ab})=(a,b)$ and $t_{k}^{ab}=t_{k}$
for every $k\not\in\{i,j\}$, where $a,b,\in\{0,1\}$ satisfy $\boldsymbol{t^{10}},\boldsymbol{t^{11}},\boldsymbol{t^{01}}\in R$
and $\boldsymbol{t^{00}}\not\in R$. However, since every Horn relation
is closed under $\wedge$, it follows that $\boldsymbol{t^{01}}\wedge\boldsymbol{t^{10}}=\boldsymbol{t^{00}}$
must be in $R$, which is a contradiction. So $R$ is OR-free, and
again by ({*}), it must be safely OR-free as well.

For the affine case, a small modification of the last step of the
above argument shows that an affine relation also is OR-free; therefore,
dually, it is also NAND-free. Namely, since a relation $R$ is affine
if and only if it is closed under ternary $\oplus$, it follows that
$\boldsymbol{t^{01}}\oplus\boldsymbol{t^{11}}\oplus\boldsymbol{t^{10}}=\boldsymbol{t^{00}}$
must be in $R$. Since the connected components of an affine relation
are both OR-free and NAND-free the subgraphs that they induce are
hypercubes, which are also bijunctive relations. Therefore an affine
relation is also componentwise bijunctive. With this, it must also
be safely OR-free, safely OR-free and safely componentwise bijunctive
by ({*}).
\end{proof}

\section{Results}

We are now ready to state the results for CNF\textsubscript{C}($\mathcal{S}$)-formulas;
in the subsequent sections we will prove them. The following two theorems
give complete classifications up to polynomial-time isomorphisms.
They are summarized in the table below.

\begin{table}[!h]
\noindent \begin{centering}
\renewcommand{\arraystretch}{1.1}
\begingroup\tabcolsep=3pt
\par\end{centering}

\noindent \begin{centering}
\begin{tabular}{|c||c|c|c|c|}
\hline 
{\footnotesize{}$\mathcal{S}$} & \noun{\footnotesize{}Sat}{\footnotesize{}\textsubscript{C}($\mathcal{S}$) } & \noun{\footnotesize{}Conn}{\footnotesize{}\textsubscript{C}($\mathcal{S}$)} & \emph{\noun{\footnotesize{}st-}}\noun{\footnotesize{}Conn}{\footnotesize{}\textsubscript{C}($\mathcal{S}$)} & {\footnotesize{}Diameter}\tabularnewline
\hline 
\hline 
{\footnotesize{}not safely tight} & \multirow{2}{*}{{\footnotesize{}NP-complete}} & {\footnotesize{}PSPACE-complete} & {\footnotesize{}PSPACE-complete} & {\footnotesize{}$2^{\Omega(\sqrt{n})}$}\tabularnewline
\cline{3-5} 
{\footnotesize{}safely tight, not Schaefer} &  & \multirow{2}{*}{{\footnotesize{}coNP-complete}} & \multirow{3}{*}{{\footnotesize{}in P}} & \multirow{3}{*}{{\footnotesize{}$O(n)$}}\tabularnewline
\cline{2-2} 
{\footnotesize{}Schaefer, not CPSS} & \multirow{2}{*}{{\footnotesize{}in P}} &  &  & \tabularnewline
\cline{3-3} 
{\footnotesize{}CPSS} &  & {\footnotesize{}in P} &  & \tabularnewline
\hline 
\end{tabular}
\par\end{centering}

\noindent \begin{centering}
\endgroup
\par\end{centering}

\protect\caption[Our classifications for CNF\protect\textsubscript{C}($\mathcal{S}$)-formulas,
in comparison to \noun{Sat}]{\emph{}\label{tab:cc}Our classifications for CNF\protect\textsubscript{C}($\mathcal{S}$)-formulas,
in comparison to \noun{Sat}.}
\end{table}

\begin{thm}[Dichotomy theorem for \noun{st-Conn}\textsubscript{C}($\mathcal{S}$)
and the diameter]
\label{thm:dich}Let $\mathcal{S}$ be a finite set of logical relations.
\begin{enumerate}
\item If $\mathcal{S}$ is safely tight, \noun{st-Conn}\textsubscript{\noun{C}}\noun{($\mathcal{S}$)}
is in \noun{P}, and for every \emph{CNF\textsubscript{C}($\mathcal{S}$)}-formula
$\phi$, the diameter of $G(\phi)$ is linear in the number of variables.
\item Otherwise,\emph{ }\noun{st-Conn}\textsubscript{\noun{C}}\noun{($\mathcal{S}$)}
is \noun{PSPACE}-complete, and there are \emph{CNF\textsubscript{C}($\mathcal{S}$)}-formulas
$\phi$, such that the diameter of $G(\phi)$ is exponential in the
number of variables.
\end{enumerate}
\end{thm}
\begin{proof}
1. See \prettyref{lem:c47}.

2. See \prettyref{cor:pc}.\end{proof}
\begin{thm}[Trichotomy theorem for \noun{Conn}\textsubscript{C}($\mathcal{S}$)]
\label{thm:trich}Let $\mathcal{S}$ be a finite set of logical relations.
\begin{enumerate}
\item If $\mathcal{S}$ is CPSS,\emph{ }\noun{Conn}\textsubscript{\noun{C}}\noun{($\mathcal{S}$)}
is in \noun{P}.
\item Else if $\mathcal{S}$ is safely tight, \noun{Conn}\textsubscript{\noun{C}}\noun{($\mathcal{S}$)}
is \noun{coNP}-complete.
\item Else,\emph{ }\noun{Conn}\textsubscript{\noun{C}}\noun{($\mathcal{S}$)}
is \noun{PSPACE}-complete.
\end{enumerate}
\end{thm}
\begin{proof}
1. See \prettyref{cor:alg-1}.

2. See \prettyref{cor:co}.

3. See \prettyref{cor:pc}.
\end{proof}

\section{\label{sec:ge}The General Case: Reduction~from~a~Turing~Machine}

We start with the general case. Gopalan et al.\ showed that for 3-CNF-formulas,\emph{\noun{
st-}}\noun{Conn}\textsubscript{C} and \noun{Conn}\textsubscript{C}
are PSPACE-complete, and the diameter can be exponential:
\begin{lem}
\label{lem:l36}\emph{\citep[Lemma 3.6]{gop}} ~For general CNF-formulas,
as well as for 3-CNF-formulas,  \noun{st-Conn}\textsubscript{\emph{C}}
and \noun{Conn}\textsubscript{\emph{C}} are $\mathrm{PSPACE}$-complete.
\end{lem}
Showing that the problems are in PSPACE is straightforward: Given
a CNF-formula $\phi$ and two solutions $\boldsymbol{s}$ and $\boldsymbol{t}$,
we can guess a path of length at most $2^{n}$ between them and verify
that each vertex along the path is indeed a solution. Hence \noun{st-Conn}
is in $\mathrm{NPSPACE}$, which equals PSPACE by Savitch\textquoteright s
theorem. For \noun{Conn}, by reusing space we can check for all pairs
of vectors whether they are satisfying, and, if they both are, whether
they are connected in $G(\phi)$.

The hardness-proof is quite intricate: it consists of a direct reduction
from the computation of a space-bounded Turing machine $M$. The input-string
$w$ of $M$ is mapped to a 3-CNF-formula $\phi$ and two satisfying
assignments $\boldsymbol{s}$ and $\boldsymbol{t}$, corresponding
to the initial and accepting configuration of a Turing machine $M'$
constructed from $M$ and $w$, s.t. $\boldsymbol{s}$ and $\boldsymbol{t}$
are connected in $G(\phi)$ iff $M$ accepts $w$. Further, all satisfying
assignments of $\phi$ are connected to either $\boldsymbol{s}$ or
$\boldsymbol{t}$, so that $G(\phi)$ is connected iff $M$ accepts.
\begin{lem}
\label{lem:l37}\emph{\citep[Lemma 3.7]{gop}} ~For $n$ even, there
is a 3-CNF-formula $\phi_{n}$ with $n$ variables and $O(n^{2})$
clauses, s.t. $G(\phi_{n})$ is a path of length greater than $2^{\frac{n}{2}}$.
\end{lem}
The proof of this lemma is by direct construction of such a formula.

\section{\label{sub:struc}Extension of PSPACE-Completeness: Structural~Expressibility}

To show that PSPACE-hardness and exponential diameter extend to all
not (safely) tight sets of relations, Gopalan et al.~used the concept
of structural expressibility, which is a modification of Schaefer's
``representability'' that he used for his dichotomy theorem\footnote{While Schaefer\textquoteright s dichotomy theorem and many related
complexity classifications can also be proved using Post's classification
of all closed classes of Boolean functions and a Galois correspondence
(see e.g. \citep{kolaitis5250complexity}), this seems not possible
for our connectivity problems: The boundaries here ``cut across Boolean
clones'' (more exactly: co-clones), as already Gopalan et al.\ noted
\citep{gop}. For example, the co-clone of both $R=\{100,010,001\}$
and $R'=\{100,010,001,110,101\}$ is $\mathsf{I_{2}}$, but $R$ is
safely OR-free and thus tight, while $R'$ is not safely tight.}, so let us have a quick look at this first:
\begin{thm}[Schaefer\textquoteright s dichotomy theorem \citep{Schaefer:1978:CSP:800133.804350}]
\label{thm:sch}Let $\mathcal{S}$ be a finite set of logical relations.
\begin{enumerate}
\item If $\mathcal{S}$ is Schaefer, then \noun{Sat}\emph{\textsubscript{C}($\mathcal{S}$)}
is in P; otherwise, \noun{Sat}\emph{\textsubscript{C}($\mathcal{S}$)}
is \noun{NP}-complete.
\item If $\mathcal{S}$ is 0-valid, 1-valid, or Schaefer, then \noun{Sat}\emph{($\mathcal{S}$)}
is in P; otherwise, \noun{Sat}\emph{($\mathcal{S}$)} is \noun{NP}-complete.\footnote{Here we assume that $\mathcal{S}$ contains no empty relations, see
\prettyref{sec:No-Constants}.}
\end{enumerate}
\end{thm}
Schaefer first proved statement 1, and from that derived the no-constants
version; we here discuss only the proof statement 1.

Schaefer used a reduction from satisfiability of 3-CNF-formulas, i.e.\ CNF\textsubscript{C}($\mathcal{S}_{3}$)-formulas
(see \prettyref{def:cnf-1}), which was already known to be NP-complete
by the Cook\textendash Levin theorem. Therefor, he exploited that
any existentially quantified formula is satisfiability-equivalent
to the formula with the quantifiers removed, and introduced the notion
of \emph{representability}, that became also know as \emph{expressibility}:
\begin{defn}
\label{def:expr}A relation $R$ is \emph{expressible} from a set
of relations $\mathcal{S}$ if there is a CNF\textsubscript{C}($\mathcal{S}$)-formula
$\phi(\boldsymbol{x},\boldsymbol{y})$ such that $R=\{\boldsymbol{a}|\exists\boldsymbol{y}\phi(\boldsymbol{a},\boldsymbol{y})\}$.
\end{defn}
He then showed that every Boolean relation is expressible from any
set of relations that is not Schaefer, and that this expression can
efficiently be constructed.

With this, it is easy to see that for every non-Schaefer set $\mathcal{S}$,
satisfiability of any CNF\textsubscript{C}($\mathcal{S}_{3}$)-formula
$\psi$ can be reduced to satisfiability of a CNF\textsubscript{C}($\mathcal{S}$)-formula,
constructed as follows: Replace in $\psi$ every constraint $R(\boldsymbol{\xi})$
by $\phi(\boldsymbol{\xi},\boldsymbol{y})$ with $\phi$ from \prettyref{def:expr},
and new variables $\boldsymbol{y}$, distinct for each constraint.

As Gopalan et al.~explain in section 3.1 of \citep{gop}, for the
connectivity problems, expressibility is not sufficient; therefore,
they introduced \emph{structural expressibility}:
\begin{defn}
A relation $R$ is \emph{structurally expressible} from a set of relations
$\mathcal{S}$ if there is a CNF\textsubscript{C}($\mathcal{S}$)-formula
$\phi$ such that the following conditions hold:
\begin{enumerate}
\item $R=\{\boldsymbol{a}|\exists\boldsymbol{y}\phi(\boldsymbol{a},\boldsymbol{y})\}$.
\item For every $\boldsymbol{a}\in R$, the graph $G(\phi(\boldsymbol{a},\boldsymbol{y}))$
is connected.
\item For $\boldsymbol{a},\boldsymbol{b}\in R$ with $|\boldsymbol{a}-\boldsymbol{b}|=1$,
there exists a \emph{witness} $\boldsymbol{w}$ such that $(\boldsymbol{a},\boldsymbol{w})$
and $(\boldsymbol{b},\boldsymbol{w})$ are solutions of $\phi$.
\end{enumerate}
\end{defn}
Gopalan et al.~now argued that connectivity were retained when replacing
every constraint $R$ with a structural expression of $R$ in a CNF\textsubscript{C}($\mathcal{S}$)-formula.
In fact, this is only true for CNF\textsubscript{C}($\mathcal{S}$)-formulas
where no variable is used more than once in any constraint, and their
proof is only correct for such formulas that also use no constants:
\begin{lem}
\label{lem:l32}\emph{\citep[corrected from][Lemma 3.2]{gop}} ~Let
$\mathcal{S}$ and $\mathcal{S}'$ be sets of relations such that
every $R\in\mathcal{S}'$ is structurally expressible from $\mathcal{S}$.
Given a CNF($\mathcal{S}'$)-formula $\psi(\boldsymbol{x})$ (without
constants), where no variable is used more than once in any constraint,
one can efficiently construct a CNF\textsubscript{C}($\mathcal{S}$)-formula
$\varphi(\boldsymbol{x},\boldsymbol{y})$ such that
\begin{enumerate}
\item $\psi(\boldsymbol{x})=\exists\boldsymbol{y}\varphi(\boldsymbol{x},\boldsymbol{y})$;
\item if $(\boldsymbol{s}\boldsymbol{,w^{s}}),$ $(\boldsymbol{t},\boldsymbol{w^{t}})$
are connected in $G(\varphi)$ by a path of length $d$, then there
is a path from $\boldsymbol{s}$ to $\boldsymbol{t}$ in $G(\psi)$
of length at most $d$;
\item if $\boldsymbol{s},\boldsymbol{t}\in\psi$ are connected in $G(\psi)$,
then for every witness $\boldsymbol{w^{s}}$ of $\boldsymbol{s}$,
and every witness $\boldsymbol{w^{t}}$ of $\boldsymbol{t}$, there
is a path from $(\boldsymbol{s}\boldsymbol{,w^{s}})$ to $(\boldsymbol{t},\boldsymbol{w^{t}})$
in $G(\varphi)$.
\end{enumerate}
\end{lem}
In Gopalan et al.'s proof, we only clarify the notation a little:
\begin{proof}
Let $\psi(\boldsymbol{x})=C_{1}\wedge\cdots\wedge C_{m}$ with $C_{j}=R_{j}(\boldsymbol{x}_{j})$,
where $R_{j}$ is some relation from $\mathcal{S}'$, and $\boldsymbol{x}_{j}$
is the vector of variables to which $R_{j}$ is applied. Let $\varphi_{j}$
be the structural expression for $R_{j}$ from ${\cal S}$, so that
$R_{j}(\boldsymbol{x}_{j})\equiv\exists\boldsymbol{y}_{j}~\varphi_{j}(\boldsymbol{x}_{j},\boldsymbol{y}_{j})$.
Let $\boldsymbol{y}$ be the vector $(\boldsymbol{y}_{1},\dots,\boldsymbol{y}_{m})$
and let $\varphi(\boldsymbol{x},\boldsymbol{y})$ be the formula $\wedge_{j=1}^{m}\varphi_{j}(\boldsymbol{x}_{j},\boldsymbol{y}_{j})$.
Then $\psi(\boldsymbol{x})\equiv\exists\boldsymbol{y}~\varphi(\boldsymbol{x},\boldsymbol{y})$.

Statement $2$ follows from 1 by projection of the path on the coordinates
of $\boldsymbol{x}$. For statement 3, consider $\boldsymbol{s},\boldsymbol{t}\in\psi$
that are connected in $G(\psi)$ via a path $\boldsymbol{s}=\boldsymbol{u^{0}}\rightarrow\boldsymbol{u^{1}}\rightarrow\dots\rightarrow\boldsymbol{u^{r}}=\boldsymbol{t}$
. For every $\boldsymbol{u^{i}},\boldsymbol{u^{i+1}}$, and clause
$C_{j}$, there exists an assignment $\boldsymbol{w_{j}^{i}}$ to
$\boldsymbol{y}_{j}$ such that both $(\boldsymbol{u_{j}^{i}},\boldsymbol{w_{j}^{i}})$
and $(\boldsymbol{u_{j}^{i+1}},\boldsymbol{w_{j}^{i}})$ are solutions
of $\varphi_{j}$, by condition 3 of structural expressibility. Thus
$(\boldsymbol{u^{i}},\boldsymbol{w^{i}})$ and $(\boldsymbol{u^{i+1}},\boldsymbol{w^{i}})$
are both solutions of $\varphi$, where $\boldsymbol{w^{i}}=(\boldsymbol{w_{1}^{i}},\dots,\boldsymbol{w_{m}^{i}})$.
Further, for every $\boldsymbol{u^{i}}$, the space of solutions of
$\varphi(\boldsymbol{u^{i}},\boldsymbol{y})$ is the product space
of the solutions of $\varphi_{j}(\boldsymbol{u_{j}^{i}},\boldsymbol{y}_{j})$
over $j=1,\dots,m$. Since these are all connected by condition 2
of structural expressibility, $G(\varphi(\boldsymbol{u^{i}},\boldsymbol{y}))$
is connected. The following describes a path from $(\boldsymbol{s},\boldsymbol{w^{s}})$
to $(\boldsymbol{t},\boldsymbol{w^{t}})$ in $G(\varphi)$: ~$(\boldsymbol{s},\boldsymbol{w^{s}})\rightsquigarrow(\boldsymbol{s},\boldsymbol{w^{0}})\rightarrow(\boldsymbol{u^{1}},\boldsymbol{w^{0}})\rightsquigarrow(\boldsymbol{u^{1}},\boldsymbol{w^{1}})\rightarrow\dots\rightsquigarrow(\boldsymbol{u^{r-1}},\boldsymbol{w^{r-1}})\rightarrow(\boldsymbol{t},\boldsymbol{w^{r-1}})\rightsquigarrow(\boldsymbol{t},\boldsymbol{w^{t}})$.
Here $\rightsquigarrow$ indicates a path in $G(\varphi(\boldsymbol{u^{i}},\boldsymbol{y}))$.
\end{proof}
It is easy to show that the statement of \prettyref{lem:l32} is also
correct if we allow constants in $\psi$; however, we don't need this
result. In \citep{csp}, we explain in detail the problem with repeated
variables in constraint applications.

We have to change Gopalan et al.'s corollary accordingly; we denote
the connectivity problems for CNF($\mathcal{S}$)-formulas without
repeated variables in constraints (and without constants) by the subscript
\textsubscript{ni}:
\begin{cor}
\label{cor:c33}\emph{\citep[corrected from][Corollary 3.3]{gop}}
~Suppose $\mathcal{S}$ and $\mathcal{S}'$ are sets of relations
such that every $R\in\mathcal{S}'$ is structurally expressible from
$\mathcal{S}$.
\begin{enumerate}
\item There are polynomial-time reductions from \noun{Conn}\emph{\textsubscript{ni}($\mathcal{S}$')}
to \noun{Conn}\emph{\textsubscript{C}($\mathcal{S}$)}, and from
\noun{st-Conn}\emph{\textsubscript{ni}($\mathcal{S}$')} to \noun{st-Conn}\emph{\textsubscript{C}($\mathcal{S}$)}.
\item If there exists a \emph{CNF\textsubscript{ni}($\mathcal{S}$')}-formula
$\psi(\boldsymbol{x})$ with $n$ variables, $m$ clauses and diameter
$d$, then there exists a \emph{CNF\textsubscript{C}($\mathcal{S}$)}-formula
$\phi(\boldsymbol{x},\boldsymbol{y})$, where $\boldsymbol{y}$ is
a vector of $O(m)$ variables, such that the diameter of $G(\phi)$
is at least $d$.
\end{enumerate}
\end{cor}
Since 3-CNF-fomulas are CNF\textsubscript{ni}($\mathcal{S}_{1}\cup\mathcal{S}_{2}\cup\mathcal{S}_{3}$)-formulas,
for the reductions to work it now remains to prove that $\mathcal{S}_{1}\cup\mathcal{S}_{2}\cup\mathcal{S}_{3}$
is structurally expressible from any not safely tight set. As \prettyref{thm:se}
below shows, in fact every Boolean relation is structurally expressible
from any such set. The long proof of the next lemma contains only
minor modifications from \citep{gop}.
\begin{lem}
\label{lem:l34}\emph{\citep[modified from][Lemma 3.4]{gop} }~If
a set $\mathcal{S}$ of relations is not safely tight, $\mathcal{S}_{3}$
is structurally expressible from $\mathcal{S}$.\end{lem}
\begin{proof}
First, observe that all $2$-clauses are structurally expressible
from ${\cal S}$. There exists $R\in{\cal S}$ which is not safely
OR-free, so we can express $(x_{1}\vee x_{2})$ by substituting constants
and identifying variables in $R$. Similarly, we can express $(\bar{x}_{1}\vee\bar{x}_{2})$
using a relation that is not safely NAND-free. The last 2-clause $(x_{1}\vee\bar{x}_{2})$
can be obtained from OR and NAND by a technique that corresponds to
reverse resolution. $(x_{1}\vee\bar{x}_{2})=\exists y~(x_{1}\vee y)\wedge(\bar{y}\vee\bar{x}_{2})$.
It is easy to see that this gives a structural expression. From here
onwards we assume that ${\cal S}$ contains all 2-clauses. The proof
now proceeds in four steps. First, we will express a relation in which
there exist two elements that are at graph distance larger than their
Hamming distance. Second, we will express a relation that is just
a single path between such elements. Third, we will express a relation
which is a path of length 4 between elements at Hamming distance 2.
Finally, we will express the 3-clauses.

\emph{Step 1. Structurally expressing a relation in which some distance
expands.}\\
 For a relation $R$, we say that the distance between $\boldsymbol{a}$
and $\boldsymbol{b}$ \emph{expands} if $\boldsymbol{a}$ and $\boldsymbol{b}$
are connected in $G(R)$, but $d_{R}(\boldsymbol{a},\boldsymbol{b})>|\boldsymbol{a}-\boldsymbol{b}|$.
Later on, we will show that no distance expands in safely componentwise
bijunctive relations. The same also holds true for the relation $R_{{\rm {NAE}}}=\{0,1\}^{3}\setminus\{000,111\}$,
which is not safely componentwise bijunctive. Nonetheless, we show
here that if $R$ is not safely componentwise bijunctive, then, by
adding $2$-clauses, we can structurally express a relation $Q$ in
which some distance expands. For instance, when $R=R_{{\rm NAE}}$,
then we can take $Q(x_{1},x_{2},x_{3})=R_{{\rm NAE}}(x_{1},x_{2},x_{3})\wedge(\bar{x_{1}}\vee\bar{x}_{3})$.
The distance between $\boldsymbol{a}=100$ and $\boldsymbol{b}=001$
in $Q$ expands. Similarly, in the general construction, we identify
$\boldsymbol{a}$ and $\boldsymbol{b}$ on a cycle, and add $2$-clauses
that eliminate all the vertices along the shorter arc between $\boldsymbol{a}$
and $\boldsymbol{b}$.

\begin{figure}[!h]
\begin{centering}
\label{fig:expand} \includegraphics[scale=0.66]{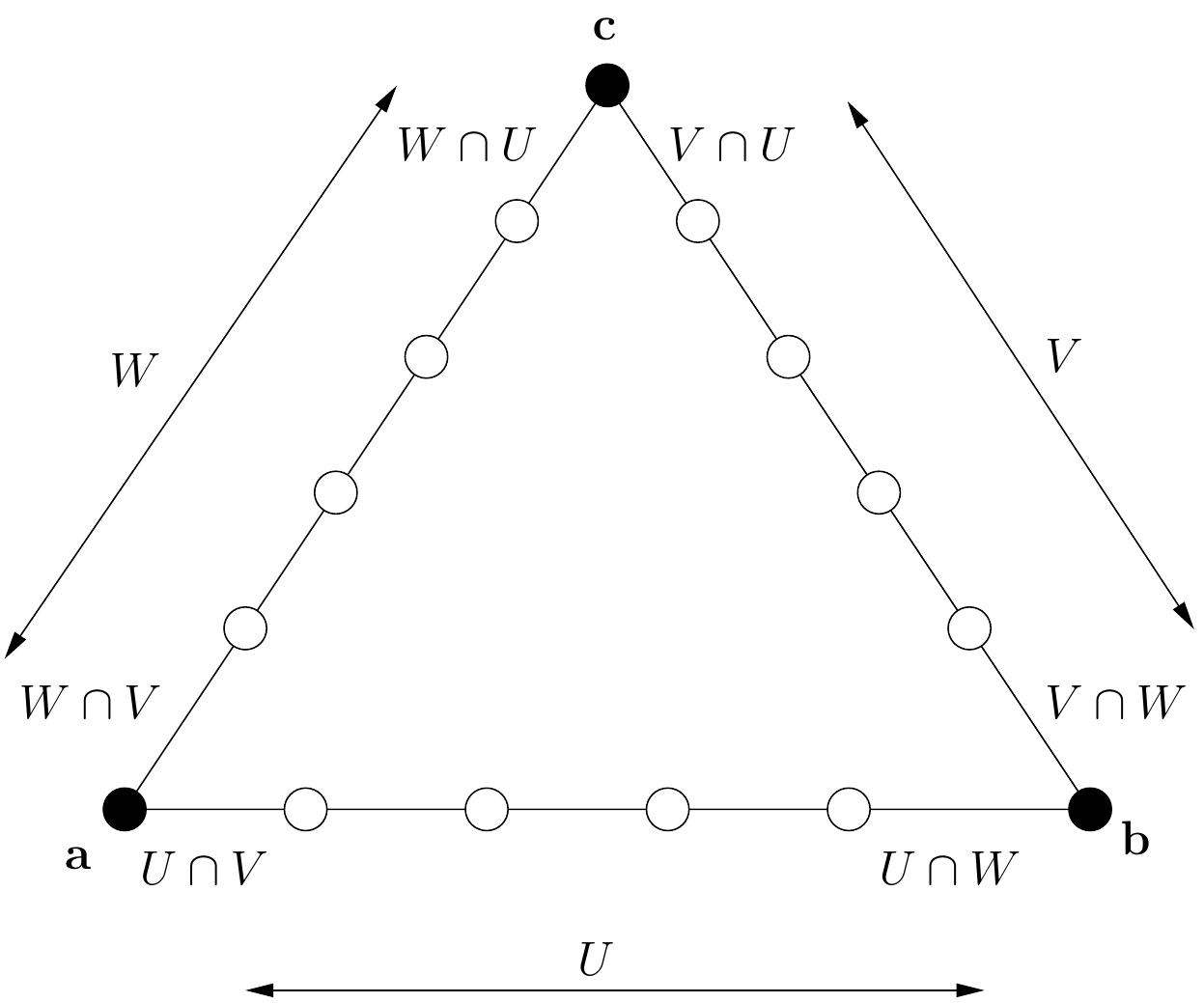}\vspace{0.6cm}
\\
\includegraphics[scale=0.66]{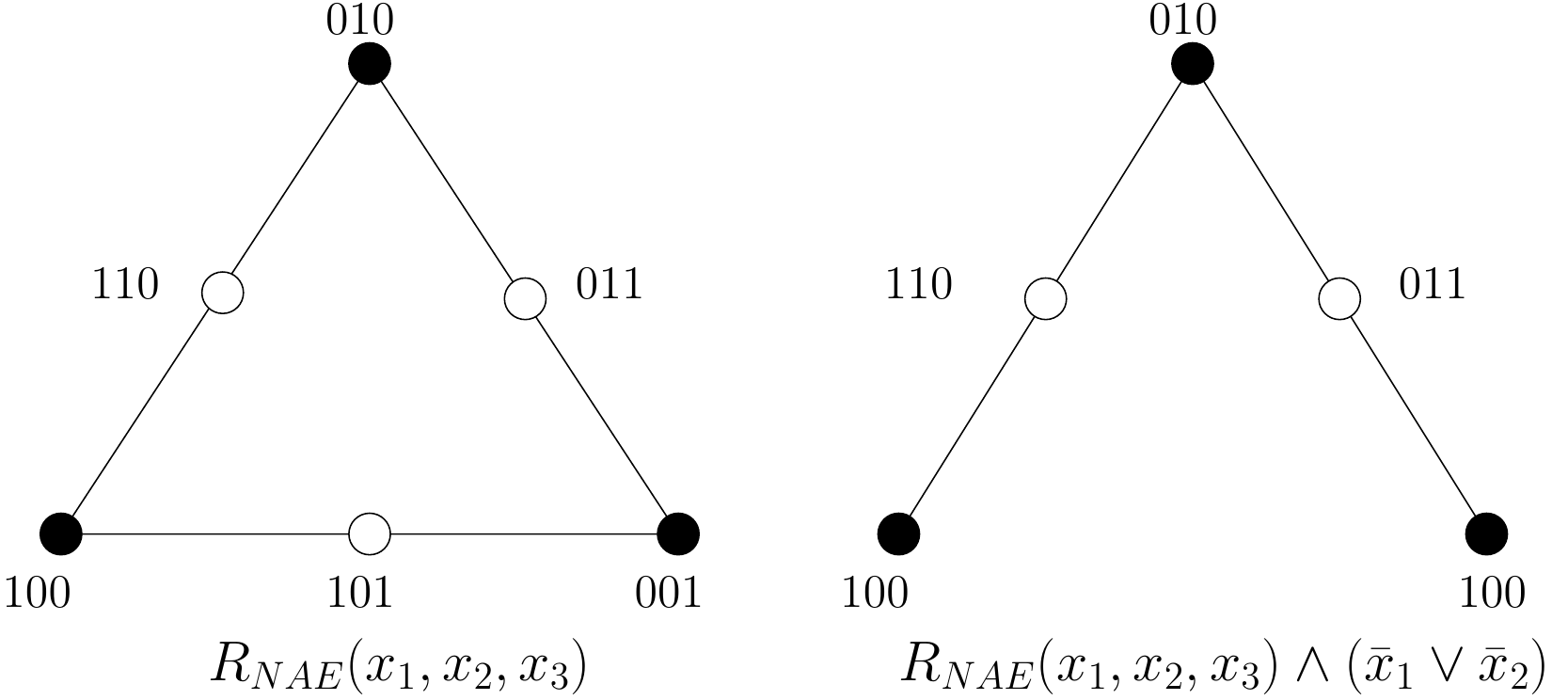}
\par\end{centering}

\centering{}\protect\caption[Proof of Step 1 of \prettyref{lem:l34}, and an example]{\emph{}Proof of Step 1, and an example.}
\end{figure}

Since ${\cal S}$ is not safely tight, it contains a relation which
is not safely componentwise bijunctive, from which we can obtain a
not componentwise bijunctive relation $R$. If $R$ contains $\boldsymbol{a},\boldsymbol{b}$
where the distance between them expands, we are done. So assume that
for all $\boldsymbol{a},\boldsymbol{b}\in G(R)$, $d_{R}(\boldsymbol{a},\boldsymbol{b})=|\boldsymbol{a}-\boldsymbol{b}|$.
Since $R$ is not componentwise bijunctive, there exists a triple
of assignments $\boldsymbol{a},\boldsymbol{b},\boldsymbol{c}$ lying
in the same component such that $\mathrm{Maj}(\boldsymbol{a},\boldsymbol{b},\boldsymbol{c})$
is not in that component (which also easily implies it is not in $R$).
Choose the triple such that the sum of pairwise distances $d_{R}(\boldsymbol{a},\boldsymbol{b})+d_{R}(\boldsymbol{b},\boldsymbol{c})+d_{R}(\boldsymbol{c},\boldsymbol{a})$
is minimized. Let $U=\{i|a_{i}\neq b_{i}\}$, $V=\{i|b_{i}\neq c_{i}\}$,
and $W=\{i|c_{i}\neq a_{i}\}$. Since $d_{R}(\boldsymbol{a},\boldsymbol{b})=|\boldsymbol{a}-\boldsymbol{b}|$,
a shortest path does not flip variables outside of $U$, and each
variable in $U$ is flipped exactly once. The same holds for $V$
and $W$. We note some useful properties of the sets $U,V,W$.
\begin{enumerate}
\item \emph{Every index $i\in U\cup V\cup W$ occurs in exactly two of $U,V,W$.}\\
 Consider going by a shortest path from $\boldsymbol{a}$ to $\boldsymbol{b}$
to $\boldsymbol{c}$ and back to $\boldsymbol{a}$. Every $i\in U\cup V\cup W$
is seen an even number of times along this path since we return to
$\boldsymbol{a}$. It is seen at least once, and at most thrice, so
in fact it occurs twice.
\item \emph{Every pairwise intersection $U\cap V,V\cap W$ and $W\cap U$
is non-empty}.\\
 Suppose the sets $U$ and $V$ are disjoint. From Property 1, we
must have $W=U\cup V$. But then it is easy to see that $\mathrm{Maj}(\boldsymbol{a},\boldsymbol{b},\boldsymbol{c})=\boldsymbol{b}$
which is in $R$. This contradicts the choice of $\boldsymbol{a},\boldsymbol{b},\boldsymbol{c}$.
\item \emph{The sets $U\cap V$ and $U\cap W$ partition the set $U$}.\\
 By Property $1$, each index of $U$ occurs in one of $V$ and $W$
as well. Also since no index occurs in all three sets $U,V,W$ this
is in fact a disjoint partition.
\item \emph{For each index $i\in U\cap W$, it holds that $\boldsymbol{a}\oplus\boldsymbol{e}_{i}\not\in R$.}\\
 Assume for the sake of contradiction that $\boldsymbol{a'}=\boldsymbol{a}\oplus\boldsymbol{e}_{i}\in R$.
Since $i\in U\cap W$ we have simultaneously moved closer to both
$\boldsymbol{b}$ and $\boldsymbol{c}$. Hence we have $d_{R}(\boldsymbol{a'},\boldsymbol{b})+d_{R}(\boldsymbol{b},\boldsymbol{c})+d_{R}(\boldsymbol{c},\boldsymbol{a'})<d_{R}(\boldsymbol{a},\boldsymbol{b})+d_{R}(\boldsymbol{b},\boldsymbol{c})+d_{R}(\boldsymbol{c},\boldsymbol{a})$.
Also $\mathrm{Maj}(\boldsymbol{a'},\boldsymbol{b},\boldsymbol{c})=\mathrm{Maj}(\boldsymbol{a},\boldsymbol{b},\boldsymbol{c})\not\in R$.
But this contradicts our choice of $\boldsymbol{a},\boldsymbol{b},\boldsymbol{c}$. 
\end{enumerate}
Property 4 implies that the shortest paths to $\boldsymbol{b}$ and
$\boldsymbol{c}$ diverge at $\boldsymbol{a}$, since for any shortest
path to $\boldsymbol{b}$ the first variable flipped is from $U\cap V$
whereas for a shortest path to $\boldsymbol{c}$ it is from $W\cap V$.
Similar statements hold for the vertices $\boldsymbol{b}$ and $\boldsymbol{c}$.
Thus along the shortest path from $\boldsymbol{a}$ to $\boldsymbol{b}$
the first bit flipped is from $U\cap V$ and the last bit flipped
is from $U\cap W$. On the other hand, if we go from $\boldsymbol{a}$
to $\boldsymbol{c}$ and then to $\boldsymbol{b}$, all the bits from
$U\cap W$ are flipped before the bits from $U\cap V$. We use this
crucially to define $Q$. We will add a set of 2-clauses that enforce
the following rule on paths starting at $\boldsymbol{a}$: \emph{Flip
variables from $U\cap W$ before variables from $U\cap V$.} This
will eliminate all shortest paths from $\boldsymbol{a}$ to $\boldsymbol{b}$
since they begin by flipping a variable in $U\cap V$ and end with
$U\cap W$. The paths from $\boldsymbol{a}$ to $\boldsymbol{b}$
via $\boldsymbol{c}$ survive since they flip $U\cap W$ while going
from $\boldsymbol{a}$ to $\boldsymbol{c}$ and $U\cap V$ while going
from $\boldsymbol{c}$ to $\boldsymbol{b}$. However all remaining
paths have length at least $|\boldsymbol{a}-\boldsymbol{b}|+2$ since
they flip twice some variables not in $U$.

Take all pairs of indices $\{(i,j)|i\in U\cap W,j\in U\cap V\}$.
The following conditions hold from the definition of $U,V,W$: $a_{i}=\bar{c}_{i}=\bar{b}_{i}$
and $a_{j}=c_{j}=\bar{b}_{j}$. Add the 2-clause $C_{ij}$ asserting
that the pair of variables $x_{i}x_{j}$ must take values in $\{a_{i}a_{j},c_{i}c_{j},b_{i}b_{j}\}=\{a_{i}a_{j},\bar{a}_{i}a_{j},\bar{a}_{i}\bar{a}_{j}\}$.
The new relation is $Q=R\wedge_{i,j}C_{ij}$. Note that $Q\subset R$.
We verify that the distance between $\boldsymbol{a}$ and $\boldsymbol{b}$
in $Q$ expands. It is easy to see that for any $j\in U$, the assignment
$\boldsymbol{a}\oplus\boldsymbol{e}_{j}\not\in Q$. Hence there are
no shortest paths left from $\boldsymbol{a}$ to $\boldsymbol{b}$.
On the other hand, it is easy to see that $\boldsymbol{a}$ and $\boldsymbol{b}$
are still connected, since the vertex $\boldsymbol{c}$ is still reachable
from both.

\emph{Step 2. Isolating a pair of assignments whose distance expands.}\\
The relation $Q$ obtained in Step 1 may have several disconnected
components. This \emph{cleanup} step isolates a single pair of assignments
whose distance expands. By adding $2$-clauses, we show that one can
express a path of length $r+2$ between assignments at distance $r$.

Take $\boldsymbol{a},\boldsymbol{b}\in Q$ whose distance expands
in $Q$ and $d_{Q}(\boldsymbol{a},\boldsymbol{b})$ is minimized.
Let $U=\{i|a_{i}\neq b_{i}\}$, and $|U|=r$. Shortest paths between
$\boldsymbol{a}$ and $\boldsymbol{b}$ have certain useful properties:
\begin{enumerate}
\item \emph{Each shortest path flips every variable from $U$ exactly once.}\\
 Observe that each index $j\in U$ is flipped an odd number of times
along any path from $\boldsymbol{a}$ to $\boldsymbol{b}$. Suppose
it is flipped thrice along a shortest path. Starting at $\boldsymbol{a}$
and going along this path, let $\boldsymbol{b'}$ be the assignment
reached after flipping $j$ twice. Then the distance between $\boldsymbol{a}$
and $\boldsymbol{b'}$ expands, since $j$ is flipped twice along
a shortest path between them in $Q$. Also $d_{Q}(\boldsymbol{a},\boldsymbol{b'})<d_{Q}(\boldsymbol{a},\boldsymbol{b})$,
contradicting the choice of $\boldsymbol{a}$ and $\boldsymbol{b}$.
\item \emph{Every shortest path flips exactly one variable $i\not\in U$.}
\\
 Since the distance between $\boldsymbol{a}$ and $\boldsymbol{b}$
expands, every shortest path must flip some variable $i\not\in U$.
Suppose it flips more than one such variable. Since $\boldsymbol{a}$
and $\boldsymbol{b}$ agree on these variables, each of them is flipped
an even number of times. Let $i$ be the first variable to be flipped
twice. Let $\boldsymbol{b'}$ be the assignment reached after flipping
$i$ the second time. It is easy to verify that the distance between
$\boldsymbol{a}$ and $\boldsymbol{b'}$ also expands, but $d_{Q}(\boldsymbol{a},\boldsymbol{b'})<d_{Q}(\boldsymbol{a},\boldsymbol{b})$.
\item \emph{The variable $i\not\in U$ is the first and last variable to
be flipped along the path.}\\
Assume the first variable flipped is not $i$. Let $\boldsymbol{a'}$
be the assignment reached along the path before we flip $i$ the first
time. Then $d_{Q}(\boldsymbol{a'},\boldsymbol{b})<d_{Q}(\boldsymbol{a},\boldsymbol{b})$.
The distance between $\boldsymbol{a'}$ and $\boldsymbol{b}$ expands
since the shortest path between them flips the variables $i$ twice.
This contradicts the choice of $\boldsymbol{a}$ and $\boldsymbol{b}$.
Assume $j\in U$ is flipped twice. Then as before we get a pair $\boldsymbol{a'},\boldsymbol{b'}$
that contradict the choice of $\boldsymbol{a},\boldsymbol{b}$. 
\end{enumerate}
Every shortest path between $\boldsymbol{a}$ and $\boldsymbol{b}$
has the following structure: first a variable $i\not\in U$ is flipped
to $\bar{a}_{i}$, then the variables from $U$ are flipped in some
order, finally the variable $i$ is flipped back to $a_{i}$.

Different shortest paths may vary in the choice of $i\not\in U$ in
the first step and in the order in which the variables from $U$ are
flipped. Fix one such path $T\subseteq Q$. Assume that $U=\{1,\dots,r\}$
and the variables are flipped in this order, and the additional variable
flipped twice is $r+1$. Denote the path by $\boldsymbol{a}\rightarrow\boldsymbol{u^{0}}\rightarrow\boldsymbol{u^{1}}\rightarrow\dots\rightarrow\boldsymbol{u^{r}}\rightarrow\boldsymbol{b}$.
Next we prove that we cannot flip the $r+1^{th}$ variable at an intermediate
vertex along the path.
\begin{enumerate}[resume]
\item \emph{For $1\leq j\leq r-1$ the assignment $\boldsymbol{u^{j}}\oplus\boldsymbol{e_{r+1}}\not\in Q$.}
Suppose that for some $j$, we have $\boldsymbol{c}=\boldsymbol{u^{j}}\oplus\boldsymbol{e_{r+1}}\in Q$.
Then $\boldsymbol{c}$ differs from $\boldsymbol{a}$ on $\{1,\dots,i\}$
and from $\boldsymbol{b}$ on $\{i+1,\dots,r\}$. The distance from
$\boldsymbol{c}$ to at least one of $\boldsymbol{a}$ or $\boldsymbol{b}$
must expand, else we get a path from $\boldsymbol{a}$ to $\boldsymbol{b}$
through $\boldsymbol{c}$ of length $|\boldsymbol{a}-\boldsymbol{b}|$
which contradicts the fact that this distance expands. However $d_{Q}(\boldsymbol{a},\boldsymbol{c})$
and $d_{Q}(\boldsymbol{b},\boldsymbol{c})$ are strictly less than
$d_{Q}(\boldsymbol{a},\boldsymbol{b})$ so we get a contradiction
to the choice of $\boldsymbol{a},\boldsymbol{b}$. 
\end{enumerate}
We now construct the path of length $r+2$. For all $i\geq r+2$ we
set $x_{i}=a_{i}$ to get a relation on $r+1$ variables. Note that
$\boldsymbol{b}=\bar{a}_{1}\dots\bar{a}_{r}a_{r+1}$. Take $i<j\in U$.
Along the path $T$ the variable $i$ is flipped before $j$ so the
variables $x_{i}x_{j}$ take one of three values $\{a_{i}a_{j},\bar{a}_{i}a_{j},\bar{a}_{i}\bar{a}_{j}\}$.
So we add a 2-clause $C_{ij}$ that requires $x_{i}x_{j}$ to take
one of these values and take $T=Q\wedge_{i,j}C_{ij}$. Clearly, every
assignment along the path lies in $T$. We claim that these are the
only solutions. To show this, take an arbitrary assignment $\boldsymbol{c}$
satisfying the added constraints. If for some $i<j\leq r$ we have
$c_{i}=a_{i}$ but $c_{j}=\bar{a}_{j}$, this would violate $C_{ij}$.
Hence the first $r$ variables of $\boldsymbol{c}$ are of the form
$\bar{a}_{1}\dots\bar{a}_{i}a_{i+1}\dots a_{r}$ for $0\leq i\leq r$.
If $c_{r+1}=\bar{a}_{r+1}$ then $\boldsymbol{c}=\boldsymbol{u}^{i}$.
If $c_{r+1}=a_{r+1}$ then $\boldsymbol{c}=\boldsymbol{u}^{i}\oplus\boldsymbol{e}_{r+1}$.
By property 4 above, such a vector satisfies $Q$ if and only if $i=0$
or $i=r$, which correspond to $\boldsymbol{c}=\boldsymbol{a}$ and
$\boldsymbol{c}=\boldsymbol{b}$ respectively.

\emph{Step 3. Structurally expressing paths of length $4$.}\\
 Let ${\cal {P}}$ denote the set of all ternary relations whose graph
is a path of length $4$ between two assignments at Hamming distance
$2$. Up to permutations of coordinates, there are 6 such relations.
Each of them is the conjunction of a $3$-clause and a $2$-clause.
For instance, the relation $M=\{100,110,010,011,001\}$ can be written
as $(x_{1}\vee x_{2}\vee x_{3})\wedge(\bar{x}_{1}\vee\bar{x}_{3})$.
(It is named so, because its graph looks like the letter 'M' on the
cube.) These relations are ``minimal\textquotedbl{} examples of relations
that are not componentwise bijunctive. By projecting out intermediate
variables from the path $T$ obtained in Step 2, we structurally express
one of the relations in ${\cal {P}}$. We structurally express other
relations in ${\cal {P}}$ using this relation.

We will write all relations in ${\cal {P}}$ in terms of $M(x_{1},x_{2},x_{3})=(x_{1}\vee x_{2}\vee x_{3})\wedge(\bar{x}_{1}\vee\bar{x}_{3})$,
by negating variables. For example $M(\bar{x}_{1},x_{2},x_{3})=(\bar{x}_{1}\vee x_{2}\vee x_{3})\wedge(x_{1}\vee\bar{x}_{3})=\{000,010,110,111,101\}$.

Define the relation $P(x_{1},x_{r+1},x_{2})=\exists x_{3}\dots x_{r}~T(x_{1},\dots,x_{r+1})$.
The table below listing all tuples in $P$ and their witnesses, shows
that the conditions for structural expressibility are satisfied, and
$P\in{\cal {P}}$.

\vspace{0.2cm}
\begin{footnotesize} %
\begin{tabular}{|c|l|}
\hline 
$x_{1},x_{2},x_{r+1}$  & $x_{3},\dots,x_{r}$\tabularnewline
\hline 
$a_{1}a_{2}a_{r+1}$  & $a_{3}\dots a_{r}$\tabularnewline
$a_{1}a_{2}\bar{a}_{r+1}$  & $a_{3}\dots a_{r}$\tabularnewline
$\bar{a}_{1}a_{2}\bar{a}_{r+1}$  & $a_{3}\dots a_{r}$\tabularnewline
$\bar{a}_{1}\bar{a}_{2}\bar{a}_{r+1}$  & $a_{3}\dots a_{k},\ \bar{a}_{3}a_{4}\dots a_{r},\ \bar{a}_{3}\bar{a}_{4}a_{5}\dots a_{r}\ \dots\bar{a}_{3}\bar{a}_{4}\dots\bar{a}_{r}$\tabularnewline
$\bar{a}_{1}\bar{a}_{2}a_{r+1}$  & $\bar{a}_{3}\bar{a}_{4}\dots\bar{a}_{r}$\tabularnewline
\hline 
\end{tabular}\end{footnotesize}\vspace{0.2cm}

Let $P(x_{1},x_{2},x_{3})=M(l_{1},l_{2},l_{3})$, where $l_{i}$ is
one of $\{x_{i},\bar{x}_{i}\}$. We can now use $P$ and 2-clauses
to express every other relation in ${\cal {P}}$. Given $M(l_{1},l_{2},l_{3})$
every relation in ${\cal {P}}$ can be obtained by negating some subset
of the variables. Hence it suffices to show that we can express structurally
$M(\bar{l}_{1},l_{2},l_{3})$ and $M(l_{1},\bar{l}_{2},l_{3})$ ($M$
is symmetric in $x_{1}$ and $x_{3}$). In the following let $\lambda$
denote one of the literals $\{y,\bar{y}\}$, such that it is $\bar{y}$
if and only if $l_{1}$ is $\bar{x}_{1}$. 
\begin{eqnarray*}
M(\bar{l}_{1},l_{2},l_{3}) & = & (\bar{l}_{1}\vee l_{2}\vee l_{3})\wedge(l_{1}\vee\bar{l}_{3})\\
 & = & \exists y~(\bar{l}_{1}\vee\bar{\lambda})\wedge(\lambda\vee l_{2}\vee l_{3})\wedge(l_{1}\vee\bar{l}_{3})\\
 & = & \exists y~(\bar{l}_{1}\vee\bar{\lambda})\wedge(\lambda\vee l_{2}\vee l_{3})\wedge(l_{1}\vee\bar{l}_{3})\wedge(\bar{\lambda}\vee\bar{l}_{3})\\
 & = & \exists y~(\bar{l}_{1}\vee\bar{\lambda})\wedge(l_{1}\vee\bar{l}_{3})\wedge M(\lambda,l_{2},l_{3})\\
 & = & \exists y~(\bar{l}_{1}\vee\bar{\lambda})\wedge(l_{1}\vee\bar{l}_{3})\wedge P(y,x_{2},x_{3})
\end{eqnarray*}
In the second step the clause $(\bar{\lambda}\vee\bar{l}_{3})$ is
implied by the resolution of the clauses $(\bar{l}_{1}\vee\bar{\lambda})\wedge(l_{1}\vee\bar{l}_{3})$.

For the next expression let $\lambda$ denote one of the literals
$\{y,\bar{y}\}$, such that it is negated if and only if $l_{2}$
is $\bar{x}_{2}$. 
\begin{eqnarray*}
M(l_{1},\bar{l}_{2},l_{3}) & = & (l_{1}\vee\bar{l}_{2}\vee l_{3})\wedge(\bar{l}_{1}\vee\bar{l}_{3})\\
 & = & \exists y~(l_{1}\vee l_{3}\vee\lambda)\wedge(\bar{\lambda}\vee\bar{l}_{2})\wedge(\bar{l}_{1}\vee\bar{l}_{3})\\
 & = & \exists y~(\bar{\lambda}\vee\bar{l}_{2})\wedge M(l_{1},\lambda,l_{3})\\
 & = & \exists y~(\bar{\lambda}\vee\bar{l}_{2})\wedge P(x_{1},y,x_{3})
\end{eqnarray*}
The above expressions are both based on resolution and it is easy
to check that they satisfy the properties of structural expressibility.

\emph{Step 4. Structurally expressing ${\cal S}_{3}$. }\\
 We structurally express $(x_{1}\lor x_{2}\lor x_{3})$ from $M$
using a formula derived from a gadget in \citep{HD02}. This gadget
expresses $(x_{1}\vee x_{2}\vee x_{3})$ in terms of ``Protected
OR'', which corresponds to our relation $M$.

\begin{eqnarray}
(x_{1}\vee x_{2}\vee x_{3}) & = & \exists y_{1}\dots y_{5}~(x_{1}\vee\bar{y}_{1})\wedge(x_{2}\vee\bar{y}_{2})\wedge(x_{3}\vee\bar{y}_{3})\wedge(x_{3}\vee\bar{y}_{4})\nonumber \\
 &  & \wedge M(y_{1},y_{5},y_{3})\wedge M(y_{2},\bar{y}_{5},y_{4})
\end{eqnarray}
The table below listing the witnesses of each assignment for $(x_{1},x_{2},x_{3})$,
shows that the conditions for structural expressibility are satisfied.

\vspace{0.2cm}
\begin{footnotesize}\hspace{-0.54cm} %
\begin{tabular}{|c|llll|}
\hline 
$x_{1},x_{2},x_{3}$  & $y_{1}\dots y_{5}$  &  &  & \tabularnewline
\hline 
$111$  & $00011$ ~$00111$ ~$00110$ ~$00100$ ~$01100$ ~$01101$  & $01001$ ~$11001$ ~$11000$  & $10000$  & $10010$ ~$10011$ \tabularnewline
$110$  &  & $01001$ ~$11001$ ~$11000$  & $10000$  & \tabularnewline
$100$  &  &  & $10000$  & \tabularnewline
$101$  & $00011$ ~$00111$ ~$00110$ ~$00100$  &  & $10000$  & $10010$ ~$10011$ \tabularnewline
$001$  & $00011$ ~$00111$ ~$00110$ ~$00100$  &  &  & \tabularnewline
$011$  & $00011$ ~$00111$ ~$00110$ ~$00100$ ~$01100$ ~$01101$  & $01001$ &  & \tabularnewline
$010$  &  & $01001$ &  & \tabularnewline
\hline 
\end{tabular}\end{footnotesize} \vspace{0.2cm}

\noindent From the relation $(x_{1}\vee x_{2}\vee x_{3})$ we derive
the other 3-clauses by reverse resolution, for instance 
\[
(\bar{x}_{1}\vee x_{2}\vee x_{3})=\exists y~(\bar{x}_{1}\vee\bar{y})\wedge(y\vee x_{2}\vee x_{3})
\]
\end{proof}
\begin{lem}
\emph{\citep[Lemma 3.5]{gop}} ~Let $R\subseteq\{0,1\}^{k}$ be any
relation of arity $k\geq1$. $R$ is structurally expressible from
$\mathcal{S}_{3}$.
\end{lem}
The next theorem follows from the last two lemmas:
\begin{thm}[{Structural expressibility theorem, modified from \citep[Theorem 2.7]{gop}}]
\emph{\label{thm:se}}Let $\mathcal{S}$ be a finite set of logical
relations. If $\mathcal{S}$ is not safely tight, then every logical
relation is structurally expressible from $\mathcal{S}$.
\end{thm}
With \prettyref{lem:l36}, \prettyref{cor:c33} and the preceding
theorem, we can now complete the proofs for PSPACE-completeness and
the exponential diameter:
\begin{cor}
\label{cor:pc}If a finite set $\mathcal{S}$ of logical relations
is not safely tight, then \noun{st-Conn}\textsubscript{\noun{C}}\noun{($\mathcal{S}$)
}and\noun{ Conn}\textsubscript{\noun{C}}\noun{($\mathcal{S}$) }are
\noun{PSPACE}-complete, and there exist \emph{CNF\textsubscript{C}($\mathcal{S}$)}-formulas
$\phi$, such that the diameter of $G(\phi)$ is exponential in the
number of variables.
\end{cor}

\section{\label{sec:stru}Safely Tight Sets of Relations: Structure~and~Algorithms}

For safely tight sets of relations, the solution graphs possess certain
structural properties that guarantee a linear diameter, and allow
for P-algorithms for $st$-connectivity, and coNP-algorithms for connectivity.
We start with safely componentwise bijunctive relations.
\begin{lem}
\label{lem:l43}\emph{\citep[corrected from][Lemma 4.3]{gop}} ~Let
$\mathcal{S}$ be a set of safely componentwise bijunctive relations
and $\varphi$ a \noun{CNF\textsubscript{C}($\mathcal{S}$)}-formula.
If $\boldsymbol{a}$ and \textbf{$\boldsymbol{b}$} are two solutions
of $\varphi$ that lie in the same component of G($\varphi$), then
$d_{\varphi}(\boldsymbol{a},\boldsymbol{b})=|\boldsymbol{a}-\boldsymbol{b}|$,
i.e., no distance expands.\end{lem}
\begin{proof}
Consider first the special case in which every relation in $\mathcal{S}$
is bijunctive. In this case, $\varphi$ is equivalent to a 2-CNF-formula
and so the space of solutions of $\varphi$ is closed under majority.
We show that there is a path in $G(\varphi)$ from $\boldsymbol{a}$
to $\boldsymbol{b}$ such that along the path only the assignments
on variables with indices from the set $D=\{i|a_{i}\neq b_{i}\}$
change. This implies that the shortest path is of length $|D|$ by
induction on $|D|$. Consider any path $\boldsymbol{a}\rightarrow\boldsymbol{u}^{1}\rightarrow\cdots\rightarrow\boldsymbol{u}^{r}\rightarrow\boldsymbol{b}$
in $G(\varphi)$. We construct another path by replacing $u_{i}$
by $v_{i}=\mathrm{maj}(a,u_{i},b)$ for $i=1,\ldots,r$ and removing
repetitions. This is a path because for any $i$ $v^{i}$ and $v^{i+1}$
differ in at most one variable. Furthermore, $v^{i}$ agrees with
$\boldsymbol{a}$ and $\boldsymbol{b}$ for every $i$ for which $a_{i}=b_{i}$.
Therefore, along this path only variables in $D$ are flipped.

For the general case, we show that every component $F$ of G($\varphi$)
is the solution space of a 2-CNF-formula \emph{$\varphi$}. Let $R\in\mathcal{S}$
be a safely componentwise bijunctive relation. Then any relation corresponding
to a clause in $\varphi$ (see \prettyref{def:cnf-1}) of the form
$R(x_{1},\ldots,x_{k})$ consists of bijunctive components $R_{1},\ldots,R_{m}$.
The projection of $F$ onto $x_{1},\ldots,x_{k}$ is itself connected
and must satisfy $R$. Hence it lies within one of the components
$R_{1},\ldots,R_{m}$; assume it is $R_{1}$. We replace $R(x_{1},\ldots,x_{k})$
by $R_{1}(x_{1},\ldots,x_{k})$. Call this new formula $\varphi_{1}$.
$G(\varphi_{1})$ consists of all components of G($\varphi$) whose
projection on $x_{1},\ldots,x_{k}$ lies in $R_{1}$. We repeat this
for every clause. Finally we are left with a formula $\varphi'$ over
a set of bijunctive relations. Hence $\varphi'$ is bijunctive and
$G(\varphi')$ is a component of $G(\varphi)$. So the claim follows
from the bijunctive case.\end{proof}
\begin{cor}
\label{lem:c44}\emph{\citep[corrected from][Corollary 4.4]{gop}}
~Let \emph{$\mathcal{S}$} be set of safely componentwise bijunctive
relations. Then 
\begin{enumerate}
\item for every $\phi\in$\emph{CNF\textsubscript{C}($\mathcal{S}$)} with
$n$ variables, the diameter of each component of $G(\phi)$ is bounded
by $n$.
\item \emph{\noun{st-Conn}}\emph{\textsubscript{C}($\mathcal{S}$) }is
in \emph{P}.
\item \emph{\noun{Conn}}\emph{\textsubscript{C}($\mathcal{S}$)} is in
\emph{coNP}.
\end{enumerate}
\end{cor}
The proof of this corollary in \citep{gop} is correct. 

We now turn to safely OR-free relations; we need the following definition:
\begin{defn}
\label{def:lm}We define the \emph{coordinate-wise partial order}
$\leq$ on Boolean vectors as follows: $\boldsymbol{a}\leq\boldsymbol{b}$
if $a_{i}\leq b_{i},$ $\forall i$. A \emph{monotone path} between
two solutions $\boldsymbol{a}$ and $\boldsymbol{b}$ is a path $\boldsymbol{a}\rightarrow\boldsymbol{u}^{1}\rightarrow\cdots\rightarrow\boldsymbol{u}^{r}\rightarrow\boldsymbol{b}$
in the solution graph such that $\boldsymbol{a}\leq\boldsymbol{u}^{1}\leq\cdots\leq\boldsymbol{u}^{r}\leq\boldsymbol{b}$.
A solution is \emph{locally minimal} if it has no neighboring solution
that is smaller than it.\end{defn}
\begin{lem}
\label{lem:l45}\emph{\citep[corrected from][Lemma 4.5]{gop} }~Let
$\mathcal{S}$ be a set of safely OR-free relations and $\varphi$
a CNF\textsubscript{C}($\mathcal{S}$)-formula. Every component of
$G(\varphi)$ contains a minimum solution with respect to the coordinatewise
order; moreover, every solution is connected to the minimum solution
in the same component via a monotone path.\end{lem}
\begin{proof}
We will show that there is exactly one such assignment in each component
of $G(\varphi)$. Suppose there are two distinct locally minimal assignments
$\boldsymbol{u}$ and $\boldsymbol{u'}$ in some component of $G(\varphi)$.
Consider the path between them where the maximum Hamming weight of
assignments on the path is minimized. If there are many such paths,
pick one where the smallest number of assignments have the maximum
Hamming weight. Denote this path by $\boldsymbol{u}=\boldsymbol{u^{1}}\rightarrow\boldsymbol{u^{2}}\rightarrow\dots\rightarrow\boldsymbol{u^{r}}=\boldsymbol{u'}$.
Let $\boldsymbol{u^{i}}$ be an assignment of largest Hamming weight
in the path. Then $\boldsymbol{u^{i}}\not=\boldsymbol{u}$ and $\boldsymbol{u^{i}}\not=\boldsymbol{u'}$,
since $\boldsymbol{u}$ and $\boldsymbol{u'}$ are locally minimal.
The assignments $\boldsymbol{u^{i-1}}$ and $\boldsymbol{u^{i+1}}$
differ in exactly 2 variables, say, in $x_{1}$ and $x_{2}$. So $\{u_{1}^{i-1}u_{2}^{i-1},~u_{1}^{i}u_{2}^{i},~u_{1}^{i+1}u_{2}^{i+1}\}=\{01,11,10\}$.
Let $\boldsymbol{\hat{u}}$ be such that $\hat{u}_{1}=\hat{u}_{2}=0$,
and $\hat{u}_{i}=u_{i}$ for $i>2$. If $\boldsymbol{\hat{u}}$ is
a solution, then the path $\boldsymbol{u^{1}}\rightarrow\boldsymbol{u^{2}}\rightarrow\dots\rightarrow\boldsymbol{u^{i}}\rightarrow\boldsymbol{\hat{u}}\rightarrow\boldsymbol{u^{i+1}}\rightarrow\dots\rightarrow\boldsymbol{u^{r}}$
contradicts the way we chose the original path. Therefore, $\boldsymbol{\hat{u}}$
is not a solution. This means that there is a clause that is violated
by it, but is satisfied by $\boldsymbol{u^{i-1}}$, $\boldsymbol{u^{i}}$,
and $\boldsymbol{u^{i+1}}$. So the relation corresponding to that
clause is not OR-free, thus \emph{$\mathcal{S}$ }must have contained
some not\emph{ }safely OR-free relation.

The unique locally minimal solution in a component is its minimum
solution, because starting from any other assignment in the component,
it is possible to keep moving to neighbors that are smaller, and the
only time it becomes impossible to find such a neighbor is when the
locally minimal solution is reached. Therefore, there is a monotone
path from any satisfying assignment to the minimum in that component.\end{proof}
\begin{cor}
\label{lem:c46}\emph{\citep[corrected from][Corollary 4.6]{gop}
}~Let \emph{$\mathcal{S}$} be a set of safely OR-free relations.
Then
\begin{enumerate}
\item for every \emph{CNF\textsubscript{C}($\mathcal{S}$)}-formula $\phi$
with $n$ variables, the diameter of each component of $G(\phi)$
is bounded by $2n$.
\item \emph{\noun{st-Conn}}\emph{\textsubscript{C}($\mathcal{S}$)} is
in \emph{P}.
\item \emph{\noun{Conn}}\emph{\textsubscript{C}($\mathcal{S}$)} is in
\emph{coNP}.
\end{enumerate}
\end{cor}
The proof of this corollary in \citep{gop} is correct. Safely NAND-free
relations are symmetric to safely OR-free relations, so that we have
the following corollary which completes the proof of the dichotomy
(\prettyref{thm:dich}).
\begin{cor}
\label{lem:c47}\emph{\citep[corrected from][Corollary 4.7]{gop}
}~Let S be a safely tight set of relations. Then 
\begin{enumerate}
\item for every $\phi\in$\emph{CNF\textsubscript{C}($\mathcal{S}$)} with
$n$ variables, the diameter of each component of $G(\phi)$ is bounded
by $2n$.
\item \emph{\noun{st-}}\noun{Conn}\emph{\textsubscript{C}($\mathcal{S}$)}
is in \emph{P}.
\item \noun{Conn}\emph{\textsubscript{C}($\mathcal{S}$)} is in \emph{coNP}.
\end{enumerate}
\end{cor}

\section{\label{sec:cpss}CPSS Sets of Relations: A~Simple~Algorithm~for~Connectivity}

The rest of this chapter is devoted to the complexity of \noun{Conn}\textsubscript{C}($\mathcal{S}$).
In this section we cover the tractable case; we show that for CPSS
sets $\mathcal{S}$ of relations, for every CNF\textsubscript{C}($\mathcal{S}$)-formula
whose solution graph is disconnected, already the projection to the
variables of some constraint is disconnected. We then use this property
to derive a simple algorithm for \noun{Conn}\textsubscript{C}($\mathcal{S}$)
(Gopalan et al.\ had given much more complicated algorithms in Lemmas
4.9, 4.10 and 4.13 of \citep{gop}).
\begin{defn}
\label{def:dcp}A set $\mathcal{S}$ of logical relations is \emph{constraint-projection
separating }(\emph{CPS}), if every CNF\textsubscript{C}($\mathcal{S}$)-formula
$\phi$ whose solution graph $G(\phi)$ is disconnected contains a
constraint $C_{i}$ s.t. $G(\phi_{i})$ is disconnected, where $\phi_{i}$
is the projection of $\phi$ to $\mathrm{Var}(C_{i})$.
\end{defn}
\begin{wrapfigure}{r}{0.27\columnwidth}%
\includegraphics[scale=0.65]{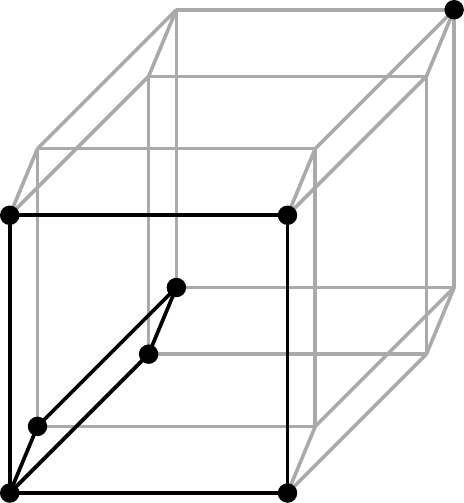}\end{wrapfigure}%
For example, $\mathcal{S}=\{x\vee\overline{y}\}$ is CPS (for the
proof see \prettyref{lem:prjs}); so, e.g. for the CNF\textsubscript{C}($\mathcal{S}$)-formula
$\left(x\vee\overline{y}\right)\wedge\left(y\vee\overline{z}\right)\wedge\left(z\vee\overline{x}\right),$
the projections to $\{x,y\}$, $\{y,z\}$ and $\{z,x\}$ are all disconnected.

In contrast, $\mathcal{S}'=\{x\vee\overline{y}\vee\overline{z}\}$
is not CPS: For example, the CNF\textsubscript{C}($\mathcal{S}'$)-formula
\[
\left(x\vee\overline{y}\vee\overline{z}\right)\wedge\left(y\vee\overline{z}\vee\overline{w}\right)\wedge\left(z\vee\overline{w}\vee\overline{x}\right)\wedge\left(w\vee\overline{x}\vee\overline{y}\right)
\]
is disconnected (see the graph on the right), but the projection to
any three variables is connected.

We cannot provide an algorithm to decide for an arbitrary set of relations
if it is CPS, but we will determine exactly which Schaefer sets are
CPS, and also exhibit classes of non-Schaefer sets that are CPS.

For IHSB$-$, IHSB$+$ and bijunctive sets of relations we can prove
even stronger properties in the next two lemmas:
\begin{lem}
\label{lem:prjs i}Let $\mathcal{S}$ be a set of IHSB$-$ (IHSB$+$)
relations and $\phi$ a \noun{CNF\textsubscript{C}($\mathcal{S}$)}-formula.
Then for any two components of $G(\phi)$, there is some constraint
$C_{i}$ of $\phi$ s.t.$\mbox{\,}$their images in the projection
$\phi_{i}$ of $\phi$ to $\mathrm{Var}(C_{i})$ are disconnected
in $G(\phi_{i})$.\end{lem}
\begin{proof}
We prove the IHSB$-$ case, the IHSB$+$ case is analogous. Consider
any two components $A$ and $B$ of $\phi$. Since every IHSB$-$
relation is safely OR-free, there is a locally minimal solution $\boldsymbol{a}$
in $A$ and a locally minimal solution $\boldsymbol{b}$ in $B$ by
\prettyref{lem:l45}. Let $U$ and $V$ be the sets of variables that
are assigned 1 in $\boldsymbol{a}$ and \textbf{$\boldsymbol{b}$},
resp.~At least one of the sets $U'=U\setminus V$ or $V'=V\setminus U$
is not empty, assume it is $U'$. Then for every $x_{1}\in U'$ there
must be a clause $x_{1}\vee\overline{x}_{2}$ with $x_{2}\in U$ since
$\boldsymbol{a}$ is locally minimal, and also $x_{2}$ must be from
$U'$, else $\boldsymbol{b}$ would not be satisfying.

But then for $x_{2}$ there must be also some variable $x_{3}\in U'$
and a clause $x_{2}\vee\overline{x}_{3}$, and we can add the clause
$x_{1}\vee\overline{x}_{3}$ to $\phi$ without changing its value.
Continuing this way, we will find a cycle, i.e.~a clause $x_{i}\vee\overline{x}_{i+1}$
with $x_{i+1}=x_{j}$, $j<i$. But then we already have $x_{j}\vee\overline{x}_{i}$
added, thus $(s_{i},s_{j})\in\{(0,0),(1,1)\}$ for any solution $\boldsymbol{s}$
of $\phi$, and there must be some constraint $C_{i}$ with both $x_{i}$
and $x_{j}$ occurring in it (the $C_{i}$ in which the original $x_{i}\vee\overline{x}_{j}$
appeared), and thus the projections of $A$ and $B$ to $\mathrm{Var}(C_{i})$
are disconnected in $G(\phi_{i})$.\end{proof}
\begin{lem}
\label{lem:prjs b}Let $\mathcal{S}$ be a set of bijunctive relations
and $\phi$ a \noun{CNF\textsubscript{C}($\mathcal{S}$)}-formula.
Then for any two components of $G(\phi)$, there is some constraint
$C_{i}$ of $\phi$ s.t.$\mbox{\,}$their images in the projection
$\phi_{i}$ of $\phi$ to $\mathrm{Var}(C_{i})$ are disconnected
in $G(\phi_{i})$.\end{lem}
\begin{proof}
The proof is similar to the last one. Consider any two components
$A$ and $B$ of $\phi$ and two solutions $\boldsymbol{a}$ in $A$
and $\boldsymbol{b}$ in $B$ that are at minimum Hamming distance.
Let $L$ be the set of literals that are assigned 1 in \textbf{$\boldsymbol{a}$},
but assigned 0 in\textbf{ $\boldsymbol{b}$}. Then for every $l_{1}\in L$
that is assigned 1 in\textbf{ $\boldsymbol{a}$}, there must be a
clause equivalent to $l_{1}\vee\overline{l}_{2}$ in $\phi$ s.t.$\mbox{\,}$$l_{2}$
is also assigned 1 in\textbf{ $\boldsymbol{a}$}, else the variable
corresponding to $l_{1}$ could be flipped in \textbf{$\boldsymbol{a}$},
and the resulting vector would be closer to \textbf{$\boldsymbol{b}$},
contradicting our choice of $\boldsymbol{a}$ and \textbf{$\boldsymbol{b}$}.
Also, $l_{2}$ must be assigned 0 in \textbf{$\boldsymbol{b}$}, i.e.~$l_{2}\in L$,
else $\boldsymbol{b}$ would not be satisfying. 

But then for $l_{2}$ there must be also some literal $l_{3}\in L$
that is assigned 1 in\textbf{ }$\boldsymbol{a}$ and a clause equivalent
to $l_{2}\vee\overline{l}_{3}$ in $\phi$, and we can add the clause
$l_{1}\vee\overline{l}_{3}$ to $\phi$ without changing its value.
Continuing this way, we will find a cycle, i.e.~a clause equivalent
to $l_{n}\vee\overline{l}_{n+1}$ with $l_{n+1}=l_{m}$, $m<n$. But
then we already have $l_{m}\vee\overline{l}_{n}$ added, thus if $x_{i}$
and $x_{j}$ are the variables corresponding to $l_{n}$ resp.~$l_{m}$,
then $(s_{i},s_{j})\in\{(0,1),(1,0)\}$ (if $l_{n}$ and $l_{m}$
were both positive or both negative), or $(s_{i},s_{j})\in\{(0,0),(1,1)\}$
(otherwise), for any solution $\boldsymbol{s}$ of $\phi$. Also,
there must be some constraint $C_{i}$ with both $x_{i}$ and $x_{j}$
occurring in it (the constraint in which the clause equivalent to
$l_{n}\vee\overline{l}_{m}$ appeared), and thus the projections of
$A$ and $B$ to $\mathrm{Var}(C_{i})$ are disconnected in $G(\phi_{i})$.\end{proof}
\begin{lem}
\label{lem:prjs}Every set $\mathcal{S}$ of  safely componentwise
bijunctive (safely componentwise IHSB$-$,  safely componentwise IHSB$+$,
affine) relations is constraint-projection separating.\end{lem}
\begin{proof}
The affine case follows from the  safely componentwise bijunctive
case since every affine relation is safely componentwise bijunctive
by \prettyref{lem:l42}.

If the relation corresponding to some $C_{i}$ is disconnected, and
there is more than one component of this relation for which $\phi$
has solutions with the variables of $C_{i}$ assigned values in that
component, the projection of $\phi$ to $\mathrm{Var}(C_{i})$ must
be disconnected in $G(\phi_{i})$.

So assume that for every constraint $C_{i}$, $\phi$ only has solutions
in which the variables of $C_{i}$ are assigned values in one component
$P_{i}$ of the relation corresponding to $C_{i}$. Then we can replace
every $C_{i}$ with $P_{i}$ to obtain an equivalent formula $\phi'$.
Since $\mathcal{S}$ is  safely componentwise bijunctive (safely componentwise
IHSB$-$,  safely componentwise IHSB$+$), each $P_{i}$ is bijunctive
(IHSB$-$, IHSB$-$), and thus so is $\phi'$, and the statement follows
from Lemmas \prettyref{lem:prjs i} and \prettyref{lem:prjs b}. 
\end{proof}
We are now ready to show how connectivity can be solved in polynomial
time for CPSS sets of relations:
\begin{lem}
\label{lem:alg}If a finite set $\mathcal{S}$ of relations is constraint-projection
separating, \noun{Conn}\textsubscript{\emph{C}}\noun{($\mathcal{S}$)
}is in $\mathrm{P^{NP}}$.\footnote{$\mathrm{P^{NP}}=\mathrm{P^{SAT}}$ is the class of languages decidable
by a deterministic polynomial-time Turing machine with oracle-access
to an NP-complete problem, e.g. \noun{Sat}.} If $\mathcal{S}$ is also Schaefer, \noun{Conn}\emph{\textsubscript{C}}\noun{($\mathcal{S}$)}
is in \noun{P.}\end{lem}
\begin{proof}
For any \noun{CNF\textsubscript{C}($\mathcal{S}$)}-formula $\phi$,
connectivity of $G(\phi)$ can be decided as follows:
\begin{itemize}[label= ]
\item For every constraint $C_{i}$ of $\phi$, obtain the projection $\phi_{i}$
of $\phi$ to the variables $\boldsymbol{x}^{i}$ occurring in $C_{i}$
by checking for every assignment $\boldsymbol{a}$ of $\boldsymbol{x}^{i}$
whether $\phi[\boldsymbol{x}^{i}/\boldsymbol{a}]$ is satisfiable.
Then $G(\phi)$ is connected iff for no $\phi_{i}$, $G(\phi_{i})$
is disconnected.
\end{itemize}
If \noun{$\mathcal{S}$} is Schaefer, every projection can be computed
in polynomial time, else we use a \noun{Sat}-oracle. Connectivity
of every $G(\phi_{i})$ can be checked in constant time.

If $G(\phi)$ is disconnected, some $G(\phi_{i})$ is disconnected
since $\phi$ is CPS by \prettyref{lem:prjs} below. If some $G(\phi_{i})$
is disconnected, $G(\phi)$ clearly is also disconnected.\end{proof}
\begin{cor}
\label{cor:alg-1}If a finite set $\mathcal{S}$ of relations is CPSS,
\noun{Conn}\textsubscript{\noun{C}}\noun{($\mathcal{S}$)} is polynomial-time
solvable. 
\end{cor}

\section{The Last Piece: \noun{coNP}-Hardness for\noun{ }Connectivity}

It remains to determine the complexity of\noun{ Conn}\textsubscript{\noun{C}}
for safely tight sets of relations that are not CPSS. For non-Schaefer
sets this was done already by Gopalan et al.:
\begin{lem}
\label{lem:l48}\emph{\citep[corrected from][Lemma 4.8]{gop}} ~For
S safely tight, but not Schaefer, \noun{Conn}\textsubscript{C}($\mathcal{S}$)
is \noun{coNP}-complete.\end{lem}
\begin{proof}
The problem \textsc{Another-Sat}$({\cal S})$ is: given a formula
$\varphi$ in CNF\textsubscript{C}($\mathcal{S}$) and a solution
$\boldsymbol{s}$, does there exist a solution $\boldsymbol{t}\not=\boldsymbol{s}$?
Juban (\citep{Juban99}, Theorem 2) shows that if ${\cal S}$ is not
Schaefer, then \textsc{Another-Sat} is NP-complete. He also shows
(\citep{Juban99}, Corollary 1) that if ${\cal S}$ is not Schaefer,
then the relation $x\neq y$ is expressible as a CNF\textsubscript{C}($\mathcal{S}$)-formula.

Since ${\cal S}$ is not Schaefer, \textsc{Another-Sat}$({\cal S})$
is NP-complete. Let $\varphi,\boldsymbol{s}$ be an instance of \textsc{Another-Sat}
on variables $x_{1},\dots,x_{n}$. We define a CNF\textsubscript{C}($\mathcal{S}$)
formula $\psi$ on the variables $x_{1},\dots,x_{n},y_{1},\dots,y_{n}$
as 
\[
\psi(x_{1},\dots,x_{n},y_{1},\dots,y_{n})=\varphi(x_{1},\dots,x_{n})\wedge_{i}(x_{i}\neq y_{i})
\]
It is easy to see that $G(\psi)$ is connected if and only if $\boldsymbol{s}$
is the unique solution to $\varphi$. 
\end{proof}
We are now left with the case of Horn (dual Horn) sets of relations
containing at least one relation that is not safely componentwise
IHSB$-$ (not safely componentwise IHSB$+$).

For one such set, namely $\{x\vee\overline{y}\vee\overline{z}\}$,
Makino, Tamaki, and Yamamoto showed in 2007 that \noun{Conn}\textsubscript{C}
is coNP-complete \citep{Makino:2007:BCP:1768142.1768162}. Consequently,
Gopalan et al.~conjectured that \noun{Conn}\textsubscript{C} is
coNP-complete for any such set, and already suggested a way for proving
that: One had to show that \noun{Conn}$_{C}$($\{M\}$) for the relation
$M=\left(x\vee\overline{y}\vee\overline{z}\right)\wedge\left(\overline{x}\vee z\right)$
is coNP-hard \citep{gop}. We will prove this in \prettyref{lem:m har}
by a reduction from the complement of a satisfiability problem.

Gopalan et al.~stated (without giving the proof) that they could
show that $M$ is structurally expressible from every such set, using
a similar reasoning as in the proof of their structural expressibility
theorem (Lemma 3.4 in \citep{gop}). We give a quite different proof
in \prettyref{lem:exp m}, that shows that $M$ actually is expressible
as a CNF\textsubscript{C}($\mathcal{S}$)-formula, which is of course
a structural expression.

The proofs of \prettyref{lem:m har} and \prettyref{lem:exp m} are
arguably the most intricate part of this thesis and will be adapted
for the no-constants and quantified cases in the next chapter.

\subsection{Connectivity of Horn Formulas}

In this subsection, we introduce terminology and develop tools we
will need for the proofs of \prettyref{lem:m har} and \prettyref{lem:exp m}.
\begin{defn}
Clauses with only one literal are called \emph{unit clauses} (\emph{positive}
if the literal is positive, \emph{negative} otherwise). Clauses with
only negative literals are \emph{restraint clauses}, and the sets
of variables occurring in restraint clauses are \emph{restraint sets}.
Clauses having one positive and one or more negative literals are
\emph{implication clauses}. Implication clauses with two or more negative
literals are \emph{multi-implication clauses.}

A variable\emph{ $x$ is implied by a set of variables $U$,} if setting
all variables from $U$ to 1 forces $x$ to be 1 in any satisfying
assignment. We write Imp($U$) for the set of variables implied by
$U$, we abbreviate Imp($\{x\}$) as Imp($x$). We simply say that
$x$ \emph{is} \emph{implied,} if $x\in\mathrm{Imp}(U\setminus\{x\})$
for some $U$. Note that $U\subseteq\mathrm{Imp}(U)$ for all sets
$U$.

$U$ is \emph{self-implicating} if every $x\in U$ is implied by $U\setminus\{x\}$.
$U$ is \emph{maximal self-implicating}, if further $U=\mathrm{Imp}(U)$.\end{defn}
\begin{rem}
\label{hyp}A Horn formula can be represented by a directed hypergraph
with hyperedges of head-size one as follows: For every variable, there
is a node, for every implication clause $y\vee\overline{x}_{1}\vee\cdots\vee\overline{x}_{k}$,
there is a directed hyperedge from $x_{1},\ldots,x_{k}$ to $y$,
for every restraint clause $\overline{x}_{1}\vee\cdots\vee\overline{x}_{k}$,
there is a directed hyperedge from $x_{1},\ldots,x_{k}$ to a special
node labeled ``false'', and for every positive unit clause $x$,
there is a directed hyperedge from a special node labeled ``true''
to $x$. 

We draw the directed hyperedges as joining lines, e.g., $x\vee\overline{y}\vee\overline{z}=$
$\vcenter{\hbox{\includegraphics[scale=0.4]{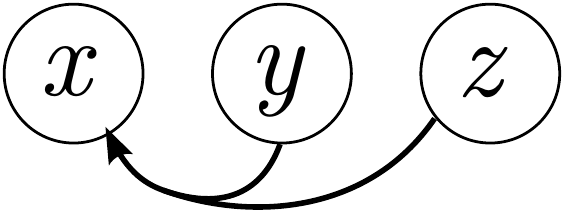}}}$. For simplicity,
we omit the ``false'' and ``true'' nodes in the drawings and let
the corresponding hyperedges end, resp. begin, in the void. 

For an introduction to directed hypergraphs see e.g. \citep{gallo1993directed}.
General CNF-formulas can be represented by general directed hypergraphs,
see e.g. \citep{andersen2001easy}.\end{rem}
\begin{lem}
\label{lem:loc min}The solution graph of a Horn formula $\phi$ without
positive unit clauses is disconnected iff $\phi$ has a locally minimal
nonzero solution.\end{lem}
\begin{proof}
This follows from \prettyref{lem:l45} since the all-zero vector is
a solution of every Horn formula without positive unit clauses, and
Horn formulas are safely OR-free by \prettyref{lem:l42}.\end{proof}
\begin{lem}
\label{lem:horn conn}For every Horn formula $\phi$ without positive
unit clauses, there is a bijection correlating each connected component
$\phi_{i}$ with a maximal self-implicating set $U_{i}$ containing
no restraint set; $U_{i}$ consists of the variables assigned 1 in
the minimum solution of $\phi_{i}$ (the ``lowest'' component is
correlated with the empty set).\end{lem}
\begin{proof}
Let $\phi_{i}$ be a connected component of $\phi$ with minimum solution
$\boldsymbol{s}$, and let $U$ be the set of variables assigned 1
in $\boldsymbol{s}$. Since $\boldsymbol{s}$ is locally minimal,
flipping any variable $x_{i}$ from $U$ to 0 results in a vector
that is no solution, so there must be a clause in $\phi$ prohibiting
that $x_{i}$ is flipped. Since $\phi$ contains no positive unit-clauses,
each $x_{i}\in U$ must appear as the positive literal in an implication
clause with also all negated variables from $U$. It follows that
$U$ is self-implicating. Also, $U$ must be maximal self-implicating
and can contain no restraint set, else $\boldsymbol{s}$ were no solution.

Conversely, let $U$ be a maximal self-implicating set containing
no restraint set. Then the vector $\boldsymbol{s}$ with all variables
from $U$ assigned 1, and all others 0, is a locally minimal solution:
All implication clauses $\overline{y}_{1}\vee\cdots\vee\overline{y}_{k}\vee x$
with some $y_{i}\notin U$ are satisfied since $y_{i}=0$, and for
the ones with all $y_{i}\in U$, also $x\in U$ holds because $U$
is maximal, so these are satisfied since $x=1$. All restraint clauses
are satisfied since $U$ contains no restraint set. $\boldsymbol{s}$
is locally minimal since every variable assigned 1 is implied by $U$,
so that any vector with one such variable flipped to 0 is no solution.
By Lemma 4.5 of \citep{gop}, every connected component has a unique
locally minimal solution, so $\boldsymbol{s}$ is the minimum solution
of some component.\end{proof}
\begin{cor}
\label{cor:horn conn}The solution graph of a Horn formula $\phi$
without positive unit clauses is disconnected iff $\phi$ has a non-empty
maximal self-implicating set containing no restraint set.\end{cor}
\begin{lem}
\label{lem:in}Let $R_{1}$ and $R_{2}$ be two connected components
of a Horn relation $R$ with minimum solutions $\boldsymbol{u}$ and
$\boldsymbol{v}$, resp., and let $U$ and $V$ be the sets of variables
assigned 1 in $\boldsymbol{u}$ and $\boldsymbol{v}$, resp. If then
$U\subsetneq V$, no vector $\boldsymbol{a}\in R_{1}$ has all variables
from $V$ assigned 1.\end{lem}
\begin{proof}
For the sake of contradiction, assume $\boldsymbol{a}\in R_{1}$ has
all variables from $V$ assigned 1. Then $\boldsymbol{a\wedge v}=\boldsymbol{v}$,
where $\boldsymbol{\wedge}$ is applied coordinate-wise. Consider
a path from $\boldsymbol{u}$ to $\boldsymbol{a}$, $\boldsymbol{u}\rightarrow\boldsymbol{w}^{1}\rightarrow\cdots\rightarrow\boldsymbol{w}^{k}\rightarrow\boldsymbol{a}$.
Since $U\subsetneq V$, we have $\boldsymbol{u\wedge v}=\boldsymbol{u}$,
so we can construct a path from $\boldsymbol{u}$ to $\boldsymbol{v}$
by replacing each $\boldsymbol{w}^{i}$ by $\boldsymbol{w}^{i}\boldsymbol{\wedge}\boldsymbol{v}$
in the above path, and removing repetitions. Since $R$ is Horn, it
is closed under $\wedge$ (see \prettyref{lem:clos}), so all vectors
of the constructed path are in $R$. But $\boldsymbol{u}$ and $\boldsymbol{v}$
are not connected in $R$, which is a contradiction.\end{proof}
\begin{defn}
\label{def:nor}For a Horn formula $\phi$, let $\nu(\phi)$ be the
formula obtained from $\phi$ by recursively applying the following
simplification rules as long as one is applicable; it is easy to check
that the operations are equivalent transformations, and that the recursion
must terminate:
\selectlanguage{british}%
\begin{enumerate}[label=({\alph*})]
\item \foreignlanguage{english}{The constants 0 and 1 are eliminated in
the obvious way.}
\selectlanguage{english}%
\item Multiple occurrences of some variable in a clause are eliminated in
the obvious way.
\item \emph{\label{enu:red imp}}If for some implication clause $c=x\vee\overline{y}_{1}\vee\cdots\vee\overline{y}_{k}$
($k\geq1$), $x$ is already implied by $\{y_{1},\ldots,y_{k}\}$
via other clauses, $c$ is removed.\emph{}\\
{\small{}\hspace*{2ex}}E.g., if there was a clause $x\vee\overline{z}_{1}\vee\cdots\vee\overline{z}_{l}$
with $\{z_{1},\ldots,z_{l}\}\subseteq\{y_{1},\ldots,y_{k}\}$, or
clauses $q\vee\overline{z}_{1}\vee\cdots\vee\overline{z}_{l}$ and
$x\vee\overline{q}$, \emph{$c$} would be removed. Which clauses
are removed by this rule may be random; e.g., for the formula $\left(\overline{x}\vee y\right)\wedge\left(\overline{x}\vee z\right)\wedge\left(\overline{z}\vee y\right)\wedge\left(\overline{y}\vee z\right)$,
$\overline{x}\vee y$ \emph{or} $\overline{x}\vee z$ would be removed:\\
\includegraphics[scale=0.55]{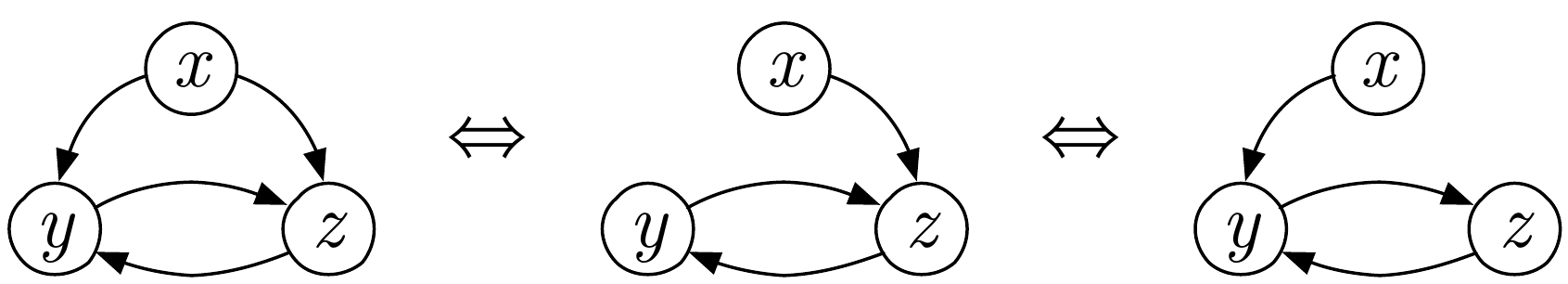}
\item \emph{\label{enu:imp imp}}If for some implication clause $c=x\vee\overline{y}_{1}\vee\cdots\vee\overline{y}_{k}$
($k\geq1$), Imp(Var($c$)) contains a restraint set, $c$ is replaced
by $\overline{y}_{1}\vee\cdots\vee\overline{y}_{k}$.\\
{\small{}\hspace*{3ex}}E.g., if there is a clause $\overline{r}_{1}\vee\cdots\vee\overline{r}_{l}$
with $\{r_{1},\ldots,r_{l}\}\subseteq\{x,y_{1},\ldots,y_{k}\}$, or
if there are clauses $q_{1}\vee\overline{r}_{1}\vee\cdots\vee\overline{r}_{l}$,
$q_{2}\vee\overline{r}_{1}\vee\cdots\vee\overline{r}_{l}$ and $\overline{q}_{1}\vee\overline{q}_{2}$.
E.g., in the formula $\left(x\vee\overline{y}\vee\overline{z}\right)\wedge\left(w\vee\overline{y}\right)\wedge\left(\overline{w}\vee\overline{x}\right)$,
$x\vee\overline{y}\vee\overline{z}$ is replaced by $\overline{y}\vee\overline{z}$:\\
\includegraphics[scale=0.55]{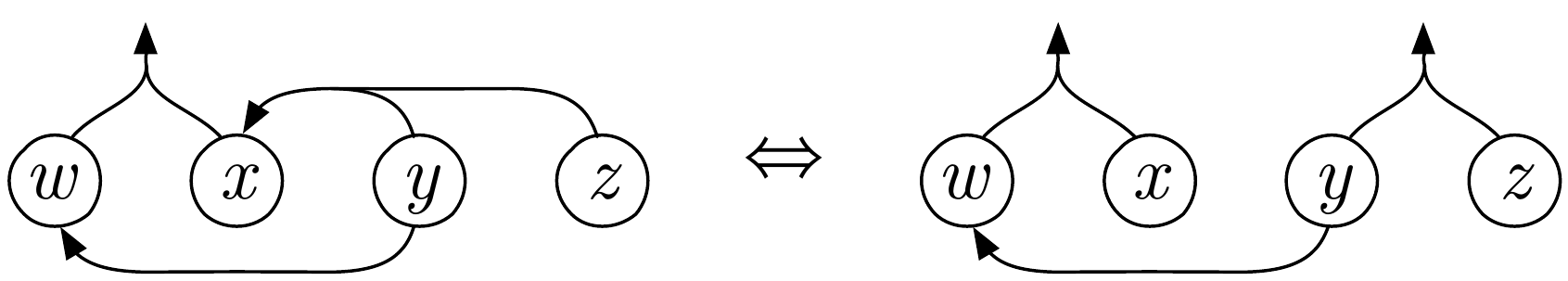}\emph{}
\item \emph{\label{enu:red bra}}If for some multi-implication clause $c=x\vee\overline{y}_{1}\vee\cdots\vee\overline{y}_{k}$
($k\geq2$), or for some restraint clause $d=\overline{y}_{1}\vee\cdots\vee\overline{y}_{k}$,
some $y_{i}\in\{y_{1},\ldots,y_{k}\}$ is implied by $\{y_{1},\ldots,y_{k}\}\setminus\{y_{i}\}$,
the literal $\overline{y}_{i}$ is removed from $c$ resp. $d$.\emph{}\\
{\small{}\hspace*{3ex}}Which literals are removed by this rule may
be random, as in the following example:\\
\includegraphics[scale=0.63]{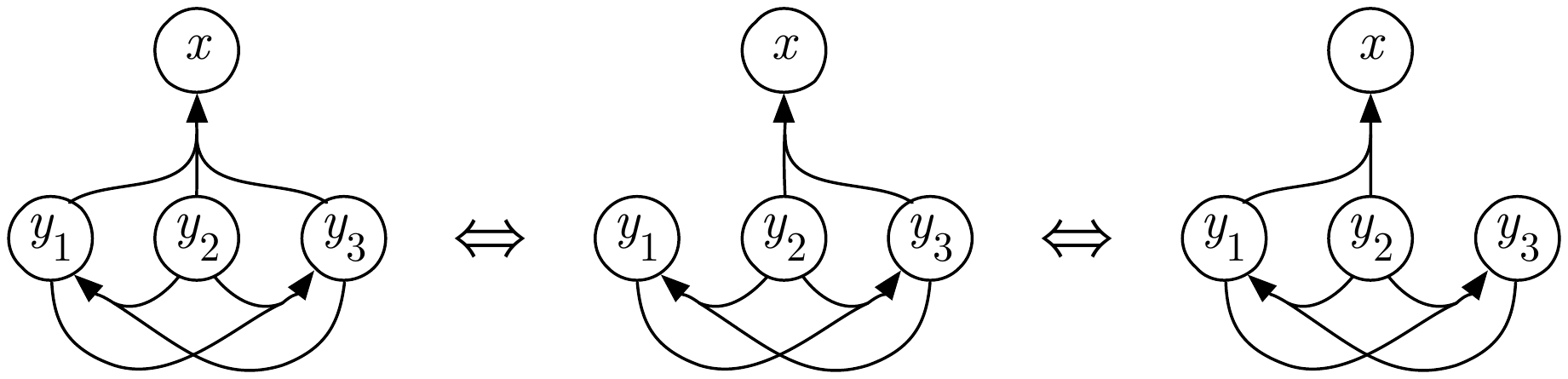}
\end{enumerate}
\selectlanguage{english}%
For a Horn relation $R$, let $\nu(R)=\nu(\phi)$ with some Horn formula
$\phi$ representing $R$.
\end{defn}

\subsection{Reduction from Satisfiability}
\begin{lem}
\noun{\label{lem:m har} Conn}\textsubscript{\noun{C}}\noun{(}\textup{$\left\{ M\right\} $}\noun{)}
with $M=\left(x\vee\overline{y}\vee\overline{z}\right)\wedge\left(\overline{x}\vee z\right)$
is \noun{coNP}-hard.\end{lem}
\begin{proof}
\begin{figure}[!h]
\begin{centering}
\includegraphics[scale=0.44]{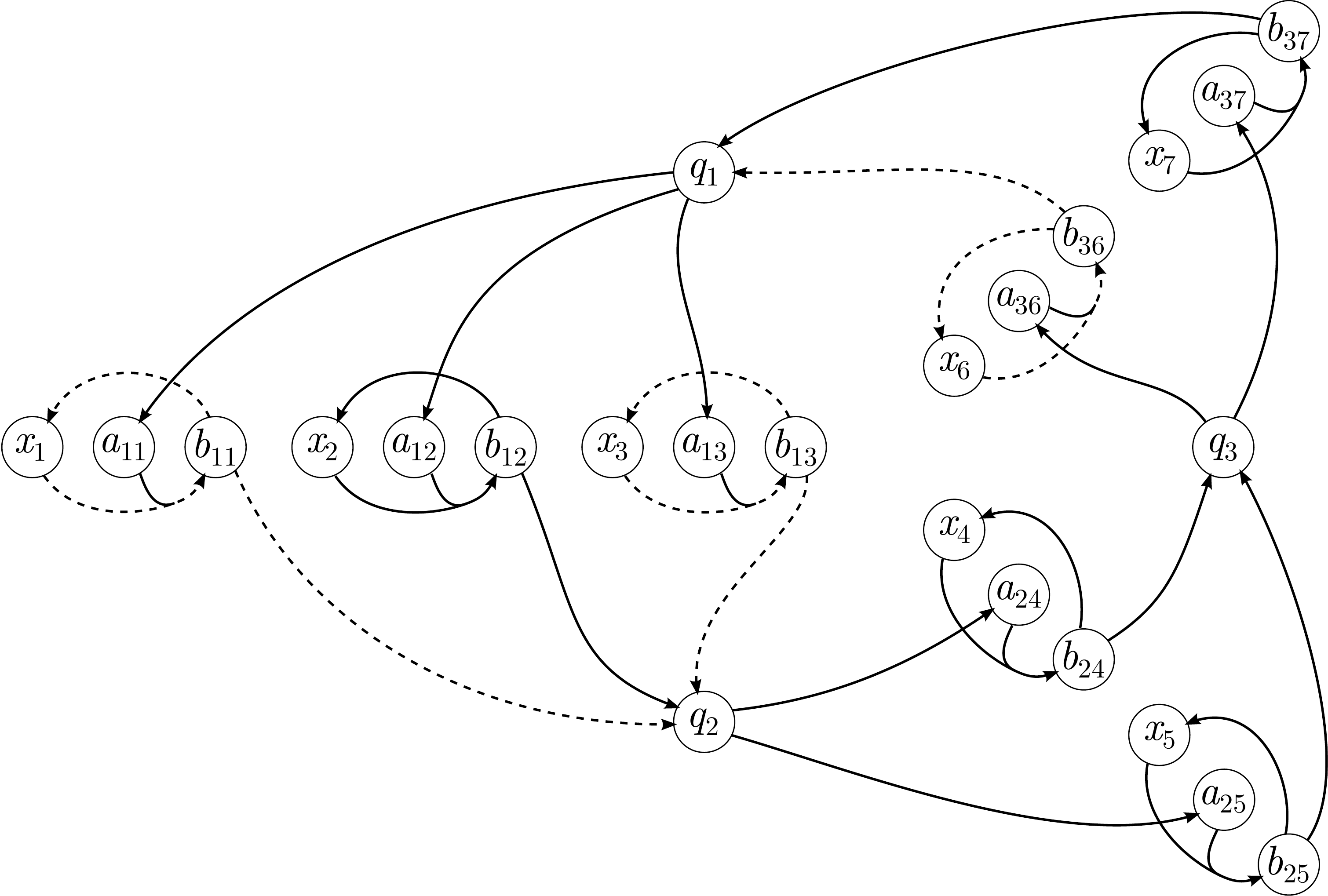}
\par\end{centering}

\protect\caption[An example for the proof of \prettyref{lem:m har}, illustrating the
idea]{\emph{}\label{fig:cir}\setlength{\parindent}{1em}An example for
the proof of \prettyref{lem:m har}, illustrating the idea. Depicted
here is the hypergraph representation (see Remark \ref{hyp}) of $\phi$
for $\psi=\left(x_{1}\vee x_{2}\vee x_{3}\right)\wedge\left(x_{4}\vee x_{5}\right)\wedge\left(x_{6}\vee x_{7}\right)$,
as constructed in the proof.\protect \\
\hspace*{4ex}Any self-implicating set of $\phi$ must contain a ``large
circulatory'', passing through each $q_{p}$ and at least one gadget
$\vcenter{\hbox{\protect\includegraphics[scale=0.46]{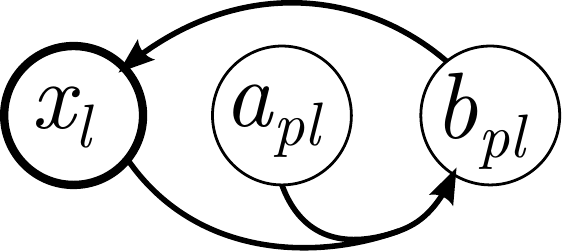}}}$ for
each $p$; these gadgets act as ``valves'': If some $x_{i}$ is
not allowed to be assigned 1 (due to restraints), the circulatory
may not pass through any gadget containing $x_{i}$.\protect \\
\hspace*{4ex}Every maximal self-implicating set also contains all
$a_{pl}$; here, for example, one maximal self-implicating set consist
of the variables with the outgoing edges drawn solid.\protect \\
\hspace*{4ex}If we would add restraint clauses to $\psi$ s.t. $\psi$
would become unsatisfiable, e.g. $\overline{x_{4}}\vee\overline{x_{6}}$,
$\overline{x_{4}}\vee\overline{x_{7}}$, $\overline{x_{5}}\vee\overline{x_{6}}$,
and $\overline{x_{5}}\vee\overline{x_{7}}$, each maximal self-implicating
set of the corresponding $\phi$ would contain a restraint set, so
that $G(\phi)$ would be connected.}
\end{figure}

\begin{figure}[!h]
\begin{centering}
\includegraphics[scale=0.46]{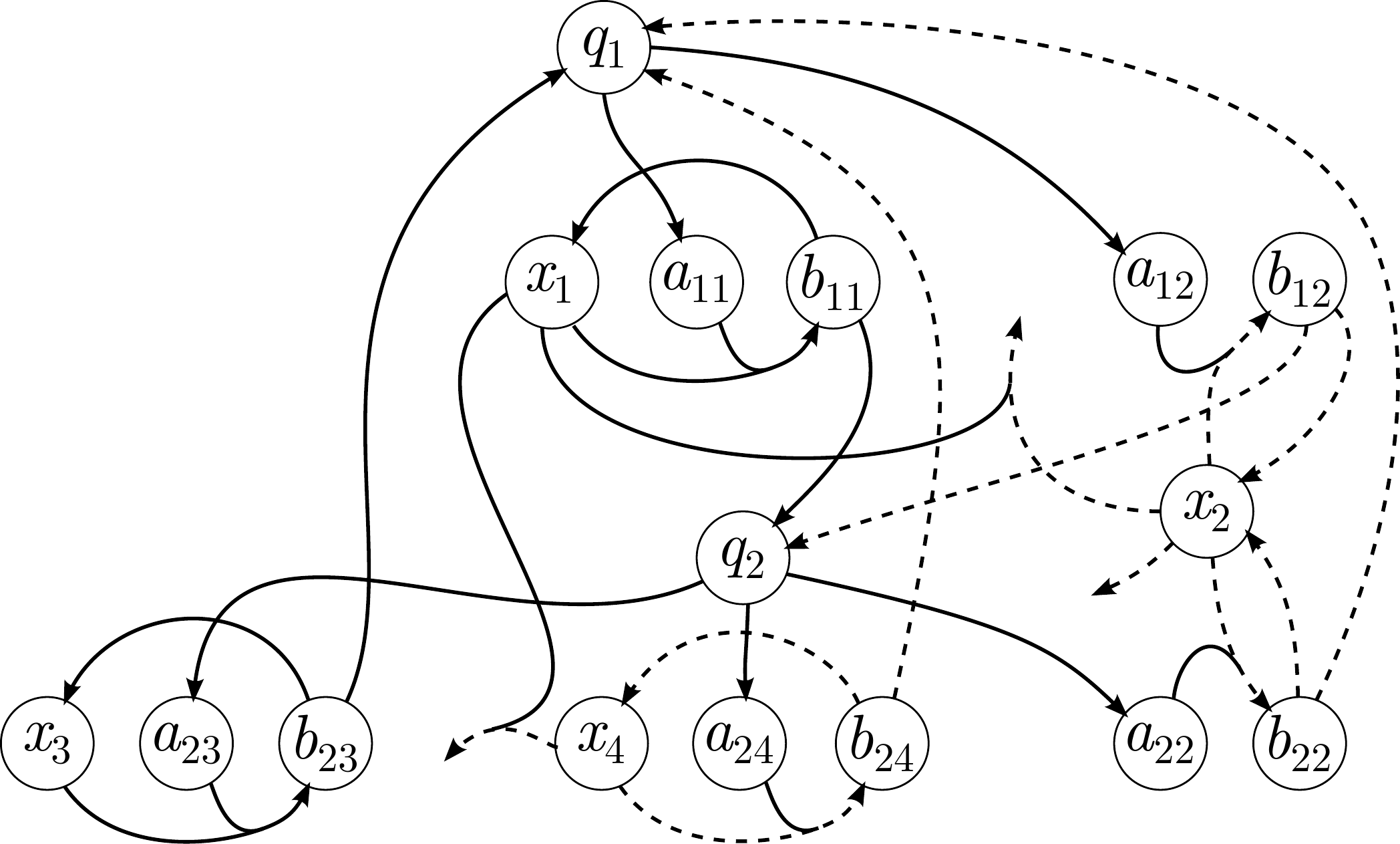}
\par\end{centering}

\protect\caption[A more complex example for the proof of \prettyref{lem:m har}]{\emph{}\label{fig:gr}A more complex example, with a variable of
$\psi$ appearing twice in a $P$-constraint: Depicted is $\phi$
for $\psi=\left(x_{1}\vee x_{2}\right)\wedge\left(x_{3}\vee x_{4}\vee x_{2}\right)\wedge\left(\overline{x}_{1}\vee\overline{x}_{2}\right)\wedge\left(\overline{x}_{1}\vee\overline{x}_{4}\right)\wedge\left(\overline{x}_{2}\right)$.
\protect \\
{\small{}\hspace*{3ex}}$\psi$ is satisfiable with the unique solution
$x_{1}=x_{3}=1$ and $x_{2}=x_{4}=0$, and $G(\phi)$ is disconnected
(with exactly two components, since there is exactly one maximal self-implicating
set containing no restraint set, consisting of the variables with
the outgoing edges drawn solid).}
\end{figure}

We reduce the no-constants satisfiability problem \noun{Sat}($\left\{ P,N\right\} $)
with $P=x\vee y\vee z$ and $N=\overline{x}\vee\overline{y}$ to the
complement of \noun{Conn}$_{C}$($\left\{ M\right\} $), where $M=\left(x\vee\overline{y}\vee\overline{z}\right)\wedge\left(\overline{x}\vee z\right)$.
\noun{Sat}($\left\{ P,N\right\} $) is NP-hard by Schaefer's dichotomy
theorem (\prettyref{thm:sch}) since $P$ is not 0-valid, not bijunctive,
not Horn and not affine, while $N$ is not 1-valid and not dual Horn.

Let $\psi$ be any CNF($\left\{ P,N\right\} $)-formula. If $\psi$
only contains $N$-constraints, it is trivially satisfiable, so assume
it contains at least one $P$-constraint. We construct a CNF\textsubscript{C}($\{M\}$)-formula
$\phi$ s.t.$\mbox{\,}$the solution graph $G(\phi)$ is disconnected
iff $\psi$ is satisfiable. First note that we can use the relations
$\overline{x}\vee\overline{y}=M(0,x,y)$ and $\overline{x}\vee y=M(x,0,y)$. 

For every variable $x_{i}$ of $\psi$ ($i=1,\ldots,n$), there is
the same variable $x_{i}$ in $\phi$. For every $N$-constraint $\overline{x}_{i}\vee\overline{x}_{j}$
of $\psi$, there is the clause $\overline{x}_{i}\vee\overline{x}_{j}$
in $\phi$ also. For every $P$-constraint $c_{p}=x_{i_{p}}\vee x_{j_{p}}\vee x_{k_{p}}$
($p=1,\ldots,m$) of $\psi$ there is an additional variable $q_{p}$
in $\phi$, and for every $x_{l}\in\{x_{i_{p}},x_{j_{p}},x_{k_{p}}\}$
appearing in $c_{p}$, there are two more additional variables $a_{pl}$
and $b_{pl}$ in $\phi$. Now for every $c_{p}$, for each $l\in\{i_{p},j_{p},k_{p}\}$
the constraints $\overline{q}_{p}\vee a_{pl}$, $\left(\overline{x}_{l}\vee\overline{a}_{pl}\vee b_{pl}\right)\wedge\left(\overline{b}_{pl}\vee x_{l}\right)$
and $\overline{b}_{pl}\vee q_{(p+1)\,\mathrm{mod}\,m}$ are added
to $\phi$. See the figures for examples of the construction. 

If $\psi$ is satisfiable, there is an assignment $\boldsymbol{s}$
to the variables $x_{i}$ s.t.$\mbox{\,}$for every $P$-constraint
$c_{p}$ there is at least one $x_{l}\in\{x_{i_{p}},x_{j_{p}},x_{k_{p}}\}$
assigned 1, and for no $N$-constraint $\overline{x}_{i}\vee\overline{x}_{j}$,
both $x_{i}$ and $x_{j}$ are assigned 1. We extend $\boldsymbol{s}$
to a locally minimal nonzero satisfying assignment $\boldsymbol{s}'$
for $\phi$; then $G(\phi)$ is disconnected by \prettyref{lem:loc min}:
Let all $q_{p}=1$, $a_{pl}=1$, and all $b_{pl}=x_{l}$ in $\boldsymbol{s}'$.
It is easy to check that all clauses of $\phi$ are satisfied, and
that all variables assigned 1 appear as the positive literal in an
implication clause with all its variables assigned 1, so that $\boldsymbol{s}'$
is locally minimal. $\boldsymbol{s}'$ is nonzero since $\psi$ contains
at least one $P$-constraint.

Conversely, if $G(\phi)$ is disconnected, $\phi$ has a maximal self-implicating
set $U$ containing no restraint set by \prettyref{cor:horn conn}.
It is easy to see that $U$ must contain all $q_{p}$, all $a_{pl}$,
and for every $p$ for at least one $l\in\{i_{p},j_{p},k_{p}\}$ both
$b_{pl}$ and $x_{l}$(see also Figure \ref{fig:cir} and the explanation
beneath). Thus the assignment with all $x_{i}\in U$ assigned 1 and
all other $x_{i}$ assigned 0 satisfies $\psi$.
\end{proof}

\subsection{Expressing \emph{M}}
\begin{lem}
\label{lem:exp m}The relation $M=\left(x\vee\overline{y}\vee\overline{z}\right)\wedge\left(\overline{x}\vee z\right)$
is expressible as a \noun{CNF\textsubscript{C}($\{R\}$)}-formula
for every Horn relation $R$ that is not  safely componentwise IHSB$-$.\end{lem}
\begin{proof}
$M=\left(x\vee\overline{y}\vee\overline{z}\right)\wedge\left(\overline{x}\vee z\right)=$
$\vcenter{\hbox{\includegraphics[scale=0.4]{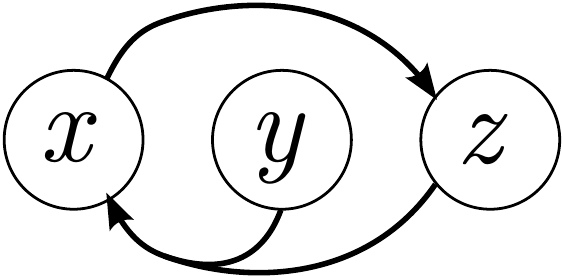}}}$ contains a multi-implication
clause where some negated variable is not implied. The only other
3-ary such relations are (up to permutation of variables)

$L=\left(x\vee\overline{y}\vee\overline{z}\right)\wedge\left(\overline{x}\vee\overline{y}\vee z\right)=$
$\vcenter{\hbox{\includegraphics[scale=0.4]{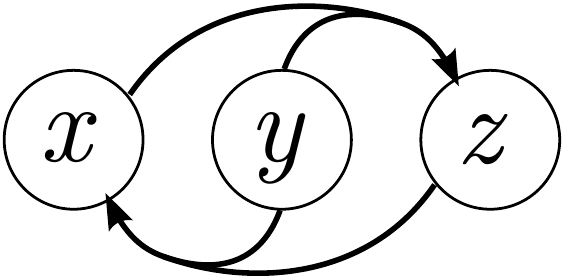}}}$$\enskip$and
$K=x\vee\overline{y}\vee\overline{z}=$ $\vcenter{\hbox{\includegraphics[scale=0.4]{drws}}}$.
We show that $M,L,$ or $K$ is expressible from $R$ by substitution
of constants and identification of variables. We can then express
$M$ from $K$ or $L$ as 
\[
M(x,y,z)\equiv K(x,y,z)\wedge K(z,x,x)\equiv L(x,y,z)\wedge L(z,x,x).
\]

We will argue with formulas simplified according to \prettyref{def:nor};
let $\phi_{0}=\nu(R)$. The following 7 numbered transformation steps
generate $K$, $L$, or $M$ from $\phi_{0}$. After each transformation,
we assume $\nu$ is applied again to the resulting formula; we denote
the formula resulting from the $i$'th transformation step in this
way by $\phi_{i}$.

In the first three steps, we ensure that the formula contains a multi-implication
clause where some variable is not implied, in the fourth step we trim
the multi-implication clause to size 3, and in the last three steps
we eliminate all remaining clauses and variables not occurring in
$K$, $L$, or $M$. Our first goal is to produce a formula with a
connected solution graph that is not IHSB$-$, which will turn out
helpful.
\begin{enumerate}
\item \emph{Obtain a not componentwise IHSB$-$ formula $\phi_{1}$ from
$\phi_{0}$ by identification of variables.}
\end{enumerate}
Let $\left[\phi_{1}^{*}\right]$ be a connected component of $\left[\phi_{1}\right]$
that is not IHSB$-$, and let $U$ be the set of variables assigned
1 in the minimum solution of $\phi_{1}^{*}$.
\begin{enumerate}[resume]
\item \emph{\label{enu:s1}Substitute 1 for all variables from $U$.}
\end{enumerate}
The resulting formula $\phi_{2}$ now contains no positive unit-clauses.
Further, the component $\left[\phi_{2}^{*}\right]$ of $\left[\phi_{2}\right]$
resulting from $\left[\phi_{1}^{*}\right]$ is still not IHSB$-$,
and it has the all-0 vector as minimum solution. We show that 
\begin{equation}
{\textstyle \phi_{2}^{*}\equiv\nu\left(\phi_{2}\wedge\left(\bigvee_{x\in V_{1}}\overline{x}\right)\wedge\cdots\wedge\left(\bigvee_{x\in V_{k}}\overline{x}\right)\right),}\label{eq:st}
\end{equation}
where $V_{1},\ldots,V_{k}$ are the sets of variables assigned 1 in
the minimum solutions $\boldsymbol{v}^{1},\ldots,\boldsymbol{v}^{k}$
of the other components of $\left[\phi_{2}\right]$, and we specified
the formula to be in normal form:
\begin{itemize}[label= ]
\item For any solution $\boldsymbol{a}$ in the component with minimum
solution $\boldsymbol{v}^{i}$ we have $\boldsymbol{a}\geq\boldsymbol{v}^{i}$
(see \prettyref{def:lm}), so all components other than $\left[\phi_{2}^{*}\right]$
are eliminated in the right-hand side of \eqref{eq:st}. By \prettyref{lem:in},
no vector from $\left[\phi_{2}^{*}\right]$ is removed.
\end{itemize}
By \prettyref{lem:horn conn}, $V_{1},\ldots,V_{k}$ are exactly the
non-empty maximal self-implicating sets of $\phi_{2}$ that contain
no restraint set.

Clearly, $\phi_{2}^{*}$ is not IHSB$-$. However, we have no restraint
clauses at our disposal to generate $\phi_{2}^{*}$ from $\phi_{2}$;
nevertheless, we can isolate a connected part of $\phi_{2}^{*}$ that
is not IHSB$-$, as we will see.

Since $\phi_{2}^{*}$ is not IHSB$-$, it contains a multi-implication
clause $c^{*}$, and by \eqref{eq:st} it is clear that $\phi_{2}$
must contain the same clause $c^{*}$.

By simplification rule \ref{enu:imp imp}, Imp(Var$(c^{*})$) contains
no restraint set in $\phi_{2}$. Now if some self-implicating set
$U^{*}$ were implied by Var$(c^{*})$, the related maximal self-implicating
set $U_{m}^{*}$ (which then were also implied by Var$(c^{*})$) could
contain no restraint set, thus a restraint clause would be added for
the variables from $U_{m}^{*}$ in \eqref{eq:st}. But then $c^{*}$
would be removed by $\nu$ in \eqref{eq:st}, again due to rule \ref{enu:imp imp},
which is a contradiction. Thus Imp(Var$(c^{*})$) also contains no
self-implicating set in $\phi_{2}$, and so the following operation
eliminates all self-implicating sets and all restraint clauses:
\begin{enumerate}[resume]
\item \emph{Substitute 0 for all remaining variables not implied by} Var($c^{*}$).
\end{enumerate}
This operation also produces no new restraint clauses since any implication
clause with the positive literal not implied by Var($c^{*}$) must
also have some negative literal not implied by Var($c^{*}$), and
thus vanishes.

Further, since $\phi_{2}$ contained no positive unit-clauses, the
formula cannot have become unsatisfiable by this operation. Also,
it is easy to see that the simplification initiated by the substitution
of 0 for some variable $x_{i}$ can only affect clauses $c$ with
$x_{i}\in\mathrm{Imp}(\mathrm{Var}(c))$, so $c^{*}$ is retained
in $\phi_{3}$. 

Since all variables not from Var($c^{*}$) are now implied by\emph{
}Var($c^{*}$), and Imp(Var($c^{*}$)) is not self-implicating, $c^{*}$
contains a variable that is not implied; w.l.o.g., let $c^{*}=x\vee\overline{y}\vee\overline{z}_{1}\vee\cdots\vee\overline{z}_{k}$
($k\geq1$) s.t. $y$ is not implied.
\begin{enumerate}[resume]
\item \emph{\label{enu:Id z-2}Identify $z_{1},\ldots,z_{k}$, call the
resulting variable $z$.}
\end{enumerate}
This produces the clause $c^{\sim}=x\vee\overline{y}\vee\overline{z}$
from $c^{*}$. Clearly, $y$ is still not implied in $\phi_{4}$,
and since $x$ was not implied by any set $U\subsetneq\{y,z_{1},\ldots,x_{k}\}$
by simplification rule \ref{enu:red bra} in $\phi_{3}$, and no $z_{i}$
was implied by $y$, it follows for $\phi_{4}$ that
\begin{itemize}[label= ({*})]
\item $x\notin\mathrm{Imp}(y)$, $x\notin\mathrm{Imp}(z)$, $z\notin\mathrm{Imp}(y)$,
$y$ is not implied.
\end{itemize}
Also, since $x$ was implied by $\{y,z_{1},\ldots,x_{k}\}$ only via
$c^{*}$ in $\phi_{3}$ due to simplification rule \ref{enu:red imp},
$x$ is implied by $\{y,z\}$ only via $c^{\sim}$ in $\phi_{4}$.

In the following steps, we eliminate all variables other than $x,y,z,$
s.t. $c^{\sim}$ is retained and ({*}) is maintained. It follows that
we are then left with $K,L$, or $M$, since the only clauses only
involving $x,y,z$ and satisfying ({*}) besides $c^{\sim}$ are from
$\{z\vee\overline{x},\,z\vee\overline{x}\vee\overline{y}\}$.
\begin{enumerate}[resume]
\item \emph{Substitute 1 for every variable from $\mathrm{Imp}(y)\setminus\{y\}$
. }
\end{enumerate}
For the simplification initiated by this operation, note that $\phi_{4}$
contained no restraint clauses. It follows that the formula cannot
have become unsatisfiable by this operation. Further, it is easy to
see that for a Horn formula without restraint clauses, at a substitution
of 1 for variables from a set $U$, only clauses $c$ containing at
least one variable $x_{i}\in\mathrm{Imp}(U)$ are affected by the
simplification. Thus, $c^{\sim}$ is not affected since $x,y$ and
$z$ were not implied by $\mathrm{Imp}(y)\setminus\{y\}$.

We must carefully check that ({*}) is maintained since substitutions
of 1 may result in new implications: Since $\mathrm{Imp}(y)\setminus\{y\}$
is empty in $\phi_{5}$, still $x\notin\mathrm{Imp}(y)$ and $z\notin\mathrm{Imp}(y)$.
It is easy to see that $x$ could only have become implied by $z$
as result of transformation 5 if there had been a multi-implication
clause (other than $c^{\sim}$) in $\phi_{4}$ with the positive variable
implying $x$, and each negated variable implied by $y$ or $z$;
but this is not the case since $x$ was implied by $\{y,z\}$ only
via $c^{\sim}$ in $\phi_{4}$, thus still $x\notin\mathrm{Imp}(z)$.

We eliminate all remaining variables besides $x,y,z$ by identifications
in the next two steps. Since now $\mathrm{Imp}(y)\setminus\{y\}$
is empty, the only condition from ({*}) we have to care about is that
$x\notin\mathrm{Imp}(z)$ remains true.
\begin{enumerate}[resume]
\item \emph{\label{enu:id x}Identify all remaining variables from $\mathrm{Imp}(z)\setminus\{z\}$
with $z$.} 
\end{enumerate}
Now $\mathrm{Imp}(z)\setminus\{z\}$ is empty, so the last step is
easy:
\begin{enumerate}[resume]
\item \emph{\label{enu:id z}Identify all remaining variables other than
$x,y,z$ with $x$.} 
\end{enumerate}
\end{proof}
This completes the \noun{coNP}-completeness proof for connectivity
and the proof of the trichotomy:
\begin{cor}
\label{cor:co}If a finite set $\mathcal{S}$ of logical relations
is safely tight but not CPSS, \noun{Conn}\textsubscript{\noun{C}}\noun{($\mathcal{S}$)}
is \noun{coNP}-complete.\end{cor}
\begin{proof}
By \prettyref{lem:c47}, \noun{Conn}\textsubscript{C}($\mathcal{S}$)
is in coNP. If $\mathcal{S}$ is not Schaefer, coNP-hardness follows
from \prettyref{lem:l48}. If $\mathcal{S}$ is Schaefer and not CPSS,
it must be Horn and contain at least one relation that is not  safely
componentwise IHSB$-$, or dual Horn and contain at least one relation
that is not  safely componentwise IHSB$+$; in the first case, coNP-hardness
follows from Lemmas \ref{lem:m har} and \ref{lem:exp m}, the second
case follows by symmetry.
\end{proof}

\section{Further Results about Constraint-Projection Separation}

This section is not needed for the proof of the trichotomy (\prettyref{thm:trich}),
but gives further insights that will be useful for the investigation
of formulas without constants in the next chapter.

We start by showing that with \prettyref{lem:prjs}, we have found
all Schaefer sets of relations that are CPS:
\begin{lem}
\label{lem:sc}If a set $\mathcal{S}$ of relations is Schaefer but
not CPSS, it is not constraint-projection separating.\end{lem}
\begin{proof}
Since $\mathcal{S}$ is Schaefer but not CPSS, it must contain some
relation that is Horn but not  safely componentwise IHSB$-$, or dual
Horn but not  safely componentwise IHSB$+$. Assume the first case,
the second one is analogous. Then by \prettyref{lem:exp m}, we can
express $M=\left(x\vee\overline{y}\vee\overline{z}\right)\wedge\left(\overline{x}\vee z\right)$
as a CNF\textsubscript{C}($\mathcal{S}$)-formula. Consider the CNF\textsubscript{C}($\mathcal{S}$)-formula
\[
T(u,v,w,x,y,z)=M(u,v,w)\wedge M(x,y,z)\wedge M(w,w,y)\wedge M(z,z,v)
\]
\[
\equiv\left(\left(u\vee\overline{v}\vee\overline{w}\right)\wedge\left(\overline{u}\vee w\right)\right)\wedge\left(\left(x\vee\overline{y}\vee\overline{z}\right)\wedge\left(\overline{x}\vee z\right)\right)\wedge\left(y\vee\overline{w}\right)\wedge\left(v\vee\overline{z}\right).
\]
Now $G(T)$ is disconnected by \prettyref{cor:horn conn} since $\{u,v,w,x,y,z\}$
is maximal self-implicating, but neither the projection $\exists x\exists y\exists zT\equiv M(u,v,w)$
to the variables of the first constraint in the CNF($\{M\}$)-representation
of $T$, nor the projection $\exists u\exists v\exists x\exists zT\equiv y\vee\overline{w}$
to the variables of the third one is disconnected. The second and
fourth constraints are symmetric to the first and third ones. 

Since in the CNF\textsubscript{C}($\mathcal{S}$)-representation
of $T$ every conjunct $M(r,s,t)$ of $T$ ($r,s,t\in\{u,v,w,x,y,z\}$)
is a CNF\textsubscript{C}($\mathcal{S}$)-formula $\bigwedge_{i}R_{i}(\boldsymbol{\xi}^{i})$
with $R_{i}\in\mathcal{S}$ and $\xi_{j}^{i}\in\{0,1,r,s,t\}$, for
every constraint $C_{i}$ of $T$, the set $\mathrm{Var}(C_{i})$
is a subset of $\{u,v,w\},$ $\{x,y,z\},$ $\{y,w\}$ or $\{v,z\}$,
and thus also for no $C_{i}$ the projection to $\mathrm{Var}(C_{i})$
is disconnected.
\end{proof}
By \prettyref{lem:prjs} we see that there are non-Schaefer sets that
are CPS, e.g. $\{R\}$ with $R=\{100,010,001\text{\}}$, which is
safely componentwise bijunctive but not Schaefer. It is open whether
there are other such sets not mentioned in \prettyref{lem:prjs}.

While we will see in \prettyref{sub:eas} that there are not safely
tight sets that are no-constants CPS (see \prettyref{def:dcp-1}),
it is likely that no not safely tight set is CPS, else we had a \noun{$\mathrm{P^{NP}}$}-algorithm
for a PSPACE-complete problem. We can show that not safely tight sets
are at least not by \prettyref{lem:prjs} CPS:
\begin{lem}
\label{lem:nst}If a set of relations $\mathcal{S}$ is not safely
tight, it also is not bijunctive, not safely componentwise IHSB$-$,
not safely componentwise IHSB$+$, and not affine.\end{lem}
\begin{proof}
By \prettyref{lem:l42}, $\mathcal{S}$ is not bijunctive and not
affine. Also, $\mathcal{S}$ must contain a relation $R$ s.t. the
relation OR=$\{01,10,11\}$ can be obtained from $R$ by identification
of variables and substitution of constants. Since these operations
are permutable, we can assume that OR can be obtained by first producing
an $n-ary$ relation $R'$ by identification of variables, and then
w.l.o.g. setting the first $n-2$ variables to constants $c_{1}\cdots c_{n-2}$.
Then $\{c_{1}\cdots c_{n-2}01,c_{1}\cdots c_{n-2}10,c_{1}\cdots c_{n-2}11\}\subset R'$,
but $c_{1}\cdots c_{n}00\notin R'$. Since these 3 vectors from $R'$
are in one component of $R'$, already that component is not safely
OR-free, so it cannot be Horn by \prettyref{lem:l42}, and thus is
not IHSB$-$. But then $R'$, and hence $R$, was not safely componentwise
IHSB$-$. \end{proof}
\begin{rem}
The Lemmas \ref{lem:prjs i} and \ref{lem:prjs b} cannot be generalized
to  safely componentwise bijunctive or  safely componentwise IHSB$-$
relations: For sets \emph{$\mathcal{S}$} of  safely componentwise
bijunctive (safely componentwise IHSB$-$) relations that are not
bijunctive (IHSB$-$), there are CNF\textsubscript{C}($\mathcal{S}$)-formulas
with pairs of components that are not disconnected in the projection
to any constraint:

E.g., for the  safely componentwise bijunctive relation $R=\left((x\vee\overline{y})\wedge\overline{z}\right)\vee\left(\overline{x}\wedge y\wedge z\right)$,
the CNF(\emph{$\{R\}$})-formula $F(x,y,z,w)=R(x,y,z)\wedge R(y,x,w)$
has the four pairwise disconnected solutions $a$=0000, $b$=1100,
$c$=0110, and $d$=1001, but $a$ is connected to $b$ in the projection
to $\{x,y,z\}$ as well as in the one to $\{x,y,w\}$.

It follows that for such relations there is no algorithm for $st$-connectivity
analogous to that of \prettyref{lem:alg}. In \prettyref{sub:eas}
we will actually see cases for no-constants formulas where we can
solve connectivity in polynomial time via constraint-projection separation
while $st$-connectivity is PSPACE-complete.
\end{rem}

\chapter{\label{chap:3}No-Constants and Quantified Variants}

\section{\label{sec:No-Constants}No-Constants}

Complexity classifications for CNF($\mathcal{S}$)-formulas without
constants seem to be more favored, but also more difficult to prove.
For example, in \citep{Schaefer:1978:CSP:800133.804350}, Schaefer
stated the no-constants classification as the main theorem, but then
first proved a classification for formulas with constants as intermediate
result. So we also will now attack the no-constants versions of our
$st$-connectivity and connectivity problems, denoted by\emph{ }\emph{\noun{st-}}\noun{Conn($\mathcal{S}$)}\emph{
}and\emph{ }\noun{Conn($\mathcal{S}$)}, respectively.

For $st$-connectivity and the diameter, we prove that the same dichotomy
holds as for formulas with constants (in \citep{Gopalan:2006}, Gopalan
et al.~already stated that they could extend their dichotomy theorem
for $st$-connectivity to formulas without constants, but didn't show
the proof). 

For connectivity, we can extend the tractable class in two ways. Thereby,
we get a quite interesting result: That there are cases when connectivity
is easier than $st$-connectivity; namely, for some sets of relations,
connectivity is in P even though $st$-connectivity is PSPACE-complete.
While we can in some cases show that PSPACE-completeness and coNP-completeness
carry over from the case with constants, we must leave open the complexity
in two situations. 

The following two theorems and the table below summarize our results.

\begin{table}[!h]
\begin{tabular}{|c||c|c||c|c|}
\hline 
{\footnotesize{}$\mathcal{S}$} & \noun{\footnotesize{}Conn}{\footnotesize{}($\mathcal{S}$) } & \noun{\footnotesize{}Conn\textsubscript{C}}{\footnotesize{}($\mathcal{S}$) } & \noun{\footnotesize{}st-Conn}{\footnotesize{}($\mathcal{S}$) } & \noun{\footnotesize{}st-Conn\textsubscript{C}}{\footnotesize{}($\mathcal{S}$) }\tabularnewline
\hline 
\hline 
{\footnotesize{}not safely tight, not {*}} & {\footnotesize{}PSPACE-c.} & \multirow{3}{*}{{\footnotesize{}PSPACE-c.}} & \multirow{3}{*}{{\footnotesize{}PSPACE-c.}} & \multirow{3}{*}{{\footnotesize{}PSPACE-c.}}\tabularnewline
\cline{1-2} 
{\footnotesize{}not safely tight, not q.disc., {*}} & \textbf{\footnotesize{}in PSPACE} &  &  & \tabularnewline
\cline{1-2} 
{\footnotesize{}not safely tight, q.disc. } & {\footnotesize{}in P} &  &  & \tabularnewline
\hline 
{\footnotesize{}safely tight, not {*}, not Schaefer} & \multirow{3}{*}{{\footnotesize{}coNP-c.}} & \multirow{8}{*}{{\footnotesize{}coNP-c.}} & \multirow{9}{*}{{\footnotesize{}in P}} & \multirow{9}{*}{{\footnotesize{}in P}}\tabularnewline
\cline{1-1} 
{\footnotesize{}Schaefer, not implicative,} &  &  &  & \tabularnewline
{\footnotesize{}not nc-CPSS} &  &  &  & \tabularnewline
\cline{1-2} 
{\footnotesize{}safely tight, {*}, not Schaefer, } & \multirow{2}{*}{\textbf{\footnotesize{}in coNP}} &  &  & \tabularnewline
{\footnotesize{}not} {\footnotesize{}q.disc., not nc-CPSS} &  &  &  & \tabularnewline
\cline{1-2} 
{\footnotesize{}safely tight, q.disc, not CPSS} & \multirow{4}{*}{{\footnotesize{}in P}} &  &  & \tabularnewline
\cline{1-1} 
{\footnotesize{}nc-CPSS, not CPSS } &  &  &  & \tabularnewline
\cline{1-1} 
{\footnotesize{}implicative, not CPSS} &  &  &  & \tabularnewline
\cline{1-1} \cline{3-3} 
{\footnotesize{}CPSS} &  & {\footnotesize{}in P} &  & \tabularnewline
\hline 
\end{tabular}

\protect\caption[The classifications for CNF($\mathcal{S}$)-formulas without constants]{\emph{}\textbf{\label{tab:cl}}The classifications for CNF($\mathcal{S}$)-formulas
without constants, in comparison to the case with constants.\protect \\
{*} = (0-valid or 1-valid or complementive) ~~~~~q.disc. = quasi
disconnecting\protect \\
The cases where the complexity (up to polynomial-time isomorphisms)
is not yet known are highlighted.}
\end{table}

\newpage{}In this whole section, we assume the sets \emph{$\mathcal{S}$}
to contain no empty relations (otherwise, since empty relations are
not 0-valid and not 1-valid, the hardness statements for sets containing
relations which are not 0-valid or not 1-valid would be wrong). This
assumption has also to be made, e.g., for Schaefer's theorem in the
no-constants case, but is often omitted.
\begin{thm}[Dichotomy theorem for \noun{st-Conn}($\mathcal{S}$) and the diameter]
\label{thm:ndich} Let $\mathcal{S}$ be a finite set of logical
relations.
\begin{enumerate}
\item If $\mathcal{S}$ is  safely tight, \emph{\noun{st-}}\noun{Conn($\mathcal{S}$)}
is in \noun{P,} and for every \noun{CNF($\mathcal{S}$)}-formula $\phi$,
the diameter of $G(\phi)$ is linear in the number of variables.
\item Otherwise, \emph{\noun{st-}}\noun{Conn($\mathcal{S}$)} is \noun{PSPACE}-complete,
and there are \noun{CNF($\mathcal{S}$)}-formulas $\phi$\emph{ }such
that the diameter of $G(\phi)$ is exponential in the number of variables.
\end{enumerate}
\end{thm}
\begin{proof}
See \prettyref{sub:st-}.\end{proof}
\begin{thm}[Classification for \noun{Conn}($\mathcal{S}$)]
\label{thm:trich-1}Let $\mathcal{S}$ be a finite set of non-empty
logical relations.
\begin{enumerate}
\item If $\mathcal{S}$ is nc-CPSS, quasi disconnecting or implicative,\emph{
}\noun{Conn($\mathcal{S}$)} is in \noun{P}.
\item Else, if $\mathcal{S}$ is Schaefer, or if $\mathcal{S}$ is safely
tight but not Schaefer and not 0-valid nor 1-valid nor complementive,
\noun{Conn($\mathcal{S}$)} is \noun{coNP}-complete.
\item Else, if $\mathcal{S}$ is safely tight, \noun{Conn($\mathcal{S}$)}
is in \noun{coNP.}
\item Else, if $\mathcal{S}$ is not 0-valid nor 1-valid nor complementive,\emph{
}\noun{Conn($\mathcal{S}$)} is \noun{PSPACE}-complete.
\item Else, \noun{Conn($\mathcal{S}$)} is in\noun{ PSPACE.}
\end{enumerate}
\end{thm}
\begin{proof}
1. See \prettyref{cor:sa}, \prettyref{cor:qua}, and \prettyref{lem:si}.

2. See \prettyref{lem:cco} and \prettyref{cor:ns}.

3. This result carries over from the case with constants (\prettyref{thm:trich}).

4. See \prettyref{cor:nr}.

5. This result carries over from the case with constants (\prettyref{thm:trich}).
\end{proof}

\subsection{\emph{\label{sub:st-}st}-Connectivity and Diameter}

The PSPACE-hardness proof for \emph{\noun{st-}}\noun{Conn($\mathcal{S}$)}
will be by reduction from \emph{\noun{st-}}\noun{Conn}\emph{\noun{\textsubscript{C}}}\noun{($\mathcal{S}$).}
An obvious way to reduce a problem for formulas with constants to
one for formulas without is to replace every occurrence of a constant
with a new variable, and then to add constraints for the new variables
such that, with regard to the problem at hand, the transformed formula
is equivalent to the original one. This approach was already used
by Schaefer \citep{Schaefer:1978:CSP:800133.804350}.

For $st$-connectivity, we have to make sure that for every two solutions
of the original formula, there are two solutions of the transformed
formula that are connected iff the solutions of the original formula
are connected. In \prettyref{lem:red} below we show how this is possible
for not safely tight sets of relations. We need the following definition
and the next lemma.
\begin{defn}
A solution $\boldsymbol{a}$ of a formula $\phi$ is \emph{isolated}
if $\boldsymbol{a}$ is not connected to any other solution $\boldsymbol{b}$
in $G(\phi)$. A formula $\phi$ is \emph{0-isolating }(\emph{1-isolating})
if it has an isolated solution $\boldsymbol{a}\neq(1\cdots1)$ ($\boldsymbol{a}\neq(0\cdots0)$).
Similarly, we define isolated\emph{ }vectors for relations, and 0-isolating
and 1-isolating relations.\end{defn}
\begin{lem}
\label{lem:iso}If an $n$-ary logical relation $R$\emph{ }is not
safely OR-free, there is a 1-isolating\emph{ }\noun{CNF($\{R\}$)}-formula
\emph{$\phi$.}\end{lem}
\begin{proof}
By identification of variables, we can obtain a not OR-free relation
$R^{*}$ from $R$. W.l.o.g., assume that OR can be obtained from
$R^{*}$ by setting the last $n-2$ variables to constants $c_{3},\ldots,c_{n}$
(for relations that also require identification of variables to obtain
OR, the identification can be done in a prior step); then $R^{*}(x_{1},x_{2},c_{3},\ldots,c_{n})$
$=x_{1}\vee x_{2}$.

If $n=2$, we take $\phi=R^{*}(x,x)=x$, then $[\phi]=\{1\}$.

Else, if all $c_{3},\ldots,c_{n}=1$, we define a 3-ary relation $R'$
by identifying the last $n-2$ variables. Then $R'(x_{1},x_{2},1)=x_{1}\vee x_{2}$,
and it follows that identifying the first two variables of $R'$ yields
a 2-ary relation $R''=R'(x_{1},x_{1},x_{2})$ with $11\in R''$ and
$01\notin R''$, thus $R''$ equals $\{11,00,10\},$ $\{11,00\},$
$\{11,10\}$ or $\{11\}$. The second and fourth relation are already
1-isolating, so we let $\phi=R''(x_{1},x_{2})$. The first is $x_{1}\vee\overline{x_{2}}$,
and we obtain a 1-isolating relation by taking $\phi=R''(x_{1},x_{2})\wedge R''(x_{2},x_{1})$
with $[\phi]=\{11,00\}$. From the third one we obtain $\{1\}$ as
$\phi=R''(x_{1},x_{1})$.

Similarly, if all $c_{3},\ldots,c_{n}=0$, by identifying the last
$n-2$ variables, and then the fist two, we get a relation $R''$
with $10\in R''$ and $00\notin R''$, thus $R''$ equals $\{10,01,11\}$,
$\{10,01\}$, $\{10,11\}$ or $\{10\}$. Here again, the second and
fourth relation are already 1-isolating, and from the first as well
as from the third we obtain $\{1\}$ by identifying the two variables.

Otherwise, we define a 3-ary relation $R''$ by identifying all variables
$x_{i}$ of $R^{*}$ with $c_{i}=0$, then all with $c_{i}=1$, and
then the first two, i.e., $R''(x_{1},x_{2},x_{3})=R(x_{1},x_{1},\xi_{3},\ldots,\xi_{n})$,
where $\xi_{i}=x_{2}$ if $c_{i}=1$ and $\xi_{i}=x_{3}$ if $c_{i}=0$.
Then $110\in R''$ and $010\notin R''$, and $R''$ is one of 64 possibles
relations; Figure \ref{fig:Producing} shows how to produce a 1-isolating
relation from each of them by identification of variables and conjunction.
\end{proof}
\begin{figure}[p]
\begin{minipage}[t]{1\columnwidth}%
{\scriptsize{}\{110\}: already 1-isolating}{\scriptsize \par}

{\scriptsize{}\{000 110\}: already 1-isolating}{\scriptsize \par}

{\scriptsize{}\{100 110\}: identify x1,x2 -> \{10\}}{\scriptsize \par}

{\scriptsize{}\{000 100 110\}: R(x1,x2,x3) AND R(x2,x1,x3) = \{000
110\}}{\scriptsize \par}

{\scriptsize{}\{110 001\}: already 1-isolating}{\scriptsize \par}

{\scriptsize{}\{000 110 001\}: already 1-isolating}{\scriptsize \par}

{\scriptsize{}\{100 110 001\}: already 1-isolating}{\scriptsize \par}

{\scriptsize{}\{000 100 110 001\}: R(x1,x2,x3) AND R(x2,x1,x3) = \{000
110 001\}}{\scriptsize \par}

{\scriptsize{}\{110 101\}: already 1-isolating}{\scriptsize \par}

{\scriptsize{}\{000 110 101\}: already 1-isolating}{\scriptsize \par}

{\scriptsize{}\{100 110 101\}: identify x1,x2 -> \{10\}}{\scriptsize \par}

{\scriptsize{}\{000 100 110 101\}: R(x1,x2,x3) AND R(x2,x1,x3) = \{000
110\}}{\scriptsize \par}

{\scriptsize{}\{110 001 101\}: already 1-isolating}{\scriptsize \par}

{\scriptsize{}\{000 110 001 101\}: already 1-isolating}{\scriptsize \par}

{\scriptsize{}\{100 110 001 101\}: identify x1,x2 -> \{10 01\}}{\scriptsize \par}

{\scriptsize{}\{000 100 110 001 101\}: R(x1,x2,x3) AND R(x2,x1,x3)
= \{000 110 001\}}{\scriptsize \par}

{\scriptsize{}\{110 011\}: already 1-isolating}{\scriptsize \par}

{\scriptsize{}\{000 110 011\}: already 1-isolating}{\scriptsize \par}

{\scriptsize{}\{100 110 011\}: already 1-isolating}{\scriptsize \par}

{\scriptsize{}\{000 100 110 011\}: already 1-isolating}{\scriptsize \par}

{\scriptsize{}\{110 001 011\}: already 1-isolating}{\scriptsize \par}

{\scriptsize{}\{000 110 001 011\}: already 1-isolating}{\scriptsize \par}

{\scriptsize{}\{100 110 001 011\}: identify x1,x2 -> \{10 01\}}{\scriptsize \par}

{\scriptsize{}\{000 100 110 001 011\}: R(x1,x2,x3) AND R(x1,x3,x2)
= \{000 100 011\}}{\scriptsize \par}

{\scriptsize{}\{110 101 011\}: already 1-isolating}{\scriptsize \par}

{\scriptsize{}\{000 110 101 011\}: already 1-isolating}{\scriptsize \par}

{\scriptsize{}\{100 110 101 011\}: already 1-isolating}{\scriptsize \par}

{\scriptsize{}\{000 100 110 101 011\}: already 1-isolating}{\scriptsize \par}

{\scriptsize{}\{110 001 101 011\}: already 1-isolating}{\scriptsize \par}

{\scriptsize{}\{000 110 001 101 011\}: already 1-isolating}{\scriptsize \par}

{\scriptsize{}\{100 110 001 101 011\}: identify x1,x2 -> \{10 01\}}{\scriptsize \par}

{\scriptsize{}\{000 100 110 001 101 011\}: R(x1,x2,x3) AND R(x1,x3,x2)
= \{000 100 110 101 011\}}{\scriptsize \par}

{\scriptsize{}\{110 111\}: identify x1,x2 -> \{10 11\}, then R(x1,x2)
AND R(x2,x1) = \{11\}}{\scriptsize \par}

{\scriptsize{}\{000 110 111\}: identify x1,x2 -> \{00 10 11\}, then
R(x1,x2) AND R(x2,x1) = \{00 11\}}{\scriptsize \par}

{\scriptsize{}\{100 110 111\}: identify x1,x2 -> \{10 11\}, then R(x1,x2)
AND R(x2,x1) = \{11\}}{\scriptsize \par}

{\scriptsize{}\{000 100 110 111\}: identify x1,x2 -> \{00 10 11\},
then R(x1,x2) AND R(x2,x1) = \{00 11\}}{\scriptsize \par}

{\scriptsize{}\{110 001 111\}: already 1-isolating}{\scriptsize \par}

{\scriptsize{}\{000 110 001 111\}: identify x1,x3 -> \{00 11\}}{\scriptsize \par}

{\scriptsize{}\{100 110 001 111\}: already 1-isolating}{\scriptsize \par}

{\scriptsize{}\{000 100 110 001 111\}: identify x1,x3 -> \{00 11\}}{\scriptsize \par}

{\scriptsize{}\{110 101 111\}: identify x1,x2 -> \{10 11\}, then R(x1,x2)
AND R(x2,x1) = \{11\}}{\scriptsize \par}

{\scriptsize{}\{000 110 101 111\}: identify x1,x2 -> \{00 10 11\},
then R(x1,x2) AND R(x2,x1) = \{00 11\}}{\scriptsize \par}

{\scriptsize{}\{100 110 101 111\}: identify x1,x2 -> \{10 11\}, then
R(x1,x2) AND R(x2,x1) = \{11\}}{\scriptsize \par}

{\scriptsize{}\{000 100 110 101 111\}: identify x1,x2 -> \{00 10 11\},
then R(x1,x2) AND R(x2,x1) = \{00 11\}}{\scriptsize \par}

{\scriptsize{}\{110 001 101 111\}: identify x1,x3 -> \{10 11\}, then
R(x1,x2) AND R(x2,x1) = \{11\}}{\scriptsize \par}

{\scriptsize{}\{000 110 001 101 111\}: identify x1,x3 -> \{00 10 11\},
then R(x1,x2) AND R(x2,x1) = \{00 11\}}{\scriptsize \par}

{\scriptsize{}\{100 110 001 101 111\}: identify x1,x3 -> \{10 11\},
then R(x1,x2) AND R(x2,x1) = \{11\}}{\scriptsize \par}

{\scriptsize{}\{000 100 110 001 101 111\}: identify x1,x3 -> \{00
10 11\}, then R(x1,x2) AND R(x2,x1) = \{00 11\}}{\scriptsize \par}

{\scriptsize{}\{110 011 111\}: identify x1,x2 -> \{10 11\}, then R(x1,x2)
AND R(x2,x1) = \{11\}}{\scriptsize \par}

{\scriptsize{}\{000 110 011 111\}: identify x1,x2 -> \{00 10 11\},
then R(x1,x2) AND R(x2,x1) = \{00 11\}}{\scriptsize \par}

{\scriptsize{}\{100 110 011 111\}: identify x1,x2 -> \{10 11\}, then
R(x1,x2) AND R(x2,x1) = \{11\}}{\scriptsize \par}

{\scriptsize{}\{000 100 110 011 111\}: identify x1,x2 -> \{00 10 11\},
then R(x1,x2) AND R(x2,x1) = \{00 11\}}{\scriptsize \par}

{\scriptsize{}\{110 001 011 111\}: identify x1,x3 -> \{11\}}{\scriptsize \par}

{\scriptsize{}\{000 110 001 011 111\}: identify x1,x3 -> \{00 11\}}{\scriptsize \par}

{\scriptsize{}\{100 110 001 011 111\}: identify x1,x3 -> \{11\}}{\scriptsize \par}

{\scriptsize{}\{000 100 110 001 011 111\}: identify x1,x3 -> \{00
11\}}{\scriptsize \par}

{\scriptsize{}\{110 101 011 111\}: identify x1,x2 -> \{10 11\}, then
R(x1,x2) AND R(x2,x1) = \{11\}}{\scriptsize \par}

{\scriptsize{}\{000 110 101 011 111\}: identify x1,x2 -> \{00 10 11\},
then R(x1,x2) AND R(x2,x1) = \{00 11\}}{\scriptsize \par}

{\scriptsize{}\{100 110 101 011 111\}: identify x1,x2 -> \{10 11\},
then R(x1,x2) AND R(x2,x1) = \{11\}}{\scriptsize \par}

{\scriptsize{}\{000 100 110 101 011 111\}: identify x1,x2 -> \{00
10 11\}, then R(x1,x2) AND R(x2,x1) = \{00 11\}}{\scriptsize \par}

{\scriptsize{}\{110 001 101 011 111\}: identify x1,x3 -> \{10 11\},
then R(x1,x2) AND R(x2,x1) = \{11\}}{\scriptsize \par}

{\scriptsize{}\{000 110 001 101 011 111\}: identify x1,x3 -> \{00
10 11\}, then R(x1,x2) AND R(x2,x1) = \{00 11\}}{\scriptsize \par}

{\scriptsize{}\{100 110 001 101 011 111\}: identify x1,x3 -> \{10
11\}, then R(x1,x2) AND R(x2,x1) = \{11\}}{\scriptsize \par}

{\scriptsize{}\{000 100 110 001 101 011 111\}: identify x1,x3 -> \{00
10 11\}, then R(x1,x2) AND R(x2,x1) = \{00 11\}}{\scriptsize \par}%
\end{minipage}

\protect\caption[Producing a 1-isolating relation from every 3-ary relation $R$ satisfying
$110\in R$ and $010\protect\notin R$ for the proof of \prettyref{lem:iso}]{\emph{}\label{fig:Producing}Producing a 1-isolating relation from
every 3-ary relation $R$ satisfying $110\in R$ and $010\protect\notin R$
for the last case of the proof of \prettyref{lem:iso}. This list
is generated by the \texttt{main}-function of the class \texttt{Isolating}
of \noun{SatConn}.}
\end{figure}

\begin{lem}
\label{lem:red}If a finite set of logical relations \emph{$\mathcal{S}$
}is not  safely tight, then \noun{st-}\emph{\noun{Conn\textsubscript{C}($\mathcal{S}$)}}
$\leq_{m}^{p}$\emph{\noun{st-}}\noun{Conn($\mathcal{S}$)}, and for
every \noun{CNF\textsubscript{C}($\mathcal{S}$)}-formula $\phi$
with $n$ variables and diameter $d$, there is a \noun{CNF}\emph{($\mathcal{S}$)}-formula
$\phi'$ with $O(n)$ variables and diameter $d'\geq d$.\end{lem}
\begin{proof}
Since \emph{$\mathcal{S}$ }is not safely tight, it must contain some
relation that is not safely OR-free, so we can construct an $n$-ary
1-isolating relation $R_{1}$ by \prettyref{lem:iso}. Similarly,
we can construct a $m$-ary 0-isolating relation $R_{0}$ from a not
 safely NAND-free relation. Let $\boldsymbol{a}\neq(1\cdots1)$ be
an isolated vector of $R_{0}$, and $\boldsymbol{b}\neq(0\cdots0)$
an isolated vector of $R_{1}$; w.l.o.g. assume $a_{1}=0$ and $b_{1}=1$.

Now let $\phi(x_{1},\ldots,x_{n})$ be any CNF\textsubscript{C}($\mathcal{S}$)-formula
and $\boldsymbol{s}$ and $\boldsymbol{t}$ two solutions of $\phi$.
We construct a CNF($\mathcal{S}$)-formula $\phi'$ by replacing every
occurrence of the constant 0 in $\phi$ with a new variable $y_{1}$,
and every occurrence of the constant 1 with a new variable $z_{1}$,
and appending $\wedge R_{0}(y_{1},y_{2},\ldots,y_{m})\wedge R_{1}(z_{1},z_{2},\ldots,z_{n})$
to $\phi$ (where $y_{2},\ldots,y_{m}$ and $z_{2},\ldots,z_{n}$
are further new variables). Then $\boldsymbol{s}\cdot\boldsymbol{a}\cdot\boldsymbol{b}$
and $\boldsymbol{t}\cdot\boldsymbol{a}\cdot\boldsymbol{b}$ are connected
in $G(\phi')$ iff $\boldsymbol{s}$ and $\boldsymbol{t}$ are connected
in $G(\phi)$. This also shows that the maximal diameter carries over.
\end{proof}
We can now prove \prettyref{thm:ndich}:
\begin{proof}[Proof of \prettyref{thm:ndich}]
1. This result carries over from the case with constants (\prettyref{thm:dich}).

2. This\noun{ }follows from \prettyref{thm:dich} with \prettyref{lem:red}.
\end{proof}

\subsection{\label{sub:eas}Deciding Connectivity via Constraint-Projection Separation}

We define the no-constants version of\emph{ }constraint-projection
separation in the obvious way:
\begin{defn}
\label{def:dcp-1}A set $\mathcal{S}$ of logical relations is \emph{no-constants
constraint-projection separating }(\emph{nc-CPS}), if every CNF($\mathcal{S}$)-formula
$\phi$ whose solution graph $G(\phi)$ is disconnected contains a
constraint $C_{i}$ s.t. $G(\phi_{i})$ is disconnected, where $\phi_{i}$
is the projection of $\phi$ to $\mathrm{Var}(C_{i})$.
\end{defn}
Analogously to \prettyref{lem:alg}, and using \prettyref{thm:sch},
we immediately have the following result:
\begin{lem}
If a finite set $\mathcal{S}$ of relations is nc-CPS, \noun{Conn($\mathcal{S}$)
}is in $\mathrm{P^{NP}}$. If $\mathcal{S}$ also is 0-valid, 1-valid,
or Schaefer, \noun{Conn($\mathcal{S}$)} is in \noun{P.}\end{lem}
\begin{rem}
The first part of the preceding lemma may be irrelevant, since connectivity
is in coNP$\subseteq\mathrm{P^{NP}}$ for all safely tight sets of
relations by \prettyref{lem:c47}, and all not  safely tight nc-CPS
sets we know of are 0-valid and 1-valid (\prettyref{lem:qu} below).
\end{rem}
Since CPS sets of relations are also nc-CPS, using \prettyref{lem:prjs},
we can extend the tractable class in the no-constants setting:
\begin{defn}
A set $\mathcal{S}$ of logical relations is \emph{nc-CPSS}, if $\mathcal{S}$
is safely componentwise bijunctive, safely componentwise IHSB$-$,
safely componentwise IHSB$+$, or affine, and if $\mathcal{S}$ also
is 0-valid, 1-valid, or Schaefer.\end{defn}
\begin{cor}
\label{cor:sa}If a finite set $\mathcal{S}$ of relations is nc-CPSS,
\noun{Conn($\mathcal{S}$) }is in \noun{P.}\end{cor}
\begin{example}
$R=(x\vee y\vee z)\wedge(\overline{x}\vee\overline{y}\vee z)\wedge(\overline{x}\vee y\vee\overline{z})=\left((y\vee z)\wedge(\overline{x}\vee z)\wedge(\overline{x}\vee y)\right)\vee(x\wedge\overline{y}\wedge\overline{z})$
is not CPSS, but since $R$ is 1-valid and e.g. safely componentwise
bijunctive, \noun{Conn($\{R\}$) }is in\noun{ P }by the preceding
corollary.
\end{example}
By \prettyref{lem:nst} we see that all nc-CPSS sets of relations
are also safely tight, but we can prove the following additional classes
to be nc-CPS, containing sets that are not safely tight:
\begin{defn}
A set $\mathcal{S}$ of relations is \emph{quasi componentwise bijunctive
(quasi componentwise IHSB$-$, quasi componentwise IHSB$+$, quasi
affine), }if the following holds for every relation $R$ in $\mathcal{S}$:
\begin{itemize}
\item $R$ is both 0-valid and 1-valid, and
\item $R$ is itself safely componentwise bijunctive (safely componentwise
IHSB$-$, safely componentwise IHSB$+$, affine), or the following
two conditions hold for $R$:

\begin{enumerate}
\item The all-0-vector is disconnected from the all-1-vector in $G(R)$.
\item the set\noun{ $\mathcal{S}'$} of all relations producible from $R$
by identification of variables (excluding $R$) is safely componentwise
bijunctive (safely componentwise IHSB$-$, safely componentwise IHSB$+$,
affine).
\end{enumerate}
\end{itemize}
$\mathcal{S}$ is \emph{quasi disconnecting}, if it is quasi componentwise
bijunctive, quasi componentwise IHSB$-$, quasi componentwise IHSB$+$,
or quasi affine.\end{defn}
\begin{lem}
\label{lem:qu}If a finite set $\mathcal{S}$ of relations is quasi
disconnecting, it is nc-CPS.\end{lem}
\begin{proof}
Let $\phi$ be any CNF\noun{($\mathcal{S}$)}-formula. First suppose
that \noun{$\mathcal{S}$} contains a relation $R$ where the all-0-vector
is disconnected from the all-1-vector in $G(R)$, and that $\phi$
contains a constraint $R(x_{1},\ldots,x_{n})$ with all variables
distinct. Then the projection $\phi_{P}$ of $\phi$ to $x_{1},\ldots,x_{n}$
contains $0\cdots0$ and $1\cdots1$ as solutions since every constraint
in $\phi$ is 0-valid and 1-valid. But $0\cdots0$ and $1\cdots1$
are disconnected in $G(R)$, thus $0\cdots0$ and $1\cdots1$ must
also be disconnected in $G(\phi_{P})$ (note that the solutions to
$\phi_{P}$ are a subset of the solutions to $R(x_{1},\ldots,x_{n})$).

Otherwise, if $\phi$ contains no such constraint $R(x_{1},\ldots,x_{n})$,
it is equivalent to a CNF\noun{($\mathcal{S}'$)-}formula $\phi'$,
where each constraint $C_{i}$ of $\phi$ corresponds to an equivalent
constraint $C_{i}'$ of $\phi'$ with $\mathrm{Var}(C_{i}')=\mathrm{Var}(C_{i})$.
Thus since \noun{$\mathcal{S}'$} is CPS by \prettyref{lem:prjs},
if $\phi$ is disconnected, there must be a constraint $C_{i}$ of
$\phi$ s.t. the projection of $\phi$ to $\mathrm{Var}(C_{i})$ is
disconnected.\end{proof}
\begin{cor}
\label{cor:qua}If a finite set $\mathcal{S}$ of relations is quasi
disconnecting, there is a polynomial-time algorithm for \noun{Conn($\mathcal{S}$).}\end{cor}
\begin{lem}
$R=\{0000,0001,0010,1001,1010,\enskip1100,\enskip0111,1111\}$ is
not  safely tight but quasi disconnecting.\end{lem}
\begin{proof}
It is easy to check that $R$ is not safely tight. But $R$ is quasi
componentwise bijunctive: Obviously, $R$ is 0-valid, 1-valid and
the all-0 vector is disconnected from the all-1 vector. It remains
to show that the set\noun{ $\mathcal{S}'$} of all relations producible
from $R$ by identification of variables is componentwise bijunctive.
By identifying each pair of variables in turn we get all 3-ary relations
in \noun{$\mathcal{S}'$:}
\begin{itemize}
\item identifying variable 1 with 2 gives \{$000,001,010,100,\enskip111$\}.
\item identifying variable 1 with 3 gives \{$000,001,100,\enskip111$\}.
\item identifying variable 1 with 4 gives \{$000,001,100,\enskip111$\}.
\item identifying variable 2 with 3 gives \{$000,001,011,101,111$\}.
\item identifying variable 2 with 4 gives \{$000,001,011,101,111$\}.
\item identifying variable 3 with 4 gives \{$000,\enskip011,110,111$\}.
\end{itemize}
The connected components are signified; it is easy to check they are
bijunctive by verifying that they are closed under maj (see \prettyref{lem:clos}).
All 2-ary and 1-ary relations are automatically bijunctive.\end{proof}
\begin{rem}
The \texttt{main-}function in the class \texttt{Sift} of \noun{SatConn}
enumerates all not safely tight relations that are quasi disconnecting.
There are no 3-ary such relations, and up to permutation of variables
and duality, the above relation is the only 4-ary one.
\end{rem}
We now see that the complexity of connectivity and $st$-connectivity
is ``inverted'' in some cases:
\begin{cor}
There are sets $\mathcal{S}$ of relations s.t. \noun{Conn($\mathcal{S}$)
}is in \emph{P} while \emph{\noun{st-}}\noun{Conn($\mathcal{S}$)
}is \emph{PSPACE}-complete.
\end{cor}
As in the case with constants, we have no algorithm to determine in
general if a set of relations is nc-CPS, so one may discover yet more
nc-CPS sets and thereby cases where connectivity is in P or in $\mathrm{P^{NP}}$.

\subsection{Deciding Connectivity via Self-Implication}

While for formulas with constants, all sets of relations with a tractable
connectivity problem are constraint-projection separating (assuming
P$\neq$coNP), without constants this is not true anymore: In \prettyref{lem:si}
below we show that for Horn sets \noun{$\mathcal{S}$} that are 1-valid,
there is a polynomial-time algorithm for \noun{Conn($\mathcal{S}$)};
now for example, $M=\left(x\vee\overline{y}\vee\overline{z}\right)\wedge\left(\overline{x}\vee z\right)$
is Horn and 1-valid but not CPSS, so by \prettyref{lem:sc} it is
not CPS, and from the proof of that lemma we find that $M$ is also
not nc-CPS.
\begin{defn}
A set \emph{$\mathcal{S}$} of relations\emph{ }is \emph{implicative},
if it is Horn and 1-valid or dual Horn and 0-valid.
\end{defn}
Horn relations that are 1-valid can contain no restraints, and without
constants, restraints also cannot be produced from such relations.
This makes it possible to decide connectivity in polynomial time:
\begin{lem}
\label{lem:si}If a finite set \emph{$\mathcal{S}$} of relations
is implicative, there is a polynomial-time algorithm for \noun{Conn($\mathcal{S}$).}\end{lem}
\begin{proof}
We show the proof for \emph{$\mathcal{S}$ }being Horn and 1-valid,
the dual Horn and 0-valid case is symmetric. We can decide for any
CNF\noun{($\mathcal{S}$)}-formula $\phi$ whether $G(\phi)$ is connected
as follows:.
\begin{itemize}
\item First assign all variables in positive unit-clauses; this produces
a connectivity-equivalent formula $\phi'$. Since \emph{$\mathcal{S}$
}is 1-valid, $\phi'$ contains no restraints, so $G(\phi')$ is disconnected
iff $\phi'$ has a non-empty self-implicating set by \prettyref{cor:horn conn}.
The following polynomial-time algorithm finds the largest self-implicating
set of $\phi'$ (which is the union of all self-implicating sets):

\begin{itemize}[label= ({*})]
\item Let $U$ be the set of all variables of $\phi'$. Repeat the following
as long as variables are removed:

\begin{itemize}[label= ]
\item For each $x\in U$, check if there is a clauses with $x$ as the
positive literal and all negated variables from $U$; if not, remove
$x$ from $U$.
\end{itemize}
\end{itemize}
\item Now $G(\phi)$ is connected iff $U$ is empty.
\end{itemize}
The correctness of algorithm ({*}) is easy to check by induction:
At first, $U$ includes every self-implicating set, and if $U$ includes
every self-implicating set, no variable from a self-implicating set
is removed from $U$. Further, as long as $U$ is not self-implicating,
a variable \emph{is} removed from $U$.
\end{proof}
It is tempting to extend this algorithm for formulas containing restraints,
by checking for every maximal self-implicating set if it contains
no restraint set. However, this seems to require checking an exponential
number of possibilities to find all maximal self-implicating sets;
this presumption is strongly supported by \prettyref{lem:cco} below,
which shows that connectivity is coNP-complete for such formulas.

\subsection{\noun{coNP}-Completeness for Connectivity within Schaefer}

In this subsection we prove that \noun{Conn($\mathcal{S}$)} is \noun{coNP}-complete
for all remaining Schaefer sets of relations\footnote{Note that Schaefer sets of relations that are quasi disconnecting
are also nc-CPSS or implicative.}:
\begin{lem}
\label{lem:cco}If \emph{$\mathcal{S}$} is a finite set of relations
that is Schaefer but not nc-CPSS and not implicative,\emph{ }\noun{Conn($\mathcal{S}$)
}is \noun{$\mathrm{coNP}$}-complete.\end{lem}
\begin{proof}
By \prettyref{lem:c47}, \noun{Conn}($\mathcal{S}$) is in coNP. Since
\emph{$\mathcal{S}$} is Schaefer but not nc-CPSS, it must be Horn
and contain at least one relation that is not  safely componentwise
IHSB$-$, or dual Horn and contain at least one relation that is not
safely componentwise IHSB$+$; since further \emph{$\mathcal{S}$}
is not implicative, it must in the first case also contain at least
one relation that is not 1-valid, and in the second case at least
one that is not 0-valid. Thus in the first case, the statement follows
from the Lemmas \ref{lem:Mex} and \ref{lem:m har-1} below, the second
case is symmetric.
\end{proof}
We need the following lemma:
\begin{lem}
\label{lem:obt}If a non-empty logical relation $R$ is Horn but not
1-valid, at least one of the relations $\{0\}$ or $\{01\}$ can be
obtained from $R$ by identification and permutation of variables.\end{lem}
\begin{proof}
Since $R$ is not 1-valid, but also not empty, there must be some
vector $\boldsymbol{a}\in R$ with some $a_{i}=0$. If $(0\cdots0)\in R$,
we identify all variables and obtain $\{0\}$.

Otherwise, we define the relation $R'$ by identifying all variables
$i$ with $a_{i}=0$, and then all with $a_{i}=1$. We show that $R'$
equals $\{01\}$ or $\{10\}$:

Since $(0\cdots0)$ and $(1\cdots1)$ are not in $R$, $\{00\}$ and
$\{11\}$ are not in $R'$. Further, since $R$ is Horn, it is closed
under $x\wedge y$ (see \prettyref{lem:clos}), and so the ``to $\boldsymbol{a}$
complementary'' vector $\boldsymbol{b}=\boldsymbol{a}\boldsymbol{\oplus}\boldsymbol{1}$
is not in $R$, else $\boldsymbol{a}\boldsymbol{\wedge}\boldsymbol{b}=(0\cdots0)$
were in $R$ (where $\boldsymbol{\oplus}$ and $\boldsymbol{\wedge}$
are applied coordinate-wise). Thus if $\{01\}\in R'$, $\{10\}\notin R'$,
and the other way around.
\end{proof}
The proof for coNP-hardness is by modifying the two corresponding
Lemmas for formulas with constants (\prettyref{lem:exp m} and \prettyref{lem:m har}).
While we cannot express the relation $M$ from \prettyref{lem:exp m}
as a \noun{CNF(}\emph{$\mathcal{S}$}\noun{)}-formula, we can assemble
a \noun{CNF(}\emph{$\mathcal{S}$}\noun{)}-formula $\mu$ which is
suitable for a reduction from satisfiability similar to the one of
\prettyref{lem:m har}:
\begin{lem}
\label{lem:Mex}If \emph{$\mathcal{S}$} is a finite set of non-empty
Horn relations that contains at least one relation that is not 1-valid,
and at least one relation that is not safely componentwise IHSB$-$,
then there is a \noun{CNF(}\emph{$\mathcal{S}$}\noun{)}-formula $\mu$
s.t. $\nu(\mu)$ contains no restraint clauses of size greater than
1 and no self-implicating sets, and s.t. $\nu(\mu)$ contains the
unit-clause $\overline{v_{0}}$ and the clause $x\vee\overline{y}\vee\overline{z}$,
s.t. $y$ is not implied and $x$ is implied only via the clause $x\vee\overline{y}\vee\overline{z}$.\end{lem}
\begin{proof}
We modify the proof of \prettyref{lem:exp m}.

Again, let $\phi_{0}=\nu(R)$. The first transformation remains unchanged:
\begin{enumerate}
\item \emph{Obtain a not  componentwise IHSB$-$ formula $\phi_{1}^{\sim}$
from $\phi_{0}$ by identification of variables.}
\end{enumerate}
We renamed the resulting formula since it will not directly be used
as input for step 2; namely, we insert the following step:
\begin{itemize}[label=\#.]
\item Identify all variables in negative unit-clauses and call the resulting
variable $v_{0}$, then identify all variables in positive unit-clauses
and call the resulting variable $v_{1}$. If there were no negative
unit-clauses, add the clauses $\overline{v_{0}}$ or $\overline{v_{0}}\wedge v_{1}$
obtained by \prettyref{lem:obt} from a not 1-valid Horn relation.
\end{itemize}
The resulting formula $\phi_{\#}$ can be written as $\phi_{\#}=\phi_{1}\wedge\overline{v_{0}}$
or $\phi_{\#}=\phi_{1}\wedge\overline{v_{0}}\wedge v_{1}$, where
$\phi_{1}$ contains no unit-clauses, and s.t. $v_{0}$ and $v_{1}$
do not appear in $\phi_{1}$. It is clear that $\phi_{\#}$ still
is not componentwise IHSB$-$. The following transformations will
only affect the part of the formula without unit-clauses, and our
notation will be such that the (entire) formula resulting from step
$i$ is $\phi_{i}\wedge\overline{v_{0}}$ resp. $\phi_{i}\wedge\overline{v_{0}}\wedge v_{1}$.

We use $\phi_{1}$ as input for the next transformation, which replaces
step 2 of the original proof. Again, let $\left[\phi_{1}^{*}\right]$
be a connected component of $\left[\phi_{1}\right]$ that is not IHSB$-$,
and let $U$ be the set of variables assigned 1 in the minimum solution
of $\phi_{1}^{*}$.
\begin{enumerate}[resume]
\item \emph{Identify all} \emph{variables from $U$.}
\end{enumerate}
We show that the resulting formula $\phi_{2}$ has the same crucial
properties as the formula $\phi_{2}$ in proof of \prettyref{lem:exp m}:
\begin{itemize}[label= ]
\item W.l.o.g., assume $U$ was not empty; then the vector $\boldsymbol{m}$
resulting from the minimum solution of $\phi_{1}^{*}$ has one variable
$x_{m}$ assigned 1 and all others 0. All vectors resulting from $\phi_{1}^{*}$
belong to the same component $\phi_{2}^{*}$. In comparison to the
relation $\phi_{2}^{*}$ of the original proof, $\phi_{2}^{*}$ contains
additional vectors, resulting from vectors of $\phi_{1}$ having all
variables from $U$ assigned 0.\\
Since $\phi_{1}$ contained no unit clauses, this also holds for $\phi_{2}$,
so $\phi_{2}$ contains the all-0 vector, which is connected to $\boldsymbol{m}$.
It follows that also here, $\phi_{2}^{*}$ has the all-0 vector as
minimum solution. Further, $\phi_{2}^{*}$ also here is not componentwise
IHSB$-$:

\begin{itemize}[label= ]
\item Let $a,b,c$ be vectors from $[\phi_{1}^{*}]$ s.t. $a\wedge(b\vee c)$
is not in $[\phi_{1}^{*}]$. Since $a,b,c$ all have 1 assigned to
all variables from $U$, this also holds for $a\wedge(b\vee c)$,
so for the vectors $a',b',c'$ resulting from the identification in
step 2, $a'\wedge(b'\vee c')\notin[\phi_{2}^{*}]$. 
\end{itemize}
\end{itemize}
With this, the reasoning can proceed as in the original proof. We
can perfectly simulate the next step (note that the entire formula
after step 2 is $\phi_{2}\wedge\overline{v_{0}}$ or $\phi_{2}\wedge\overline{v_{0}}\wedge v_{1}$):
\begin{enumerate}[resume]
\item \emph{Identify all variables of $\phi_{2}$ not implied by Var($c^{*}$)
with $v_{0}$.}
\end{enumerate}
It is clear that this has the same effect on $\phi_{2}$ as a substitution
with 0, so as in the original proof, we now have a clause $c=x\vee\overline{y}\vee\overline{z}_{1}\vee\cdots\vee\overline{z}_{k}$
($k\geq3$) s.t. $y$ is not implied. The next step remains unchanged:
\begin{enumerate}[resume]
\item \emph{Identify $z_{1},\ldots,z_{k}$, call the resulting variable
$z$.}
\end{enumerate}
The resulting formula $\phi_{4}\wedge\overline{v_{0}}$ resp. $\phi_{4}\wedge\overline{v_{0}}\wedge v_{1}$
now contains the clause $x\vee\overline{y}\vee\overline{z}$, and
$y$ is not implied; further, it still contains no restraints of size
greater than 1. However, step 4 may have produced self-implicating
sets. Before we deal with this problem we append one more step to
ensure that $x$ is implied only via $x\vee\overline{y}\vee\overline{z}$.

Since $x\notin\mathrm{Imp}(y)$, $x\notin\mathrm{Imp}(z)$, and $x$
is implied by $\{y,z\}$ only via $x\vee\overline{y}\vee\overline{z}$,
each other implication clause with $x$ as positive literal must have
at least one negative literal with a variable not implied by $\{y,z\}$,
so the following step eliminates all other clauses with $x$ as positive
literal:
\begin{enumerate}[resume]
\item \emph{If a variable $x_{i}\notin\mathrm{Imp}(\{y,z\})$ appears in
a negative literal of an implication clause having $x$ as positive
literal, identify $x_{i}$ with $x$; repeat as long as there is such
a variable.}
\end{enumerate}
We take into account that also step 5 may have produced self-implicating
sets (although that seems not possible).

To show how to eliminate the self-implicating sets which may have
been produced in the last two steps, we first prove that $\phi_{5}$
is still is not componentwise IHSB$-$:
\begin{itemize}[label= ]
\item By the proof of \prettyref{lem:exp m}, we know that we can express
a not componentwise IHSB$-$ relation (i.e., $K,L,$ or $M$) from
$\phi_{5}$ by identification of variables and substitution of constants.
We prove that any formula obtained from a componentwise IHSB$-$ one
by identification and substitution is also componentwise IHSB$-$;
the statement then follows by reversal:

\begin{itemize}[label= ]
\item A relation obtained from a componentwise IHSB$-$ one by identification
of variables is componentwise IHSB$-$ by definition. For substitution,
consider a formula $\psi=\psi_{1}\vee\cdots\vee\psi_{k}$ where $\psi_{1},\ldots,\psi_{k}$
are the connected components that can be written as IHSB$-$ formulas.
Then the formula $\psi'$ resulting from the substitution is equivalent
to $\psi'_{1}\vee\cdots\vee\psi'_{k}$, where each $\psi'_{i}$ is
obtained from $\phi_{i}$ by substitution, and thus is IHSB$-$ also
(some $\psi'_{i}$ may be empty). Since any two vectors of $\psi'$
resulting from vectors of different components of $\psi$ differ in
at least two variables, all $\psi'_{i}$ are disconnected. Now by
\prettyref{lem:l41}, the connected components of each $\phi'_{i}$
are IHSB$-$ since IHSB$-$ relations are characterized by closure
under an idempotent operation (see \prettyref{lem:clos}). It follows
that $\psi'$ is componentwise IHSB$-$.
\end{itemize}
\end{itemize}
Now if $\phi_{5}$ contains self-implicating sets, we repeat steps
1 to 5 (skipping step \#), with $\phi_{5}\wedge\overline{v_{0}}$
resp. $\phi_{5}\wedge\overline{v_{0}}\wedge v_{1}$ as input to step
1, until we obtain a formula containing no self-implicating sets.
This leads to a formula with all the demanded properties since the
input to step 1 always is not safely componentwise IHSB$-$, and since
the recursion must terminate as variables are removed in every pass.
\begin{itemize}[label= ]
\item \emph{Remark:} We cannot execute or simulate the next step of \prettyref{lem:exp m}
since the formula may contain no positive unit-clause, so we are not
able to produce $M$. We could simplify the formula more, but this
would be quite involved, and we can already use the formula for the
reduction from satisfiability.
\end{itemize}
\end{proof}
\begin{lem}
\noun{\label{lem:m har-1} Conn}\emph{(}\textup{$\left\{ \mu\right\} $)}
is \noun{coNP}-hard for every Horn\noun{ }formula $\mu$ containing
no restraint clauses of size greater than 1 and no self-implicating
sets, and containing the unit-clause $\overline{v_{0}}$ and the clause
$x\vee\overline{y}\vee\overline{z}$, s.t. $y$ is not implied and
$x$ is implied only via the clause $x\vee\overline{y}\vee\overline{z}$.\end{lem}
\begin{proof}
We modify the proof of \prettyref{lem:m har}: The CNF\textsubscript{C}($\{M\}$)-formula
$\phi$ is replaced by a connectivity-equivalent CNF($\{\mu\}$)-formula
$\phi'$; we assemble $\phi'$ by replacing the gadgets from which
$\phi$ is built.

Clearly, $\mu$ can be written as $\mu=\lambda\wedge\varepsilon$,
where $\varepsilon$ contains only unit clauses (including $\overline{v_{0}}$),
and $\lambda$ contains no unit clauses and no variables appearing
in unit clauses. Let the variables of $\lambda$ be ordered such that
$x,y,z$ are the first three variables, in this order; let $r$ be
the arity of $\lambda$. 
\begin{itemize}[label= ]
\item \emph{Remark:} The construction would be easy if we could structurally
express $M$  from $\mu$; however, we see no easy way to that. $\phi'$
will satisfy a relation $\phi'=\exists\ldots\phi$, but this will
be no structural expression in general.
\end{itemize}
We first show how to simulate the implication clause $\overline{u}\vee w$:
\begin{itemize}[label= ]
\item Since $\lambda$ contains no unit-clauses and no restraints, it has
both the all-0- and the all-1-vector as solution. \\
First assume that $z$ is not implied by $x$. Starting from the all-0
vector, then setting all variables implied by $x$ to 1, we see that
$\boldsymbol{a}=(1,0,0,a_{4},\ldots,a_{r})$ is a solution to $\lambda$
for some constants $a_{4},\ldots,a_{r}$. By the clause $x\vee\overline{y}\vee\overline{z}$,
the to $\boldsymbol{a}$ complementary vector $\boldsymbol{a\oplus1}$
is no solution, so identifying all variables $x_{i}$ where $a_{i}=0$
with $y$, and all where $a_{i}=1$ with $x$ results in $x\vee\overline{y}$.\\
Now assume that $z$ is implied by $x$. Since $x$ is not implied
by $z$ by simplification rule \ref{enu:red imp}, there is some solution
$\boldsymbol{a}$ with $a_{1}=0$ and $a_{3}=1$. Since $z$ is implied
by $x$, the to $\boldsymbol{a}$ complementary vector is no solution.
So again, identifying all variables $x_{i}$ with $a_{i}=0$ with
$x$, and all with $a_{i}=1$ with $z$ results in $\overline{x}\vee z$.\\
It follows that we can in both cases express $(\overline{u}\vee w)\wedge\varepsilon$
as a CNF($\{\mu\}$)-formula.
\end{itemize}
Now we construct $\phi'$, by replacing the gadgets from which the
formula $\phi$ in the proof of \prettyref{lem:m har} is built by
CNF($\{\mu\}$)-expressions; we make sure that all clauses of $\phi$
are entailed by the replacements, and no replacement contains a self-implicating
set, which will turn out useful:

We replace $\overline{q}_{p}\vee a_{pl}$ by $\left(\overline{q}_{p}\vee a_{pl}\right)\wedge\varepsilon$,
and $\overline{b}_{pl}\vee q_{(p+1)\,\mathrm{mod}\,m}$ by $\left(\overline{b}_{pl}\vee q_{(p+1)\,\mathrm{mod}\,m}\right)\wedge\varepsilon$.

We replace $\overline{x}_{i}\vee\overline{x}_{j}$ by $\overline{x_{i}}\vee v_{0}$
if $x_{i}=x_{j}$, and otherwise by 
\begin{equation}
\lambda(v_{0},x_{i}x_{j},w_{ij4},\ldots,w_{ijr})\wedge\varepsilon,\label{eq:r1}
\end{equation}
where for each $i,j,$ $w_{ij4},\ldots,w_{ijr}$ are new variables.
Clearly, $\overline{x}_{i}\vee\overline{x}_{j}$ is entailed by the
replacement. Since no variables are identified in $\lambda$, there
are no self-implicating sets in \eqref{eq:r1}.

Finally, we replace $\left(\overline{x}_{l}\vee\overline{a}_{pl}\vee b_{pl}\right)\wedge\left(\overline{b}_{pl}\vee x_{l}\right)$
by 
\begin{equation}
\left(\lambda(b_{pl},a_{pl},x_{l},z_{pl4},\ldots,z_{plr})\wedge\varepsilon\right)\wedge\left(\left(\overline{b}_{pl}\vee x_{l}\right)\wedge\varepsilon\right),\label{eq:r2}
\end{equation}
where for each $p,l$, $z_{pl4},\ldots,z_{plr}$ are new variables.
By definition of $\lambda$, \prettyref{eq:r2} contains the clause
$\overline{x}_{l}\vee\overline{a}_{pl}\vee b_{pl}$. We show that
\prettyref{eq:r2} contains no self-implicating set:
\begin{itemize}[label=]
\item Clearly, we can ignore $\varepsilon$. Since no variables are identified
in $\lambda(b_{pl},a_{pl},x_{l},z_{pl4},\ldots,z_{plr})$, this expression
contains no self-implicating set. Since $a_{pl}$ is not implied in
\prettyref{eq:r2}, and $b_{pl}$ is implied only via $\overline{x}_{l}\vee\overline{a}_{pl}\vee b_{pl}$,
$b_{pl}$ cannot belong to a self-implicating set in \prettyref{eq:r2}.
Thus proceeding from $\lambda(b_{pl},a_{pl},x_{l},z_{pl4},\ldots,z_{plr})$
to the whole expression \prettyref{eq:r2}, it is easy to see that
the clause $\overline{b}_{pl}\vee x_{l}$ cannot have produced a self-implicating
set.
\end{itemize}
Now note that not only each clause $\overline{x}_{i}\vee\overline{x}_{j}$
of $\phi$ is entailed by its replacement, but
\begin{itemize}[label= ({*})]
\item the only restraint set of $\phi'$ is $\{v_{0}\}$, and a subset
$U$ of $\mathrm{Var}(\phi)$ implies $v_{0}$ in $\phi'$ exactly
if $U$ implies a restraint set $\{x_{i},x_{j}\}$ in $\phi$.
\end{itemize}
With this, we are ready to show that $\phi'$ is connected iff $\phi$
is connected, i.e., that $\phi'$ contains a maximal self-implicating
set $U$ containing no restraint set iff $\phi$ does:

Since $\phi'$ contains all implication clauses of $\phi$, $\phi'$
also contains all self-implicating sets of $\phi$. By ({*}), each
maximal self-implicating set containing no restraint set of $\phi$
can be extended to a maximal self-implicating set containing no restraint
set of $\phi'$.

For the converse, first recall that all used gadgets are guaranteed
to contain no self-implicating sets, and note that all additional
variables of $\phi'$ except $v_{0}$ (which implies no other variable)
each appear in only one gadget, so that no ``shortcuts'' are introduced.
With this, a consideration analogously to the one for $\phi$ in the
original proof shows that any maximal self-implicating set $U'$ of
$\phi'$ must contain all $q_{p}$, all $a_{pl}$, and for every $p$
for at least one $l\in\{i_{p},j_{p},k_{p}\}$ both $b_{pl}$ and $x_{l}$,
so that $U'$ must contain some maximal self-implicating set $U$
of $\phi$ as subset. By ({*}), if $U'$ contains no restraint set,
the same holds for $U$.
\end{proof}

\subsection{\label{sub:Towards}Reductions for Connectivity}

The CNF($\mathcal{S}$)-formula $\phi'$ constructed from a CNF\textsubscript{C}($\mathcal{S}$)-formula
$\phi$ in the proof of \prettyref{lem:red} using 0- and 1-isolating
relations may contain multiple components even if $\phi$ has only
one component, so that construction cannot be used for the connectivity
problem; But if we use relations with unique solutions instead, the
number of components is retained, so that analogous to \prettyref{lem:red},
a reduction is possible:
\begin{defn}
A formula $\phi$ is \emph{0-unique }(\emph{1-unique}) if it has exactly
one solution $\boldsymbol{a}$ s.t. $\boldsymbol{a}\neq(1\cdots1)$
($\boldsymbol{a}\neq(0\cdots0)$).\end{defn}
\begin{lem}
\label{lem:uni}Let \emph{$\mathcal{S}$} be a finite set of logical
relations.\emph{ }If there is a 0-unique and a 1-unique \noun{CNF($\mathcal{S}$)}-formula,
then \emph{\noun{Conn\textsubscript{C}($\mathcal{S}$)}}$\leq_{m}^{p}$\noun{Conn($\mathcal{S}$).}
\end{lem}
Using a result from Creignou et al., we can determine exactly for
which \emph{$\mathcal{S}$} this is the case:
\begin{lem}
\label{lem:red-1-1}Let \emph{$\mathcal{S}$} be a finite set of logical
relations. There is a 0-unique and a 1-unique \noun{CNF($\mathcal{S}$)}-formula
exactly\emph{ }if \emph{$\mathcal{S}$} contains at least one relation
that is not 0-valid, at least one relation that is not 1-valid, and
at least one relation that is not complementive.\end{lem}
\begin{proof}
It is easy to see that if every relation in \emph{$\mathcal{S}$ }is
0-valid, there is no 1-unique\emph{ }\noun{CNF($\mathcal{S}$)}-formula,
if every relation is 1-valid, there is no 0-unique one, and if if
every relation in \emph{$\mathcal{S}$ }is complementive, there is
neither a 0-unique nor a 1-unique \noun{CNF($\mathcal{S}$)}-formula.

Otherwise, Lemma 4.13 of \citep{creignou1996complexity} shows that
in this case, $\overline{x}\wedge y$ is expressible as a \noun{CNF($\mathcal{S}$)}-formula,
which is both 0-unique and 1-unique. \end{proof}
\begin{cor}
\label{cor:red-1}Let \emph{$\mathcal{S}$} be a finite set of logical
relations.\emph{ }If \emph{$\mathcal{S}$} contains at least one relation
that is not 0-valid, at least one relation that is not 1-valid, and
at least one relation that is not complementive, then \emph{\noun{Conn\textsubscript{C}($\mathcal{S}$)}}$\leq_{m}^{p}$\noun{Conn($\mathcal{S}$).}
\end{cor}
So for such sets \emph{\noun{$\mathcal{S}$,}} we can transfer the
hardness results for formulas with constants (\prettyref{thm:trich})
to the no-constants case:
\begin{cor}
\label{cor:ns}If \emph{$\mathcal{S}$} is a finite set of relations
that is  safely tight but not Schaefer and contains at least one relation
that is not 0-valid, at least one relation that is not 1-valid, and
at least one relation that is not complementive,\noun{ Conn($\mathcal{S}$)}
is \noun{coNP}-complete.
\end{cor}

\begin{cor}
\label{cor:nr}If \emph{$\mathcal{S}$} is a finite set of relations
that is not  safely tight and contains at least one relation that
is not 0-valid, at least one relation that is not 1-valid, and at
least one relation that is not complementive,\noun{ Conn($\mathcal{S}$)}
is \emph{PSPACE}-complete.
\end{cor}
Here we end our investigation of no-constants formulas. It is easy
to check that thus, the complexity of \noun{Conn($\mathcal{S}$) }remains
open for sets $\mathcal{S}$ that are\noun{ }0-valid, 1-valid, or
complementive, but not Schaefer, nor nc-CPSS, nor quasi disconnecting\noun{
(}See also Table \noun{\ref{tab:cl}).}
\begin{example}
The relation $R_{{\rm {NAE}}}=\{0,1\}^{3}\setminus\{000,111\}$ is
complementive, but not safely tight and not quasi disconnecting, so
we only know that \noun{Conn($\mathcal{S}$) }is in PSPACE.
\end{example}
\newpage{}

\section{Quantified Constraints}

Now we look in the other direction and examine connectivity for more
powerful versions of CNF\textsubscript{C}($\mathcal{S}$)-formulas
by allowing quantifiers. Since it is easy to transform any quantified
formula into one in prenex normal form, we will assume all formulas
to have that form, i.e.
\[
Q_{1}y_{1}\cdots Q_{m}y_{m}\phi(y_{1},\ldots,y_{m},x_{1},\ldots,x_{n}),
\]
where $\phi$ is a CNF\textsubscript{C}($\mathcal{S}$)-formula,
and $Q_{1},\ldots,Q_{m}\in\{\exists,\forall\}$ are quantifiers. We
call these expressions \emph{Q-CNF\textsubscript{C}($\mathcal{S}$)-formulas}
and denote the corresponding connectivity resp. $st$-connectivity
problems by \noun{Q-Conn}\textsubscript{C}($\mathcal{S}$) resp.
\noun{st-Q-Conn}\textsubscript{C}($\mathcal{S}$); the solution graph
only involves the free variables$x_{1},\ldots,x_{n}$.

Here, we can present a complete classification for both connectivity
problems and the diameter, stated in the following two theorems and
summarized in the table below:

\begin{table}[!h]
\begin{tabular}{|c||c|c||c|c|}
\hline 
{\footnotesize{}$\mathcal{S}$} & \noun{\footnotesize{}Q-Conn\textsubscript{C}}{\footnotesize{}($\mathcal{S}$) } & \noun{\footnotesize{}Conn\textsubscript{C}}{\footnotesize{}($\mathcal{S}$) } & \noun{\footnotesize{}st-Q-Conn\textsubscript{C}}{\footnotesize{}($\mathcal{S}$) } & \noun{\footnotesize{}st-Conn\textsubscript{C}}{\footnotesize{}($\mathcal{S}$) }\tabularnewline
\hline 
\hline 
{\footnotesize{}not safely tight} & \multirow{2}{*}{{\footnotesize{}PSPACE-c.}} & {\footnotesize{}PSPACE-c.} & \multirow{2}{*}{{\footnotesize{}PSPACE-c.}} & {\footnotesize{}PSPACE-c.}\tabularnewline
\cline{1-1} \cline{3-3} \cline{5-5} 
{\footnotesize{}safely tight, not Schaefer} &  & \multirow{3}{*}{{\footnotesize{}coNP-c.}} &  & \multirow{6}{*}{{\footnotesize{}in P}}\tabularnewline
\cline{1-2} \cline{4-4} 
{\footnotesize{}Horn, not c. I$-$ /} & \multirow{4}{*}{{\footnotesize{}coNP-c.}} &  & \multirow{5}{*}{{\footnotesize{}in P}} & \tabularnewline
{\footnotesize{}dual Horn, not c. I$+$} &  &  &  & \tabularnewline
\cline{1-1} \cline{3-3} 
{\footnotesize{}Horn, c. I$-$, not I$-$ /} &  & \multirow{3}{*}{{\footnotesize{}in P}} &  & \tabularnewline
{\footnotesize{}dual Horn, c. I$+$, not I$+$} &  &  &  & \tabularnewline
\cline{1-2} 
{\footnotesize{}bijunctive / affine / I$-$ / I$+$} & {\footnotesize{}in P} &  &  & \tabularnewline
\hline 
\end{tabular}

\protect\caption[The classifications for Q-CNF\protect\textsubscript{C}($\mathcal{S}$)-formulas]{\emph{}The classifications for Q-CNF\protect\textsubscript{C}($\mathcal{S}$)-formulas,
in comparison to the case without quantifiers.\protect \\
c. = componentwise $\quad$I$-$ = IHSB$-$ $\quad$I$+$ = IHSB$+$}
\end{table}

\begin{thm}[Dichotomy theorem for\noun{ Q-CNF\textsubscript{C}}($\mathcal{S}$)\emph{-}formulas]
\label{thm:qdich} Let $\mathcal{S}$ be a finite set of logical
relations.
\begin{enumerate}
\item If $\mathcal{S}$ is Schaefer,\noun{ st-Q-Conn\textsubscript{C}($\mathcal{S}$)}
is in \noun{P, Q-Conn}\textsubscript{\noun{C}}\noun{($\mathcal{S}$)}
is in \noun{coNP}, and for every Q-CNF\textsubscript{C}($\mathcal{S}$)-formula
$\phi$, the diameter of $G(\phi)$ is linear in the number of free
variables.
\item Otherwise, both \noun{st-Q-Conn\textsubscript{C}($\mathcal{S}$)}
and \noun{Q-Conn\textsubscript{C}($\mathcal{S}$)} are \noun{PSPACE}-complete,
and there are Q-CNF\textsubscript{C}($\mathcal{S}$)-formulas $\phi$,
such that the diameter of $G(\phi)$ is exponential in the number
of free variables.
\end{enumerate}
\end{thm}
\begin{proof}
See \prettyref{sub:di}.\end{proof}
\begin{thm}[Trichotomy theorem for \noun{Q-Conn}\textsubscript{C}($\mathcal{S}$)]
\label{thm:qtrich} Let $\mathcal{S}$ be a finite set of logical
relations.
\begin{enumerate}
\item If $\mathcal{S}$ is bijunctive, IHSB$-$, IHSB$+$ or affine, \noun{Q-Conn\textsubscript{C}($\mathcal{S}$)}
is in \noun{P}.
\item Else if $\mathcal{S}$ is Schaefer, \noun{Q-Conn}\textsubscript{\noun{C}}\noun{($\mathcal{S}$)}
is \noun{coNP}-complete.
\item Else,\emph{ }\noun{Q-Conn}\textsubscript{\noun{C}}\noun{($\mathcal{S}$)}
is \noun{PSPACE}-complete.
\end{enumerate}
\end{thm}
\begin{proof}
1. See Lemmas \ref{lem:ihsb}, \ref{lem:bijunc}, and \ref{lem:affin}.

2. See \prettyref{cor:coq}.

3. This follows from \prettyref{thm:qdich}.
\end{proof}

\subsection{\label{sub:di}Properties that Persist}

Recall from \prettyref{lem:clos} that bijunctive, Horn / dual Horn,
affine, and IHSB$-$ / IHSB$+$ relations are characterized by closure
properties. We begin by showing that these properties are retained
when we quantify over some variables.
\begin{lem}
Let $R$ be a logical relation that is closed under the coordinate-wise
application of some operation $f$.
\begin{enumerate}
\item The relation obtained by quantifying existentially over some variable
of $R$ is also closed under $f$.
\item If $f$ is not constant, the relation obtained by quantifying universally
over some variable of $R$ is also closed under $f$.
\end{enumerate}
\end{lem}
\begin{proof}
1. Let $R$ be a $n+1$-ary relation, consisting of $m$ vectors $(a^{i},b_{1}^{i},\ldots,b_{n}^{i})$,
$i=1,\ldots,m$, that is closed under the coordinate-wise application
of the $k$-ary relation $f$, i.e. 
\[
\left(f(a^{i_{1}},\ldots,a^{i_{k}}),\,f(b_{1}^{i_{1}},\ldots,b_{1}^{i_{k}}),\ldots,\,f(b_{n}^{i_{1}},\ldots,b_{n}^{i_{k}})\right)\in R
\]
 for all $1\leq i_{1},\ldots,i_{k}\leq m$. Let the relation $R'=\exists x\,R(x,\boldsymbol{y})$
be obtained, w.l.o.g., by quantifying existentially over the first
variable. If then $\boldsymbol{b}^{i_{1}},\ldots,\boldsymbol{b}^{i_{k}}\in R'$,
also $\left(f(b_{1}^{i_{1}},\ldots,b_{1}^{i_{k}}),\ldots,\,f(b_{n}^{i_{1}},\ldots,b_{n}^{i_{k}})\right)\in R'$
since for each $\boldsymbol{b}^{i}\in R'$, $(0,\boldsymbol{b}^{i})\in R$
or $(1,\boldsymbol{b}^{i})\in R$. Thus $R$ also contains
\[
\left(0,\,f(b_{1}^{i_{1}},\ldots,b_{1}^{i_{k}}),\ldots,\,f(b_{n}^{i_{1}},\ldots,b_{n}^{i_{k}})\right)
\]
 or 
\[
\left(1,\,f(b_{1}^{i_{1}},\ldots,b_{1}^{i_{k}}),\ldots,\,f(b_{n}^{i_{1}},\ldots,b_{n}^{i_{k}})\right).
\]

2. Let $R$ and $f$ be as in 1., but $f$ not constant, and let $R'=\forall x\,R(x,\boldsymbol{y})$.
If then $\boldsymbol{b}^{i_{1}},\ldots,\boldsymbol{b}^{i_{k}}\in R'$,
also $\left(f(b_{1}^{i_{1}},\ldots,b_{1}^{i_{k}}),\ldots,\,f(b_{n}^{i_{1}},\ldots,b_{n}^{i_{k}})\right)\in R'$:
Since $f$ is not constant, we can chose values $a_{0}^{i_{1}},\ldots,a_{0}^{i_{k}}\in\{0,1\}$
such that $f(a_{0}^{i_{1}},\ldots,a_{0}^{i_{k}})=0$, and $a{}_{1}^{i_{1}},\ldots,a{}_{1}^{i_{k}}\in\{0,1\}$
such that $f(a{}_{1}^{i_{1}},\ldots,a{}_{1}^{i_{k}})=1$. Since for
each $\boldsymbol{b}^{i}\in R'$ both $(0,\boldsymbol{b}^{i})\in R$
and $(1,\boldsymbol{b}^{i})\in R$, $R$ also contains both 
\[
\left(f(a_{0}^{i_{1}},...,a_{0}^{i_{k}}),\,f(b_{1}^{i_{1}},...,b_{1}^{i_{k}}),...,\,f(b_{n}^{i_{1}},...,b_{n}^{i_{k}})\right)=\left(0,\,f(b_{1}^{i_{1}},...,b_{1}^{i_{k}}),...,\,f(b_{n}^{i_{1}},...,b_{n}^{i_{k}})\right)
\]
 and 
\[
\left(f(a{}_{1}^{i_{1}},...,a{}_{1}^{i_{k}}),\,f(b_{1}^{i_{1}},...,b_{1}^{i_{k}}),...,\,f(b_{n}^{i_{1}},...,b_{n}^{i_{k}})\right)=\left(1,\,f(b_{1}^{i_{1}},...,b_{1}^{i_{k}}),...,\,f(b_{n}^{i_{1}},...,b_{n}^{i_{k}})\right).
\]
\end{proof}
\begin{cor}
\label{cor:quanti closed}Any relation obtained by arbitrarily quantifying
over some variables of a bijunctive (Horn, dual Horn, affine, IHSB$-$,
IHSB$+$) relation is itself bijunctive (Horn, dual Horn, affine,
IHSB$-$, IHSB$+$).
\end{cor}
We are now ready to prove the dichotomy theorem:
\begin{proof}[Proof of \prettyref{thm:qdich}]
 1. We first show that the structural properties proved for \noun{CNF\textsubscript{C}($\mathcal{S}$)}-formulas
with safely tight sets $\mathcal{S}$ in \prettyref{sec:stru} also
hold for Q-CNF\textsubscript{C}($\mathcal{S}$)-formulas if $\mathcal{S}$
is Schaefer.

If $\mathcal{S}$ is bijunctive or Horn (dual Horn), any CNF\textsubscript{C}($\mathcal{S}$)-formula
$\phi$ is itself bijunctive, resp. Horn (dual Horn), so by \prettyref{cor:quanti closed},
this also holds for any Q-CNF\textsubscript{C}($\mathcal{S}$)-formula.
Now since by \prettyref{lem:l42} every bijunctive relation is safely
componentwise bijunctive, and every Horn resp. dual Horn relation
is safely OR-free resp. safely NAND-free, the structural properties
stated in \prettyref{lem:l43} and \prettyref{lem:l45} apply to Q-CNF\textsubscript{C}($\mathcal{S}$)-formulas
also, and the statement for the diameter follows as in \prettyref{lem:c44}
and \prettyref{lem:c46}.

The above reasoning does not work for affine sets $\mathcal{S}$ since
for such $\mathcal{S}$, CNF\textsubscript{C}($\mathcal{S}$)-formula
are not necessarily affine itself. But nevertheless, the relations
obtained by arbitrarily quantifying over CNF\textsubscript{C}($\mathcal{S}$)-formulas
are expressible as CNF\textsubscript{C}($\mathcal{S}'$)-formulas
for some affine set $\mathcal{S}'$: Figure 4.2 of \citep{DBLP:phd/de/Bauland2007}
shows an algorithm to transform any Q-CNF\textsubscript{C}($\mathcal{S}$)
formula $\phi$ into an equivalent system of linear equations, i.e.
conjunction of affine expressions. Now since by \prettyref{lem:l42},
affine relations are safely componentwise bijunctive, safely OR-free,
and safely NAND-free\emph{, }the structural properties from \prettyref{lem:l43}
and \prettyref{lem:l45} apply to Q-CNF\textsubscript{C}($\mathcal{S}$)-formulas
also for affine $\mathcal{S}$.

It remains to prove that \noun{st-Q-Conn\textsubscript{C}($\mathcal{S}$)}
is in \noun{P }and\noun{ Q-Conn}\textsubscript{\noun{C}}\noun{($\mathcal{S}$)}
is in \noun{coNP}. The algorithms given in the proofs of \prettyref{lem:c44}
and \prettyref{lem:c46} for showing that \noun{st-Conn\textsubscript{C}($\mathcal{S}$)}
is in \noun{P }and\noun{ Conn}\textsubscript{\noun{C}}\noun{($\mathcal{S}$)}
is in \noun{coNP }are by following paths in the solution graph in
a given direction. In doing so, the formula has to be evaluated for
a certain vector at each step; for Q-CNF\textsubscript{C}($\mathcal{S}$)-formulas,
this means assigning the free variables and then evaluating the fully
quantified formula. Now by Theorem 6.1 in \citep{Schaefer:1978:CSP:800133.804350},
the evaluation problem for quantified formulas is in P, so that the
algorithms from proofs of \prettyref{lem:c44} and \prettyref{lem:c46}
can also be used to solve the problems for Q-CNF\textsubscript{C}($\mathcal{S}$)-formulas
in polynomial time.

2. By Schaefer's ``expressibility theorem'' (Theorem 3.0 of \citep{Schaefer:1978:CSP:800133.804350}),
if $\mathcal{S}$ is not Schaefer, every Boolean relation is expressible
from $\mathcal{S}$ by existentially quantifying over some CNF\textsubscript{C}($\mathcal{S}$)
formula, and thus the statements follow from \prettyref{lem:l36}
and \prettyref{lem:l37}.
\end{proof}

\subsection{coNP-Completeness for Connectivity}

It remains to determine the complexity of \noun{Q-Conn\textsubscript{C}}
for Schaefer sets of relations. We begin by showing that with quantifiers,
we can extend the coNP-complete class. Since by \prettyref{lem:m har},
connectivity for $\mathcal{S}=\left\{ \left(x\vee\overline{y}\vee\overline{z}\right)\wedge\left(\overline{x}\vee z\right)\right\} $
is coNP-hard already for \noun{CNF\textsubscript{C}($\mathcal{S}$)}-formulas,
the following lemma shows that \noun{Q-Conn}\textsubscript{\noun{C}}($\mathcal{S}$)
is coNP-hard for all sets of Horn relations that are not IHSB$-$.
\begin{lem}
\label{lem:M2}The relation $M=\left(x\vee\overline{y}\vee\overline{z}\right)\wedge\left(\overline{x}\vee z\right)$
is expressible as an existentially quantified \noun{CNF\textsubscript{C}($\{R\}$)}-formula
for every Horn relation $R$ that is not IHSB$-$.\end{lem}
\begin{proof}
We will use quantification only at the very end. As in the proof of
\prettyref{lem:exp m}, in the following numbered transformation steps,
we use identification and substitution to obtain one of a few simple
formulas from which we can then express $M$ as an existentially quantified
formula.

We again argue with formulas in normal form $\nu$; let $\phi_{0}=\nu(R)$.
Since $R$ is not IHSB$-$, $\phi_{0}$ contains a multi-implication
$c=x\vee\overline{y}\vee\overline{z}_{1}\vee\cdots\vee\overline{z}_{k}$
($k\geq1$).
\begin{enumerate}[resume]
\item \emph{Identify $z_{1},\ldots,x_{k}$, call the resulting variable
$z$.}
\end{enumerate}
This produces the clause $x\vee\overline{y}\vee\overline{z}$ from
$c$. Since by simplicity condition \ref{enu:red imp}, $x$ was not
implied by any set $U\subsetneq\{y,z_{1},\ldots,x_{k}\}$, and by
\ref{enu:red bra}, no $z_{i}$ was implied by $y$, and $y$ was
implied by no set $U\subseteq\{z_{1},\ldots,x_{k}\}$, it follows
that
\begin{itemize}[label= ({*})]
\item  $x\notin\mathrm{Imp}(y)$, $x\notin\mathrm{Imp}(z)$, $z\notin\mathrm{Imp}(y)$,
$y\notin\mathrm{Imp}(z)$.
\end{itemize}
In the following steps, we eliminate all variables other than $x,y,z,$
s.t. ({*}) is maintained. 
\begin{enumerate}[resume]
\item \emph{Substitute 1 for every variable from $\mathrm{Imp}(y)\cap\mathrm{Imp}(z)$.}
\end{enumerate}
By ({*}), the set of these variables cannot have implied $x,y,$ or
$z$, thus there can emerge no unit clause on $x,y$ or $z$; since
by simplicity condition \prettyref{enu:imp imp} the set of the substituted
variables contains no restraint set, $\phi$ cannot become unsatisfiable.
\begin{enumerate}[resume]
\item \emph{Identify all remaining variables from $\mathrm{Imp}(y)\setminus\{y\}$
with $y$.}
\end{enumerate}
Since none of these variables was implied by $z$, still $y\notin\mathrm{Imp}(z)$
and it is easy to see that the other conditions of ({*}) are also
maintained.
\begin{enumerate}[resume]
\item \emph{Identify all remaining variables from $\mathrm{Imp}(z)\setminus\{z\}$
with $z$.}
\end{enumerate}
Analogous to step 3, ({*}) is maintained. Now $\mathrm{Imp}(y)\setminus\{y\}$
and $\mathrm{Imp}(z)\setminus\{z\}$ are empty, so the last step is
easy:
\begin{enumerate}[resume]
\item \emph{Identify all remaining variables other than $x,y,z$ with $x$.} 
\end{enumerate}
The formula now contains only the variables $x,y,z$, and all clauses
satisfy ({*}). If follows that all clauses besides $c$ are from 
\[
\{z\vee\overline{x},\:z\vee\overline{x}\vee\overline{y},\:y\vee\overline{x},\:y\vee\overline{x}\vee\overline{z}\}.
\]
Considering all possible combinations of these clauses, we find that
the the formula is equivalent to $K=x\vee\overline{y}\vee\overline{z}$,
$L=\left(x\vee\overline{y}\vee\overline{z}\right)\wedge\left(\overline{x}\vee\overline{y}\vee z\right)$,
$M=\left(x\vee\overline{y}\vee\overline{z}\right)\wedge\left(\overline{x}\vee z\right)$,
$M'=\left(x\vee\overline{y}\vee\overline{z}\right)\wedge\left(\overline{x}\vee y\right)$,
$S=(x\vee\overline{y}\vee\overline{z})\wedge(y\vee\overline{x})\wedge(z\vee\overline{x})$
or $T=(x\vee\overline{y}\vee\overline{z})\wedge(y\vee\overline{z}\vee\overline{x})\wedge(z\vee\overline{x}\vee\overline{y})$.

We express $M$ from $M'$ by permutation, and $L$ from $S$ or $T$
as 
\[
L=\exists wS(w,y,x)\wedge S(w,y,z)=\exists wT(w,y,x)\wedge T(w,y,z).
\]
 Finally, we express $M$ from $K$ or $L$ as in the proof of \prettyref{lem:exp m}.\end{proof}
\begin{cor}
\label{cor:coq}If a finite set $\mathcal{S}$ of relations is Schaefer,
but not bijunctive, IHSB$-$, IHSB$+$ or affine, \noun{Q-Conn}\textsubscript{\noun{C}}\noun{($\mathcal{S}$)}
is \noun{coNP}-complete.\end{cor}
\begin{proof}
By \prettyref{thm:qdich}, the problem is in coNP. From the definitions
we find that $\mathcal{S}$ must be Horn and contain at least one
relation that is not IHSB$-$, or dual Horn and contain at least one
relation that is not IHSB$+$. In the first case, coNP-hardness follows
from \prettyref{lem:m har} with \prettyref{lem:M2}. The second case
is symmetric. 
\end{proof}

\subsection{Deciding Connectivity in Polynomial Time}

We are now left with\noun{ }bijunctive, IHSB$-$ / IHSB$+$, and affine
sets of relations. We will devise a polynomial-time algorithm for
connectivity in each case.
\begin{lem}
\label{lem:ihsb}If $\mathcal{S}$ is a set of IHSB$-$ or IHSB$+$
relations, there is a polynomial-time algorithm for \noun{Q-Conn\textsubscript{C}($\mathcal{S}$)}.\end{lem}
\begin{proof}
The algorithm is essentially a modified version of Gopalan et al.'s
algorithm from the proof of Lemma 4.13 in \citep{gop}.

Assume $\mathcal{S}$ is IHSB$-$; the IHSB$+$ case is symmetric.
Let $\phi$ be any Q-CNF\textsubscript{C}($S$)-formula.

Since any CNF\textsubscript{C}($S$)-formula is itself IHSB$-$,
also $[\phi]$ can be expressed as an IHSB$-$ formula $\psi$ without
quantifiers due to \prettyref{cor:quanti closed}. So if we could
transform $\phi$ into $\psi$ in polynomial time, we could then simply
use the constraint-projection algorithm from \prettyref{lem:alg}
to decide connectivity. However, quantifier-elimination can lead to
an exponential increase of the formula size, even for Horn formulas
\citep{DBLP:journals/dam/BubeckB08}, and this seems to apply to IHSB$-$
formulas also. Thus we need another strategy.

Fortunately, there at least exists a polynomial time algorithm to
transform a quantified Horn formula into an equivalent Horn formula
with only existential quantifiers, described by Bubeck et al.\ in
Definition 8 of \citep{DBLP:journals/dam/BubeckB08}; we apply this
algorithm to $\phi$ to obtain an equivalent formula $\phi^{\exists}$
with only existential quantifiers.

Next we assign all variables in positive unit clauses (and remove
the corresponding quantifiers in the case of bound variables) to obtain
a connectivity-equivalent formula $\phi{}^{\exists}\mathsf{'}$ without
positive unit clauses; let $\phi{}^{\exists}\text{\textasciiacute}=\exists x_{1}\cdots\exists x_{m}\phi_{0}(x_{1},\ldots,x_{n})$,
where $\phi_{0}$ contains no quantifiers.

We show that
\begin{itemize}[label= ({*})]
\item $\phi{}^{\exists}\mathsf{'}$ (and thus $\phi$) is disconnected
iff

\begin{itemize}[label= (1)]
\item there are two free variables $x$ and $y$ of $\phi{}^{\exists}\mathsf{'}$
s.t. $x\in$$\mathrm{Imp}_{\phi_{0}}$($y$)\footnote{$\mathrm{Imp}_{\psi}$ denotes implication in formula $\psi$},
$y\in$$\mathrm{Imp}_{\phi_{0}}$($x$), and $\mathrm{Imp}_{\phi_{0}}$($x$)
contains no restraint set of $\phi_{0}$.
\end{itemize}
\end{itemize}
It is easy to see that condition (1) can be checked in polynomial
time as follows:
\begin{itemize}[label=]
\item For every free variable $x$ of $\phi'^{\exists}$:

\begin{itemize}[label=]
\item If $\mathrm{Imp}_{\phi_{0}}$($x$) contains no restraint set of
$\phi_{0}$,

\begin{itemize}[label=]
\item For each free variable $y$ from $\mathrm{Imp}_{\phi_{0}}$($x$):

\begin{itemize}[label=]
\item If $x\in$$\mathrm{Imp}_{\phi_{0}}$($y$), return ``disconnected''.
\end{itemize}
\end{itemize}
\end{itemize}
\item Return ``connected''.
\end{itemize}
To prove ({*}), consider the formula $\phi^{\sim}$ we would obtain
from $\phi{}^{\exists}\mathsf{'}$ by eliminating all quantifiers,
and then transforming into conjunctive normal form. It is clear that
also $\phi^{\sim}$ contains no positive unit-clauses. Further, $\phi^{\sim}$
is again IHSB$-$ by \prettyref{cor:quanti closed}, and thus Horn.
So by \prettyref{cor:horn conn}, $\phi^{\sim}$ (and therefore $\phi$)
is disconnected iff
\begin{itemize}[label= (2)]
\item $\phi^{\sim}$ has a non-empty maximal self-implicating set containing
no restraint set.
\end{itemize}
Since $\phi^{\sim}$ is IHSB$-$, it contains no multi-implication
clauses; therefore, in the hypergraph-representation, all implication
clauses of $\phi^{\sim}$ correspond to simple edges, not hyperedges.
Thus, the implication clauses of $\phi^{\sim}$ can be represented
as an ordinary digraph, and implication corresponds to reachability
in that digraph. With this, it is easy to see that (2) is equivalent
to
\begin{itemize}[label= (3)]
\item in $\phi^{\sim}$, there are two variables $x$ and $y$ with $x\in$Imp($y$),
$y\in$Imp($x$), and s.t. Imp($x$) contains no restraint set.
\end{itemize}
It remains to show that (1) is equivalent to (3). The proof will be
by induction. Therefor first note that $\phi^{\sim}$ can be obtained
from $\phi{}^{\exists}\mathsf{'}$ by eliminating the quantifiers
one by one, and always writing the resulting formula in conjunctive
normal form. It follows that if we define $\phi_{k}=\mathrm{CNF}(\phi_{k-1}[x_{k}/0]\vee\phi_{k-1}[x_{k}/1])$,
where CNF($\psi$) denotes a CNF-formula equivalent to $\psi$, then
$\phi^{\sim}=\phi_{m}$. Let

\[
\phi_{k}=c_{1}\wedge\cdots\wedge c_{p}\wedge d_{1}\wedge\cdots\wedge d_{q},
\]
where $c_{1},\ldots,c_{p}$ resp. $d_{1},\ldots,d_{q}$ are the clauses
containing $x_{k}$ resp. not containing $x_{k}$. Then 
\[
\phi_{k-1}[x_{k}/0]\vee\phi_{k-1}[x_{k}/1]\equiv\left(\bigwedge_{i,j}\,c_{i}[y_{k}/0]\vee c_{j}[y_{k}/1]\right)\wedge d_{1}\wedge\cdots\wedge d_{q}.
\]
The clauses $c_{1},\ldots,c_{p}$ are of the form $\overline{x}_{i}\vee x_{k}$,
$x_{i}\vee\overline{x_{k}}$, or $\overline{x}_{k}\vee\overline{x}_{i_{1}}\vee\cdots\vee\overline{x}_{i_{r}}$;
considering all combinations to the disjunctions $c_{i}[y_{k}/0]\vee c_{j}[y_{k}/1]$,
and discarding tautological clauses, it is then easy to see that $\phi_{k}$
may\footnote{a CNF representation is not unique} consist of exactly
the following clauses:
\begin{enumerate}
\item $d_{1},\ldots d_{q}$,
\item for each pair of clauses $\overline{x}_{i}\vee x_{k}$ and $x_{j}\vee\overline{x_{k}}$
of $\phi_{k-1}$ with $x_{i}\neq x_{j}$, we have $\overline{x}_{i}\vee x_{j}$
in $\phi_{k}$,
\item for each pair of clauses $\overline{x}_{i}\vee x_{k}$ and $\overline{x}_{k}\vee\overline{x}_{i_{1}}\vee\cdots\vee\overline{x}_{i_{r}}$
of $\phi_{k-1}$, we have $\overline{x}_{i}\vee\overline{x}_{i_{1}}\vee\cdots\vee\overline{x}_{i_{r}}$
in $\phi_{k}$.
\end{enumerate}
We can now show that for any two variables $x$ and $y$ with $x\neq x_{k}$
and $y\neq x_{k}$,
\begin{enumerate}[label=(\alph{enumi})]
\item \label{enu:-iff-}$x\in\mathrm{Imp}_{\phi_{k}}(y)$ iff $x\in\mathrm{Imp}_{\phi_{k-1}}(y)$,

\begin{itemize}[label=]
\item Proof: ``$\Longleftarrow$'': If $x\in\mathrm{Imp}_{\phi_{k-1}}(y)$,
there was a chain of clauses 
\[
\overline{y}\vee x_{w_{1}},\;\overline{x}_{w_{1}}\vee x_{w_{2}},\;\ldots,\;\overline{x}_{w_{m}}\vee x
\]
 in $\phi_{k-1}$; now if $x_{k}\notin\{x_{w_{1}},\ldots,x_{w_{k}}\}$,
all clauses of the chain are also in $\phi_{k}$ by 1., else a clause
``bridging'' $x_{k}$ is added by 2.
\item ``$\Longrightarrow$'': This is clear since all implication clauses
of $\phi_{k}$ are entailed by $\phi_{k-1}$.
\end{itemize}
\item $\mathrm{Imp}_{\phi_{k}}(x)$ contains a restraint set $U'$ of $\phi_{k}$
iff $\mathrm{Imp}_{\phi_{k-1}}(x)$ contains a restraint set $U$
of $\phi_{k-1}$.

\begin{itemize}[label=]
\item Proof: ``$\Longleftarrow$'': If $U$ did not contain $x_{k}$,
$U$ is also a restraint set of $\phi_{k}$, and by 1., $U\subseteq\mathrm{Imp}_{\phi_{k}}(x)$.
Otherwise, $U=\{x_{k},x_{i_{1}},\ldots,x_{i_{r}}\}$ with each $x_{i_{j}}\in\mathrm{Imp}_{\phi_{k-1}}(x)$,
and there was some clause $\overline{x}_{l}\vee x_{k}$ with $x_{l}\in\mathrm{Imp}_{\phi_{k-1}}(x)$,
but then $\phi_{k}$ contains the clause $\overline{x}_{l}\vee\overline{x}_{i_{1}}\vee\cdots\vee\overline{x}_{i_{r}}$
by 3., and since also all $x_{i_{j}}\in\mathrm{Imp}_{\phi_{k}}(x)$
and $x_{l}\in\mathrm{Imp}_{\phi_{k}}(x)$ by 1., $U'=\{x_{l},x_{i_{1}},\ldots,x_{i_{r}}\}\subseteq\mathrm{Imp}_{\phi_{k}}(x)$.
\item ``$\Longrightarrow$'': If $U'$ corresponds to one of the clauses
$d_{1},\ldots d_{q}$, it was also a restraint set of $\phi_{k-1}$,
and $U'\subseteq\mathrm{Imp}_{\phi_{k-1}}(x)$ by 1. Otherwise, it
corresponds to a clause $\overline{x}_{i}\vee\overline{x}_{i_{1}}\vee\cdots\vee\overline{x}_{i_{r}}$
from 3.; but then $\phi_{k-1}$ must have contained the clauses $\overline{x}_{i}\vee x_{k}$
and $\overline{x}_{k}\vee\overline{x}_{i_{1}}\vee\cdots\vee\overline{x}_{i_{r}}$,
and $U=\{x_{k},x_{i_{1}},\ldots,x_{i_{r}}\}\subseteq\mathrm{Imp}_{\phi_{k-1}}(x)$
by 1.
\end{itemize}
\end{enumerate}
Now (1)$\Longleftrightarrow$(3) follows from (a) and (b) by induction.
\end{proof}
We can reduce the bijunctive case to the IHSB$-$ one:
\begin{lem}
\label{lem:bijunc}If $\mathcal{S}$ is a bijunctive set of relations,
\noun{Q-Conn\textsubscript{C}($\mathcal{S}$)} is in \noun{P}.\end{lem}
\begin{proof}
Makino, Tamaki and Yamamoto show in \citep{Makino:2007:BCP:1768142.1768162}
below Proposition 2 that any bijunctive formula can be transformed
in a connectivity-equivalent Horn 2-CNF formula by ``renaming''
variables: We can calculate a solution $\boldsymbol{a}$ in linear
time \citep{DBLP:journals/ipl/AspvallPT79} (w.l.o.g. we may assume
that a solution exists) and then take 
\[
\psi(\boldsymbol{x})=\phi(x_{1}\oplus a_{1},x_{2}\oplus a_{2},\ldots).
\]
Now $\psi$ is clearly Horn since $\psi(0,\ldots0)=1$, and the connectivity
is retained since 
\[
|(x_{1}\oplus a_{1},x_{2}\oplus a_{2},\ldots)-(y_{1}\oplus a_{1},y_{2}\oplus a_{2},\ldots)|=|(x_{1},x_{2},\ldots)-(y_{1},y_{2},\ldots)|.
\]

Since any CNF\textsubscript{C}($\mathcal{S}$) formula $\phi$ is
itself bijunctive, given $Q_{1}y_{1}\cdots Q_{m}y_{m}\phi(\boldsymbol{x},\boldsymbol{y})$,
we can instead take $Q_{1}y_{1}\cdots Q_{m}y_{m}\psi(\boldsymbol{x},\boldsymbol{y})$,
where $\psi$ is the Horn 2-CNF formula obtained from $\phi$ as described,
and then apply the algorithm for the IHSB$-$ case, since any Horn
2-CNF formula is also IHSB$-$.
\end{proof}
Finally, for affine sets of relations, we again make use of the quantifier-elimination
algorithm from \citep{DBLP:phd/de/Bauland2007}, and then use an algorithm
from Gopalan et al.:
\begin{lem}
\label{lem:affin}If $\mathcal{S}$ is an affine set of relations,
\noun{Q-Conn\textsubscript{C}($\mathcal{S}$)} is in \noun{P}.\end{lem}
\begin{proof}
We can use the polynomial-time algorithm from Figure 4.2 in \citep{DBLP:phd/de/Bauland2007}
to transform any quantified CNF\textsubscript{C}($\mathcal{S}$)
formula into an equivalent one without any quantifiers, and then apply
the algorithm from the proof of Lemma 4.10 of \citep{gop}. Note that
although the clause-size of the formulas produced by the algorithm
from \citep{DBLP:phd/de/Bauland2007} is unbounded (with regard to
$\mathcal{S}$), the algorithm from \citep{gop} decides connectivity
for these formulas in polynomial time, as is easy to see.
\end{proof}

\chapter{\label{chap:4}Connectivity of~Nested~Formulas~and Circuits}

We now turn to a quite different type of representation for Boolean
relations. First observe that we can naturally identify Boolean relations
with \emph{Boolean function}s, i.e.\ functions $f:\{0,1\}^{n}\rightarrow\{0,1\}$.
To build \emph{$B$-}formulas, we again use a fixed finite set of
Boolean relations, resp. functions, as source material, but instead
of connecting them with $\wedge$'s, we now allow inserting them into
each other arbitrarily.

In a \emph{$B$-}formula, if some function is used repeatedly with
the same arguments, it has to be duplicated, what can make the formula
large and evaluation slow. This is remedied by \emph{$B$-}circuits,
that allow to use the result of a function any number of times, which
can lead to exponential savings of space and time.

\section{\label{sec:Cir}Preliminaries:\emph{ B-}Circuits,~\emph{B-}Formulas,~and~Post's~Lattice}

We begin with formal definitions of \emph{$B$-}circuits and \emph{$B$-}formulas.
Let $B$ be a finite set of Boolean functions.
\begin{defn}
\label{def:b}A \emph{$B$-circuit} $\mathcal{C}$ with input variables
$x_{1},\ldots,x_{n}$ is a directed acyclic graph, augmented as follows:
Each node (here also called \emph{gate}) with indegree 0 is labeled
with an $x_{i}$ or a 0-ary function from $B$, each node with indegree
$k>0$ is labeled with a $k$-ary function from $B$. The edges (here
also called \emph{wires}) pointing into a gate are ordered. One node
is designated the output gate. Given values $a_{1},\ldots,a_{n}\in\{0,1\}$
to $x_{1},\ldots,x_{n}$, $\mathcal{C}$ computes an $n$-ary function
$f_{\mathcal{C}}$ as follows: A gate $v$ labeled with a variable
$x_{i}$ returns $a_{i}$, a gate $v$ labeled with a function $f$
computes the value $f(b_{1},\ldots,b_{k})$, where $b_{1},\ldots,b_{k}$
are the values computed by the predecessor gates of $v$, ordered
according to the order of the wires.
\end{defn}
For a detailed introduction to Boolean circuits and circuit complexity,
see e.g. \citep{Vollmer:1999:ICC:520668}.
\begin{defn}
\label{def:bc}A \emph{$B$-formula} is defined inductively: A variable
$x$ is a $B$-formula. If $\phi_{1},\ldots,\phi_{m}$ are $B$-formulas,
and $f$ is an $n$-ary function from $B$, then $f(\phi_{1},\ldots,\phi_{n})$
is a $B$-formula. In turn, any $B$-formula defines a Boolean function
in the obvious way, and we will identify $B$-formulas with the function
they define.
\end{defn}
It is easy to see that the functions computable by a $B$-circuit,
as well as the functions definable by a $B$-formula, are exactly
those that can be obtained from $B$ by \emph{superposition}, together
with all projections \citep{bloc}. By superposition, we mean substitution
(that is, composition of functions), permutation and identification
of variables, and introduction of \emph{fictive variables} (variables
on which the value of the function does not depend). This class of
functions is denoted $[B]$. $B$ is \emph{closed} (or said to be
a \emph{clone}) if $[B]=B$. A \emph{base} of a clone $F$ is any
set $B$ with $[B]=F$.

Already in the early 1920s, Emil Post extensively studied Boolean
functions \citep{post1941two}. He identified all closed classes,
found a finite base for each of them, and detected their inclusion
structure: The closed classes form a lattice, called \emph{Post's
lattice}, depicted in Figure \ref{fig:Post's-lattice}. We do not
use Post's original names for the closed classes, but the modern terminology
developed by Reith and Wagner in \citep{Reith99}; the layout of the
lattice is also from \citep{Reith99}.

The following clones are defined by properties of the functions they
contain, all other ones are intersections of these. Let $f$ be an
$n$-ary Boolean function.
\begin{itemize}
\item $\mathsf{BF}$ is the class of all Boolean functions.
\item $\mathsf{R}_{0}$ ($\mathsf{R}_{1}$) is the class of all 0-reproducing
(1-reproducing) functions,\\
$f$\emph{ }is\emph{ $c$-reproducing}, if $f(c,\ldots,c)=c$, where
$c\in\{0,1\}$.
\item $\mathsf{M}$ is is the class of all monotone functions,\emph{}\\
$f$\emph{ }is\emph{ monotone}, if $a_{1}\leq b_{1},\ldots,a_{n}\leq b_{n}$
implies $f(a_{1},\ldots,a_{n})\leq f(b_{1},\ldots,b_{n})$.
\item $\mathsf{D}$ is the class of all self-dual functions,\emph{}\\
$f$\emph{ }is\emph{ self-dual}, if $f(x_{1},\ldots,x_{n})=\overline{f(\overline{x_{1}},\ldots,\overline{x_{n}})}$.
\item $\mathsf{L}$ is the class of all affine (on \emph{linear}) functions,\emph{}\\
$f$\emph{ }is\emph{ affine}, if $f(x_{1},\ldots,x_{n})=x_{i_{1}}\oplus\cdots\oplus x_{i_{m}}\oplus c$
with $i_{1},\ldots,i_{m}\in\{1,\ldots,n\}$ and $c\in\{0,1\}$.
\item $\mathsf{S}_{0}$ ($\mathsf{S}_{1}$) is the class of all 0-separating
(1-separating) functions,\emph{}\\
$f$\emph{ }is\emph{ $c$-separating}, if there exists an $i\in\{1,\ldots,n\}$
s.t. $a_{i}=c$ for all $\boldsymbol{a}\in f^{-1}(c)$, where $c\in\{0,1\}$.
\item $\mathsf{S}_{0}^{m}$ ($\mathsf{S}_{1}^{m}$) is the class of all
functions 0-separating (1-separating) of degree $m$,\\
$f$\emph{ }is\emph{ $c$-separating of degree $m$}, if for all $U\subseteq f^{-1}(c)$
of size $|U|=m$ there exists an $i\in\{1,\ldots,n\}$ s.t. $a_{i}=c$
for all $\boldsymbol{a}\in U$ ($c\in\{0,1\}$, $m\geq2$).
\end{itemize}
The definitions and bases of all classes are given in Table \ref{tab:List-of-all}.
For an introduction to Post's lattice and further references see e.g.
\citep{bloc}.

The complexity of numerous problems for $B$-circuits and $B$-formulas
has been classified by the types of functions allowed in $B$ with
help of Post's lattice (see e.g. \citep{reith2000,schnoor2007}),
starting with satisfiability: Analogously to Schaefer's dichotomy
for CNF($\mathcal{S}$)-formulas\emph{ }from 1978\emph{,} Harry R.
Lewis shortly thereafter found a dichotomy for $B$-formulas \citep{lewis1979}:
If $[B]$ contains the function $x\wedge\overline{y}$, \noun{Sat}
is NP-complete, else it is in P.

While for $B$-circuits the complexity of every decision problem solely
depends on $[B]$ (up to AC$^{0}$ isomorphisms), for $B$-formulas
this need not be the case (though it usually is, as for satisfiability
and our connectivity problems, as we will see): The transformation
of a $B$-formula into a $B'$-formula might require an exponential
increase in the formula size even if $[B]=[B']$, as the $B'$-representation
of some function from $B$ may need to use some input variable more
than once \citep{michael2012applicability}. For example, let $h(x,y)=x\wedge\overline{y}$;
then $(\text{x\ensuremath{\wedge}y)}\in[\{h\}]$ since $x\wedge y=h(x,h(x,y))$,
but it is easy to see that there is no shorter $\{h\}$-representation
of $x\wedge y$.

We denote the $st$-connectivity and connectivity problems for $B$-formulas
by\noun{ st-BF-Conn($B$)} and \noun{BF-Conn($B$)}, respectively,
and the problems for circuits by \noun{st-Circ-Conn($B$) }resp. \noun{Circ-Conn($B$)}.

Also for a Boolean function $f$, we will speak of the solution graph
$G(f)$ and denote the shortest-path distance in $G(f)$ by $d_{f}$.

\begin{figure}[p]
\begin{centering}
\begin{tikzpicture}[clone/.style={circle,inner sep=0,minimum width=6mm,fill=white,font=\small},x=0.95cm,y=0.75cm,
edge/.style={draw=black,thin,-},
clonefill/.style={draw=black,circle,minimum width=6.1mm,inner sep=0,font=\scriptsize},
clonefil/.style={draw=black,very thick,circle,minimum width=6.1mm,inner sep=0,font=\scriptsize},scale=1, transform shape]

\input{lattice}

	\node[clonefil,fill=WSAT] at	 (BF) 		{$\CloneBF$};
	\node[clonefill,fill=WSAT] at  (R1)  	{$\CloneR_1$};
	\node[clonefil,fill=WSAT] at  (R0) 		{$\CloneR_0$};
	\node[clonefill,fill=WSAT] at  (R2)  	{$\CloneR_2$};
	
	\node[clonefill,fill=white] at  (M)  		{$\CloneM$};
	\node[clonefill,fill=white] at  (M1)  	{$\CloneM_1$};
	\node[clonefill,fill=white] at  (M0) 		{$\CloneM_0$};
	\node[clonefill,fill=white] at  (M2)  	{$\CloneM_2$};

	\node[clonefil,fill=WSAT] at  (S21) 	{$\CloneS_{1}^2$};
	\node[clonefil,fill=WSAT] at  (S31)		{$\CloneS_{1}^3$};
	\node[clonefil,fill=WSAT] at  (Sn1)		{$\CloneS_{1}^n$};
	\node[clonefil,fill=WSAT] at  (S1)		{$\CloneS_{1}$};

	\node[clonefill,fill=WSAT] at  (S212) 	{$\CloneS_{12}^2$};
	\node[clonefill,fill=WSAT] at  (S312) 	{$\CloneS_{12}^3$};
	\node[clonefill,fill=WSAT] at  (Sn12) 	{$\CloneS_{12}^n$};
	\node[clonefill,fill=WSAT] at  (S12)	  	{$\CloneS_{12}$};

	\node[clonefill,fill=white] at  (S211) 	{$\CloneS_{11}^2$};
	\node[clonefill,fill=white] at  (S311) 	{$\CloneS_{11}^3$};
	\node[clonefill,fill=white] at  (Sn11) 	{$\CloneS_{11}^n$};
	\node[clonefill,fill=white] at  (S11)	  	{$\CloneS_{11}$};

	\node[clonefill,fill=white] at  (S210) 	{$\CloneS_{10}^2$};
	\node[clonefill,fill=white] at  (S310) 	{$\CloneS_{10}^3$};
	\node[clonefill,fill=white] at  (Sn10) 	{$\CloneS_{10}^n$};
	\node[clonefill,fill=white] at  (S10)	  	{$\CloneS_{10}$};

	\node[clonefill,fill=WSAT] at  (S20) 	{$\CloneS_{0}^2$};
	\node[clonefill,fill=WSAT] at  (S30) 	{$\CloneS_{0}^3$};
	\node[clonefill,fill=WSAT] at  (Sn0) 	{$\CloneS_{0}^n$};
	\node[clonefill,fill=nWSAT] at  (S0)	  	{$\CloneS_{0}$};

	\node[clonefill,fill=WSAT] at  (S202) 	{$\CloneS_{02}^2$};
	\node[clonefill,fill=WSAT] at  (S302) 	{$\CloneS_{02}^3$};
	\node[clonefill,fill=WSAT] at  (Sn02) 	{$\CloneS_{02}^n$};
	\node[clonefill,fill=nWSAT] at  (S02)	  	{$\CloneS_{02}$};

	\node[clonefill,fill=white] at  (S201) 	{$\CloneS_{01}^2$};
	\node[clonefill,fill=white] at  (S301) 	{$\CloneS_{01}^3$};
	\node[clonefill,fill=white] at  (Sn01) 	{$\CloneS_{01}^n$};
	\node[clonefill,fill=white] at  (S01)	  	{$\CloneS_{01}$};
	
	\node[clonefill,fill=white] at  (S200) 	{$\CloneS_{00}^2$};
	\node[clonefill,fill=white] at  (S300) 	{$\CloneS_{00}^3$};
	\node[clonefill,fill=white] at  (Sn00) 	{$\CloneS_{00}^n$};
	\node[clonefill,fill=white] at  (S00)	  	{$\CloneS_{00}$};
	
	\node[clonefill,fill=WSAT] at  (D)  		{$\CloneD$};
	\node[clonefill,fill=WSAT] at  (D1)  	{$\CloneD_1$};
	\node[clonefill,fill=white] at  (D2)  	{$\CloneD_2$};
	
	\node[clonefill,fill=white] at  (E)  		{$\CloneE$};
	\node[clonefill,fill=white] at  (E1)  	{$\CloneE_1$};
	\node[clonefill,fill=white] at  (E0) 		{$\CloneE_0$};
	\node[clonefill,fill=white] at  (E2)  	{$\CloneE_2$};
	
	\node[clonefill,fill=white] at  (V)  		{$\CloneV$};
	\node[clonefill,fill=white] at  (V0) 		{$\CloneV_0$};
	\node[clonefill,fill=white] at  (V1)  	{$\CloneV_1$};
	\node[clonefill,fill=white] at  (V2)  	{$\CloneV_2$};
	
	\node[clonefill,fill=white] at  (L)  		{$\CloneL$};
	\node[clonefill,fill=white] at  (L0)  	{$\CloneL_0$};
	\node[clonefill,fill=white] at  (L1)  	{$\CloneL_1$};
	\node[clonefill,fill=white] at  (L3) 		{$\CloneL_3$};
	\node[clonefill,fill=white] at  (L2)  	{$\CloneL_2$};
	
	\node[clonefill,fill=white] at  (N)  		{$\CloneN$};
	\node[clonefill,fill=white] at  (N2)  	{$\CloneN_2$};
	
	\node[clonefill,fill=white] at  (I)  		{$\CloneI$};
	\node[clonefill,fill=white] at  (I0)  	{$\CloneI_0$};
	\node[clonefill,fill=white] at  (I1)  	{$\CloneI_1$};
	\node[clonefill,fill=white] at  (I2)  	{$\CloneI_2$};
	
\end{tikzpicture}
    \bigskip
\par\end{centering}

\protect\caption[Post's lattice with our results]{\emph{}\label{fig:Post's-lattice}Post's lattice with our results.\protect \\
The classes on the hard side of the dichotomy for the connectivity
problems and the diameter are shaded; the light shaded ones are on
the hard side only for formulas with quantifiers.\protect \\
{\small{}\hspace*{3ex}}For comparison, the classes for which SAT
(without quantifiers) is NP-complete are circled bold.}
\end{figure}

\begin{table}[p]

  \fontsize{11}{11}\selectfont
    \setlength{\tabcolsep}{6pt}
    \rowcolors{2}{gray!10}{}
    \centering
    \begin{tabular}{llll}
    \toprule
    Class  & Definition & Base\\
    \midrule
    $\CloneBF$    & All Boolean functions & $\{x \land y ,\neg x\}$ &\\
    $\CloneR_0$   & $\{f \in \CloneBF \mid f \mbox{ is 0-reproducing}\}$ & $\{ x \land y , x \oplus y \}$ &\\
    $\CloneR_1$   & $\{f \in \CloneBF \mid f \mbox{ is 1-reproducing}\}$ & $\{x \lor y, x \leftrightarrow y \}$ &\\
    $\CloneR_2$   & $\CloneR_0 \cap \CloneR_1$ & $\{x \lor y, x \land (y \leftrightarrow z) \}$ &\\
    $\CloneM$     & $\{f \in \CloneBF \mid f \mbox{ is monotone}\}$ & $\{x \land y, x\lor y,0,1\}$ &\\
    $\CloneM_0$   & $\CloneM \cap \CloneR_0$ & $\{x \land y,x\lor y,0\}$ &\\
    $\CloneM_1$   & $\CloneM \cap \CloneR_1$ & $\{x \land y,x\lor y,1\}$ &\\  
    $\CloneM_2$   & $\CloneM \cap \CloneR_2$ & $\{x \land y,x\lor y\}$ &\\
    $\CloneS_0$   & $\{f \in \CloneBF \mid f \mbox{ is 0-separating}\}$ & $\{x \rightarrow y\}$ &\\
    $\CloneS_0^n$ & $\{f \in \CloneBF \mid f \mbox{ is 0-separating of degree $n$}\}$ & $\{x \rightarrow y,\mbox{dual}({\mathrm{T}^{n+1}_n})\}$ &\\
    $\CloneS_1$   & $\{f \in \CloneBF \mid f \mbox{ is 1-separating}\}$ & $\{x \nrightarrow y\}$ &\\
    $\CloneS_1^n$ & $\{f \in \CloneBF \mid f \mbox{ is 1-separating of degree $n$}\}$ & $\{x \nrightarrow y,\mathrm{T}^{n+1}_n\}$ &\\
    $\CloneS_{02}^n$& $\CloneS_0^n \cap \CloneR_2$ & $\{x \lor (y \land \neg z), \mbox{dual}({\mathrm{T}^{n+1}_n})\}$ &\\
    $\CloneS_{02}$  & $\CloneS_0 \cap \CloneR_2$ & $\{x \lor (y \land \neg z)\}$ &\\
    $\CloneS_{01}^n$& $\CloneS_0^n \cap \CloneM$ & $\{\mbox{dual}({\mathrm{T}^{n+1}_n}),1\}$ &\\
    $\CloneS_{01}$  & $\CloneS_0 \cap \CloneM$ & $\{x \lor (y \land z),1\}$ &\\
    $\CloneS_{00}^n$& $\CloneS_0^n \cap \CloneR_2 \cap \CloneM$ 
    	&$\{x \lor (y \land z),\mbox{dual}({\mathrm{T}^{n+1}_n})\}$& \\
    $\CloneS_{00}$  & $\CloneS_0 \cap \CloneR_2 \cap \CloneM$ & $\{x \lor (y \land z)\}$ &\\
    $\CloneS_{12}^n$& $\CloneS_1^n \cap \CloneR_2$ & $\{x\land (y\lor \neg z),\mathrm{T}^{n+1}_n\}$ &\\
    $\CloneS_{12}$  & $\CloneS_1 \cap \CloneR_2$ & $\{x\land (y\lor \neg z)\}$ &\\
    $\CloneS_{11}^n$& $\CloneS_1^n \cap \CloneM$ & $\{\mathrm{T}^{n+1}_n,0\}$ &\\
    $\CloneS_{11}$  & $\CloneS_1 \cap \CloneM$ & $\{x\land (y\lor z),0\}$ &\\
    $\CloneS_{10}^n$& $\CloneS_1^n \cap \CloneR_2 \cap \CloneM$ 
    & $\{x \land (y \lor z),{\mathrm{T}^{n+1}_n}\}$ &\\
    $\CloneS_{10}$  & $\CloneS_1 \cap \CloneR_2 \cap \CloneM$ & $\{x\land (y\lor z)\}$ &\\
    
    $\CloneD$     & $\{f \in \CloneBF \mid f \mbox{ is self-dual}\}$ & $\{\mbox{maj}(x,\neg y,\neg z)\}$ &\\
    $\CloneD_1$   & $\CloneD \cap \CloneR_2$ & $\{\mbox{maj}(x, y,\neg z)\}$ &\\
    $\CloneD_2$   & $\CloneD \cap \CloneM$ & $\{\mbox{maj}(x,y,z)\}$ &\\  
    
    $\CloneL$     & $\{f \in \CloneBF \mid f \mbox{ is linear}\}$ & $\{x \oplus y,1\}$ &\\
    $\CloneL_0$   & $\CloneL \cap \CloneR_0$ & $\{x \oplus y\}$ &\\
    $\CloneL_1$   & $\CloneL \cap \CloneR_1$ & $\{x \leftrightarrow y\}$ &\\
    $\CloneL_2$   & $\CloneL \cap \CloneR_2$ & $\{x \oplus y \oplus z\}$ &\\
    $\CloneL_3$   & $\CloneL \cap \CloneD$ & $\{x \oplus y \oplus z \oplus 1 \}$ &\\
    
    $\CloneE$     & $\{f \in \CloneBF \mid f \mbox{ is constant or a conjunction}\}$ & $\{x \land y,0, 1\}$ &\\
    $\CloneE_0$   & $\CloneE \cap \CloneR_0$ & $\{x \land y,0\}$ &\\
    $\CloneE_1$   & $\CloneE \cap \CloneR_1$ & $\{x \land y,1\}$ &\\
    $\CloneE_2$   & $\CloneE \cap \CloneR_2$ & $\{x \land y\}$ &\\
    
    $\CloneV$     & $\{f \in \CloneBF \mid f \mbox{ is constant or a disjunction}\}$ & $\{x \lor y,0,1\}$ &\\
    $\CloneV_0$   & $\CloneV \cap \CloneR_0$ & $\{x \lor y,0\}$ &\\
    $\CloneV_1$   & $\CloneV \cap \CloneR_1$ & $\{x \lor y,1\}$ &\\
    $\CloneV_2$   & $\CloneV \cap \CloneR_2$ & $\{x \lor y\}$ &\\
    
    $\CloneN$     & $\{f \in \CloneBF \mid f \mbox{ is essentially unary}\}$ & $\{\neg x,0,1\}$ &\\ 
    $\CloneN_2$   & $\CloneN \cap \CloneD$ & $\{\neg x\}$ &\\

    $\CloneI$     & $\{f \in \CloneBF \mid f \mbox{ is constant or a projection}\}$ & $\{x,0,1\}$ &\\ 
    $\CloneI_0$   & $\CloneI \cap \CloneR_0$ & $\{x,0\}$ &\\
    $\CloneI_1$   & $\CloneI \cap \CloneR_1$ & $\{x,1\}$ &\\
    $\CloneI_2$   & $\CloneI \cap \CloneR_2$ & $\{x\}$ &\\
    \bottomrule
    \end{tabular}
    \bigskip

\protect\caption[List of all closed classes of Boolean functions with definitions and
bases]{\emph{}\label{tab:List-of-all}List of all closed classes of Boolean
functions with definitions and bases.\protect \\
$T_{k}^{n}$ denotes the threshold function, i.e., $T_{k}^{n}(x_{1},\ldots,x_{n})=1\protect\Longleftrightarrow\sum_{i=1}^{n}x_{i}\geq k$,\protect \\
(dual($f$))$(x_{1},\ldots,x_{n})=\overline{f(\overline{x_{1}},\ldots,\overline{x_{n}})}$) }
\end{table}
\newpage{}

\section{Results}

The following two theorems give complete classifications up to polynomial-time
isomorphisms. See also Figure \ref{fig:Post's-lattice}.
\begin{thm}
\label{thm:func}Let $B$ be a finite set of Boolean functions.
\begin{enumerate}
\item If $B\subseteq\mathsf{M}$, $B\subseteq\mathsf{L}$, or $B\subseteq\mathsf{\mathsf{S}_{0}}$,
then

\begin{enumerate}
\item \noun{st-Circ-Conn(}$B$\emph{)} and \noun{Circ-Conn(}$B$\emph{)}
are in \noun{P},
\item \noun{st-BF-Conn(}$B$\emph{)} and \noun{BF-Conn(}$B$\emph{)} are
in \noun{P},
\item for every $B$-formula $\phi,$ the diameter of $G(\phi)$ is linear
in the number of variables.
\end{enumerate}
\item Otherwise,

\begin{enumerate}
\item \noun{st-Circ-Conn(}$B$\emph{)} and \noun{Circ-Conn(}$B$\emph{)}
are $\mathrm{PSPACE}$-complete,
\item \noun{st-BF-Conn(}$B$\emph{)} and \noun{BF-Conn(}$B$\emph{)} are
$\mathrm{PSPACE}$-complete,
\item there are $B$-formulas $\phi$ such that the diameter of $G(\phi)$
is exponential in the number of variables.
\end{enumerate}
\end{enumerate}
\end{thm}
The proof follows from the lemmas in the next two sections.

\section{\label{sub:The-easy-cases}The Easy Side of the Dichotomy}

We will often use the following proposition to relate the complexity
of $B$-formulas and $B$-circuits:
\begin{prop}
\label{pro:trf}Every $B$-formula $\phi$ can be transformed into
an equivalent $B$-circuit $\mathcal{C}$ in polynomial time.\end{prop}
\begin{proof}
Any $B$-formula is equivalent to a special $B$-circuit where all
function-gates have outdegree at most one: For every variable $x$
of $\phi$ and for every occurrence of a function $f$ in $\phi$
there is a gate in $\mathcal{C}$, labeled with $x$ resp. $f$. It
is clear how to connect the gates.\end{proof}
\begin{lem}
\label{lem:M}If $B\subseteq\mathsf{M}$, the solution graph of any
$n$-ary function $f\in[B]$ is connected, and $d_{f}(\boldsymbol{a},\boldsymbol{b})=|\boldsymbol{a}-\boldsymbol{b}|\leq n$
for any two solutions $\boldsymbol{a}$ and $\boldsymbol{b}$.\end{lem}
\begin{proof}
Table \ref{tab:List-of-all} shows that $f$ is monotone in this case.
Thus, either $f=0$, or $(1,\ldots,1)$ must be a solution, and every
other solution $\boldsymbol{a}$ is connected to $(1,\ldots,1)$ in
$G(f)$ since $(1,\ldots,1)$ can be reached by flipping the variables
assigned 0 in $\boldsymbol{a}$ one at a time to 1. Further, if $\boldsymbol{a}$
and $\boldsymbol{b}$ are solutions, $\boldsymbol{b}$ can be reached
from $\boldsymbol{a}$ in $|\boldsymbol{a}-\boldsymbol{b}|$ steps
by first flipping all variables that are assigned 0 in $\boldsymbol{a}$
and 1 in $\boldsymbol{b}$, and then flipping all variables that are
assigned 1 in $\boldsymbol{a}$ and 0 in $\boldsymbol{b}$.\end{proof}
\begin{lem}
If $B\subseteq\mathsf{S}_{0}$, the solution graph of any function
$f\in[B]$ is connected, and $d_{f}(\boldsymbol{a},\boldsymbol{b})\leq|\boldsymbol{a}-\boldsymbol{b}|+2$
for any two solutions $\boldsymbol{a}$ and $\boldsymbol{b}$.\end{lem}
\begin{proof}
Since $f$ is 0-separating, there is an $i$ such that $a_{i}=0$
for every vector $\boldsymbol{a}$ with $f(\boldsymbol{a})=0$, thus
every $\boldsymbol{b}$ with $b_{i}=1$ is a solution. It follows
that every solution $\boldsymbol{t}$ can be reached from any solution
$\boldsymbol{s}$ in at most $|\boldsymbol{s}-\boldsymbol{t}|+2$
steps by first flipping the $i$-th variable from 0 to 1 if necessary,
then flipping all other variables in which $\boldsymbol{s}$ and $\boldsymbol{t}$
differ, and finally flipping back the $i$-th variable if necessary.\end{proof}
\begin{lem}
\label{lem:L}If $B\subseteq\mathsf{L}$,
\begin{enumerate}
\item \noun{st-Circ-Conn(}$B$\emph{)} and \noun{Circ-Conn(}$B$\emph{)}
are in P,
\item \noun{st-BF-Conn(}$B$\emph{)} and \noun{BF-Conn(}$B$\emph{)} are
in P,
\item for any function $f\in[B]$, $d_{f}(\boldsymbol{a},\boldsymbol{b})=|\boldsymbol{a}-\boldsymbol{b}|$
for any two solutions $\boldsymbol{a}$ and $\boldsymbol{b}$ that
lie in the same connected component of $G(\phi)$.
\end{enumerate}
\end{lem}
\begin{proof}
Since every function $f\in\mathsf{L}$ is linear, $f(x_{1},\ldots,x_{n})=x_{i_{1}}\oplus\ldots\oplus x_{i_{m}}\oplus c$,
and any two solutions $\boldsymbol{s}$ and $\boldsymbol{t}$ are
connected iff they differ only in fictive variables: If $\boldsymbol{s}$
and $\boldsymbol{t}$ differ in at least one non-fictive variable
(i.e., an $x_{i}\in\{x_{i_{1}},\ldots,x_{i_{m}}\}$), to reach $\mathbf{t}$
from $\mathbf{s}$, $x_{i}$ must be flipped eventually, but for every
solution $\boldsymbol{a}$, any vector $\boldsymbol{b}$ that differs
from $\boldsymbol{a}$ in exactly one non-fictive variable is no solution.
If $\boldsymbol{s}$ and $\boldsymbol{t}$ differ only in fictive
variables, $\boldsymbol{t}$ can be reached from $\boldsymbol{s}$
in $|\boldsymbol{s}-\boldsymbol{t}|$ steps by flipping one by one
the variables in which they differ.

Since $\{x\oplus y,1\}$ is a base of $\mathsf{L}$, every $B$-circuit
$\mathcal{C}$ can be transformed in polynomial time into an equivalent
$\{x\oplus y,1\}$-circuit $\mathcal{C}'$ by replacing each gate
of $\mathcal{C}$ with an equivalent $\{x\oplus y,1\}$-circuit. Now
one can decide in polynomial time whether a variable $x_{i}$ is fictive
by checking for $\mathcal{C}'$ whether the number of ``backward
paths'' from the output gate to gates labeled with $x_{i}$ is odd,
so \noun{st-Circ-Conn(}$B$) is in P.

$G(\mathcal{C})$ is connected iff at most one variable is non-fictive,
thus \noun{Circ-Conn(}$B$) is in P.

By Proposition \ref{pro:trf}, \noun{st-BF-Conn(}$B$) and \noun{BF-Conn(}$B$)
are in P also.
\end{proof}
This completes the proof of the easy side of the dichotomy.

\section{\label{sub:The-hard-cases}The Hard Side of the Dichotomy}

Clearly, we can transfer the upper bounds for general CNF-formulas
to $B$-formulas and $B$-circuits:
\begin{prop}
\noun{st-Circ-Conn(}$B$\emph{)} and \noun{Circ-Conn(}$B$\emph{)},
as well as \noun{st-BF-Conn(}$B$\emph{)} and \noun{BF-Conn(}$B$\emph{)},
are in $\mathrm{PSPACE}$ for any finite set \noun{$B$} of Boolean
functions\noun{.}\end{prop}
\begin{proof}
This follows as in Lemma 3.6 of \citep{gop} (see \prettyref{lem:l36}).
\end{proof}
All hardness proofs will be by reductions from the problems for 1-reproducing
3-CNF-formulas, which are $\mathrm{PSPACE}$-complete by the following
proposition.
\begin{prop}
For 1-reproducing 3-CNF-formulas, the problems \noun{st-Conn }and
\noun{Conn} are $\mathrm{PSPACE}$-hard.\end{prop}
\begin{proof}
In the $\mathrm{PSPACE}$-hardness proof for general 3-CNF-formulas
(Lemma 3.6 of \citep{gop}, see \prettyref{lem:l36}), two satisfying
assignments $\boldsymbol{s}$ and $\boldsymbol{t}$ to the constructed
formula $\phi$ are known, so we can construct a connectivity-equivalent
1-reproducing 3-CNF-formula $\psi$, e.g. as $\psi(\boldsymbol{x})=\phi(x_{1}\oplus s_{1}\oplus1,\ldots,x_{n}\oplus s_{n}\oplus1)$,
and then check connectivity for $\psi$ instead of $\phi$.
\end{proof}
An inspection of Post's lattice shows that if $B\nsubseteq\mathsf{M}$,
$B\nsubseteq\mathsf{L}$, and $B\nsubseteq\mathsf{S}_{0}$, then $[B]\supseteq\mathsf{S}_{12}$,
$[B]\supseteq\mathsf{D}_{1}$, or $[B]\supseteq\mathsf{S}_{02}^{k}\,\forall k\geq2$,
so we have to prove $\mathrm{PSPACE}$-completeness and show the existence
of $B$-formulas with an exponential diameter in these cases.
\begin{defn}
We write $\boldsymbol{x}=\boldsymbol{c}$ or $\boldsymbol{x}=c_{1}\cdots c_{n}$
for $(x_{1}=c_{1})\wedge\cdots\wedge(x_{n}=c_{n})$, where $\boldsymbol{c}=(c_{1},\ldots,c_{n})$
is a vector of constants; e.g.. $\boldsymbol{x}=\boldsymbol{0}$ means
$\overline{x}_{1}\wedge\cdots\wedge\overline{x}_{n}$, and $\boldsymbol{x}=101$
means $x_{1}\wedge\overline{x}_{2}\wedge x_{3}$. Further, we use
$\boldsymbol{x}\in\left\{ \boldsymbol{a},\boldsymbol{b},\ldots\right\} $
for $(\boldsymbol{x}=\boldsymbol{a})\vee(\boldsymbol{x}=\boldsymbol{b})\vee\ldots$.
Also, we write $\psi(\boldsymbol{\overline{x}})$ for $\psi(\overline{x}_{1},\ldots,\overline{x}_{n})$.
If we have two vectors of Boolean values $\boldsymbol{a}$ and $\boldsymbol{b}$
of length $n$ and $m$ resp., we write $\boldsymbol{a}\cdot\boldsymbol{b}$
for their concatenation $(a_{1},\ldots,a_{n},b_{1},\ldots b_{m})$.\end{defn}
\begin{lem}
\label{lem:s12}If $[B]\supseteq\mathsf{S}_{12}$,
\begin{enumerate}
\item \noun{st-BF-Conn(}$B$\emph{)} and \noun{BF-Conn(}$B$\emph{)} are
$\mathrm{PSPACE}$-complete,
\item \noun{st-Circ-Conn(}$B$\emph{)} and \noun{Circ-Conn(}$B$\emph{)}
are $\mathrm{PSPACE}$-complete,
\item for $n\geq3$, there is an $n$-ary function $f\in[B]$ with diameter
of at least $2^{\left\lfloor \frac{n-1}{2}\right\rfloor }$.
\end{enumerate}
\end{lem}
\begin{proof}
1. We reduce the problems for 1-reproducing 3-CNF-formulas to the
ones for $B$-formulas: We map a 1-reproducing 3-CNF-formula $\phi$
and two solutions $\boldsymbol{s}$ and $\boldsymbol{t}$ of $\phi$
to a $B$-formula $\phi'$ and two solutions $\boldsymbol{s'}$ and
$\boldsymbol{t'}$ of $\phi'$ such that $\boldsymbol{s'}$ and $\boldsymbol{t'}$
are connected in $G(\phi')$ iff $\boldsymbol{s}$ and $\boldsymbol{t}$
are connected in $G(\phi)$, and such that $G(\phi')$ is connected
iff $G(\phi)$ is connected.

While the construction of $\phi'$ is quite easy for this lemma, the
construction for the next two lemmas is analogous but more intricate,
so we proceed carefully in two steps, which we will adapt in the next
two proofs: In the first step, we give a transformation $T$ that
transforms any 1-reproducing formula $\psi$ into a connectivity-equivalent
formula $T_{\psi}\in\mathsf{S}_{12}$ built from the standard connectives.
Since $\mathsf{S}_{12}\subseteq[B]$, we can express $T_{\psi}$ as
a $B$-formula $T_{\psi}^{*}$. Now if we would apply $T$ to $\phi$
directly, we would know that $T_{\phi}$ can be expressed as a $B$-formula.
However, this could lead to an exponential increase in the formula
size (see \prettyref{sec:Cir}), so we have to show how to construct
the $B$-formula in polynomial time. For this, in the second step,
we construct a $B$-formula $\phi'$ directly from $\phi$ (by applying
$T$ to the clauses and the $\wedge$'s individually), and then show
that $\phi'$ is equivalent to $T_{\phi}$; thus we know that $\phi'$
is connectivity-equivalent to $\phi$.

\emph{Step 1.} From Table \ref{tab:List-of-all}, we find that $\mathsf{S}_{12}=\mathsf{S}_{1}\cap\mathsf{R}_{2}=\mathsf{S}_{1}\cap\mathsf{R}_{0}\cap\mathsf{R}_{1}$,
so we have to make sure that $T_{\psi}$ is 1-seperating, 0-reproducing,
and 1-reproducing. Let 
\[
T_{\psi}=\psi\wedge y,
\]
where $y$ is a new variable.

All solutions $\boldsymbol{a}$ of $T_{\psi}(\boldsymbol{x},y)$ have
$a_{n+1}=1$, so $T_{\psi}$ is 1-seperating and 0-reproducing; also,
$T_{\psi}$ is still 1-reproducing. Further, for any two solutions
$\boldsymbol{s}$ and $\boldsymbol{t}$ of $\psi(\boldsymbol{x})$,
$\boldsymbol{s}'=\boldsymbol{s}\cdot1$ and $\boldsymbol{t}'=\boldsymbol{t}\cdot1$
are solutions of $T_{\psi}(\boldsymbol{x},y)$, and it is easy to
see that they are connected in $G(T_{\psi})$ iff $\boldsymbol{s}$
and $\boldsymbol{t}$ are connected in $G(\psi)$, and that $G(T_{\psi})$
is connected iff $G(\psi)$ is connected.

\emph{Step 2.} The idea is to parenthesize the conjunctions of $\phi$
such that we get a tree of $\wedge$'s of depth logarithmic in the
size of $\phi$, and then to replace each clause and each $\wedge$
with an equivalent $B$-formula. This can increase the formula size
by only a polynomial in the original size even if the $B$-formula
equivalent to $\wedge$ uses some input variable more than once.

Let $\phi=C_{1}\wedge\cdots\wedge C_{n}$ be a 1-reproducing 3-CNF-formula.
Since $\phi$ is 1-reproducing, every clause $C_{i}$ of $\phi$ is
itself 1-reproducing, and we can express $T_{C_{i}}$ through a $B$-formula
$T_{C_{i}}^{*}$. Also, we can express $T_{u\wedge v}$ through a
$B$-formula $T_{u\wedge v}^{*}$ since $\wedge$ is 1-reproducing;
we write $T_{\wedge}(\psi_{1},\psi_{2})$ for the formula obtained
from $T_{u\wedge v}$ by substituting the formula $\psi_{1}$ for
$u$ and $\psi_{2}$ for $v$, and similarly write $T_{\wedge}^{*}(\psi_{1},\psi_{2})$
for the formula obtained from $T_{u\wedge v}^{*}$ in this way. We
let $\phi'=$\noun{Tr$(\phi)$}, where \noun{Tr} is the following
recursive algorithm that takes a CNF-formula as input:\\
\\
Algorithm \noun{Tr}$\left(\psi_{1}\wedge\cdots\wedge\psi_{m}\right)$
\begin{itemize}[label= ]
\item If $m=1$, return $T_{\psi_{1}}^{*}$.
\item Else return $T_{\wedge}^{*}\left(\textsc{Tr}(\psi_{1}\wedge\cdots\wedge\psi_{\left\lfloor m/2\right\rfloor }),\textsc{Tr}(\psi_{\left\lfloor m/2\right\rfloor +1}\wedge\cdots\wedge\psi_{m})\right)$.
\end{itemize}
Since the recursion terminates after a number of steps logarithmic
in the number of clauses of $\phi$, and every step increases the
total formula size by only a constant factor, the algorithm runs in
polynomial time. We show $\phi'\equiv T_{\phi}$ by induction on $m$.
For $m=1$ this is clear. For the induction step, we have to show
$T_{\wedge}^{*}(T_{\psi_{1}},T_{\psi_{2}})\equiv T_{\psi_{1}\wedge\psi_{2}}$,
but since $T_{\wedge}(\psi_{1},\psi_{2})\equiv T_{\wedge}^{*}(\psi_{1},\psi_{2})$,
it suffices to show that $T_{\wedge}(T_{\psi_{1}},T_{\psi_{2}})\equiv T_{\psi_{1}\wedge\psi_{2}}$:

\[
T_{\wedge}(T_{\psi_{1}},T_{\psi_{2}})=(\psi_{1}\wedge y)\wedge(\psi_{2}\wedge y)\wedge y\equiv\psi_{1}\wedge\psi_{2}\wedge y=T_{\psi_{1}\wedge\psi_{2}}.
\]

2. This follows from 1.\ by Proposition \ref{pro:trf}.

3. By \prettyref{lem:l37}, there is an 1-reproducing $(n-1)$-ary
function $f$ with diameter of at least $2^{\left\lfloor \frac{n-1}{2}\right\rfloor }$.
Let $f$ be represented by a formula $\phi$; then, $T_{\phi}$ represents
an $n$-ary function of the same diameter in $\mathsf{S}_{12}$.\end{proof}
\begin{lem}
\label{lem:If}If $[B]\supseteq\mathsf{D}_{1}$,
\begin{enumerate}
\item \noun{st-BF-Conn(}$B$\emph{)} and \noun{BF-Conn(}$B$\emph{)} are
$\mathrm{PSPACE}$-complete,
\item \noun{st-Circ-Conn(}$B$\emph{)} and \noun{Circ-Conn(}$B$\emph{)}
are $\mathrm{PSPACE}$-complete,
\item for $n\geq5$, there is an $n$-ary function $f\in[B]$ with diameter
of at least $2^{\left\lfloor \frac{n-3}{2}\right\rfloor }$.
\end{enumerate}
\end{lem}
\begin{proof}
1.$\;$As noted, we adapt the two steps from the previous proof.

\emph{Step 1. }Since $\mathsf{D}_{1}=\mathsf{D}\cap\mathsf{R}_{0}\cap\mathsf{R}_{1}$,
$T_{\psi}$ must be self-dual, 0-reproducing, and 1-reproducing. For
clarity, we first construct an intermediate formula $T_{\psi}^{\sim}\in\mathsf{D}_{1}$
whose solution graph has an additional component, then we eliminate
that component.

For $\psi(\boldsymbol{x})$, let
\[
T_{\psi}^{\sim}=\left(\psi(\boldsymbol{x})\wedge(\boldsymbol{y}=\boldsymbol{1})\right)\vee\left(\overline{\psi(\boldsymbol{\overline{x}})}\wedge(\boldsymbol{y}=\boldsymbol{0})\right)\vee\left(\boldsymbol{y}\in\left\{ 100,010,001\right\} \right),
\]
where $\boldsymbol{y}=(y_{1},y_{2},y_{3})$ are three new variables.

$T_{\psi}^{\sim}$ is self-dual: for any solution ending with 111
(satisfying the first disjunct), the inverse vector is no solution;
similarly, for any solution ending with 000 (satisfying the second
disjunct), the inverse vector is no solution; finally, all vectors
ending with 100, 010, or 001 are solutions and their inverses are
no solutions. Also, $T_{\psi}^{\sim}$ is still 1-reproducing, and
it is 0-reproducing (for the second disjunct note that $\overline{\psi(\overline{0\cdots0})}\equiv\overline{\psi(1\cdots1)}\equiv0$).

Further, every solution $\boldsymbol{a}$ of $\psi$ corresponds to
a solution $\boldsymbol{a}\cdot111$ of $T_{\psi}^{\sim}$, and for
any two solutions $\boldsymbol{s}$ and $\boldsymbol{t}$ of $\psi$,
$\boldsymbol{s}'=\boldsymbol{s}\cdot111$ and $\boldsymbol{t}'=\boldsymbol{t}\cdot111$
are connected in $G(T_{\psi}^{\sim})$ iff $\boldsymbol{s}$ and $\boldsymbol{t}$
are connected in $G(\psi)$: The ``if'' is clear, for the ``only
if'' note that since there are no solutions of $T_{\psi}^{\sim}$
ending with 110, 101, or 011, every solution of $T_{\psi}^{\sim}$
not ending with 111 differs in at least two variables from the solutions
that do.

Observe that exactly one connected component is added in $G(T_{\psi}^{\sim})$
to the components corresponding to those of $G(\psi)$: It consists
of all solutions ending with 000, 100, 010, or 001 (any two vectors
ending with 000 are connected e.g. via those ending with 100). It
follows that $G(T_{\psi}^{\sim})$ is always unconnected. To fix this,
we modify $T_{\psi}^{\sim}$ to $T_{\psi}$ by adding $1\cdots1\cdot110$
as a solution, thereby connecting $1\cdots1\cdot111$ (which is always
a solution since $T_{\psi}^{\sim}$ is 1-reproducing) with $1\cdots1\cdot100$,
and thereby with the additional component of $T_{\psi}$. To keep
the function self-dual, we must in turn remove $0\cdots0\cdot001$,
which does not alter the connectivity. Formally,

\begin{eqnarray}
T_{\psi} & = & \left(T_{\psi}^{\sim}\vee\left((\boldsymbol{x}=\boldsymbol{1})\wedge(\boldsymbol{y}=110)\right)\right)\wedge\neg\left((\boldsymbol{x}=\boldsymbol{0})\wedge(\boldsymbol{y}=001)\right)\label{eq:q1}\\
 & = & \left(\psi(\boldsymbol{x})\wedge(\boldsymbol{y}=\boldsymbol{1})\right)\vee\left(\overline{\psi(\boldsymbol{\overline{x}})}\wedge(\boldsymbol{y}=\boldsymbol{0})\right)\nonumber \\
 &  & \vee\left(\boldsymbol{y}\in\left\{ 100,010,001\right\} \wedge\neg((\boldsymbol{x}=\boldsymbol{0})\wedge(\boldsymbol{y}=001))\right)\nonumber \\
 &  & \vee((\boldsymbol{x}=\boldsymbol{1})\wedge(\boldsymbol{y}=110)).\nonumber 
\end{eqnarray}

\begin{figure}[!h]
\begin{centering}
\includegraphics[scale=0.42]{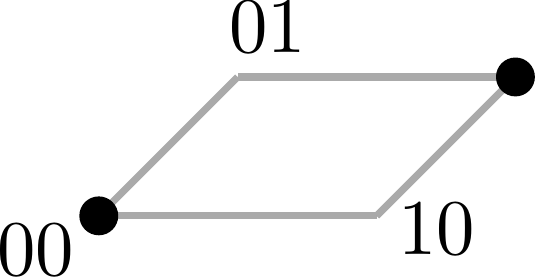}$\qquad$\includegraphics[scale=0.82]{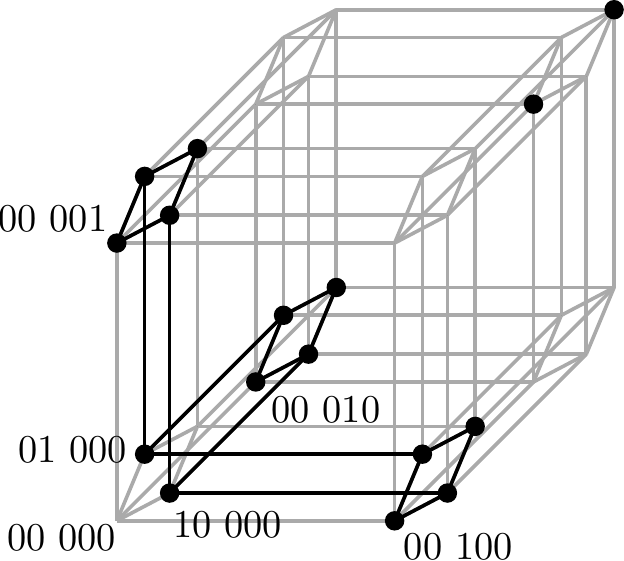}$\qquad$\includegraphics[scale=0.82]{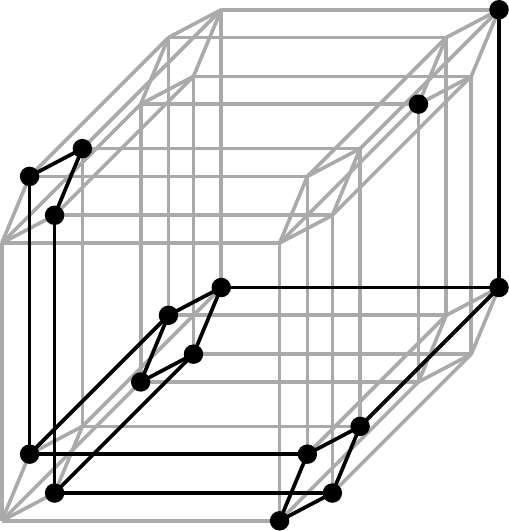}
\par\end{centering}

\protect\caption[An example for the transformation in the proof of \prettyref{lem:If}]{\emph{}An example for the transformation. Left: $\psi=\left(x_{1}\vee\overline{x_{2}}\right)\wedge\left(\overline{x_{1}}\vee x_{2}\right)$,
center: $T_{\psi}^{\sim}$, right: $T_{\psi}$. The \textquotedblleft axis
vertices\textquotedblright{} are labeled in the first two graphs.}
\end{figure}

Now $G(T_{\psi})$ is connected iff $G(\psi)$ is connected.

\emph{Step 2.} Again, we use the algorithm \noun{Tr} from the previous
proof to transform any 1-reproducing 3-CNF-formula $\phi$ into a
$B$-formula $\phi'$ equivalent to $T_{\phi}$, but with the definition
\prettyref{eq:q1} of $T$. Again, we have to show $T_{\wedge}(T_{\psi_{1}},T_{\psi_{2}})\equiv T_{\psi_{1}\wedge\psi_{2}}$.
Here, 
\begin{eqnarray*}
T_{\wedge}(T_{\psi_{1}},T_{\psi_{2}}) & = & \left(T_{\psi_{1}}\wedge T_{\psi_{2}}\wedge(\boldsymbol{y}=\boldsymbol{1})\right)\vee\left(\overline{\overline{T_{\psi_{1}}}\wedge\overline{T_{\psi_{2}}}}\wedge(\boldsymbol{y}=\boldsymbol{0})\right)\\
 &  & \vee\left(\boldsymbol{y}\in\left\{ 100,010,001\right\} \wedge\neg\left(\overline{T_{\psi_{1}}}\wedge\overline{T_{\psi_{2}}}\wedge(\boldsymbol{y}=001)\right)\right)\\
 &  & \vee\left(T_{\psi_{1}}\wedge T_{\psi_{2}}\wedge(\boldsymbol{y}=110)\right).
\end{eqnarray*}
We consider the parts of the formula in turn: For any formula $\xi$
we have $T_{\xi}(\boldsymbol{x}_{\xi})\wedge(\boldsymbol{y}=\boldsymbol{1})\equiv\xi(\boldsymbol{x}_{\xi})\wedge(\boldsymbol{y}=\boldsymbol{1})$
and $T_{\xi}(\boldsymbol{x}_{\xi})\wedge(\boldsymbol{y}=\boldsymbol{0})\equiv\overline{\psi(\overline{\boldsymbol{x}_{\xi}})}\wedge(\boldsymbol{y}=\boldsymbol{0})$,
where $\boldsymbol{x}_{\xi}$ denotes the variables of $\xi$. Using
$\overline{\overline{T_{\psi_{1}}(\boldsymbol{x}_{\psi_{1}})}\wedge\overline{T_{\psi_{2}}(\boldsymbol{x}_{\psi_{2}})}}\wedge(\boldsymbol{y}=\boldsymbol{0})=\left(T_{\psi_{1}}(\boldsymbol{x}_{\psi_{1}})\vee T_{\psi_{2}}(\boldsymbol{x}_{\psi_{2}})\right)\wedge(\boldsymbol{y}=\boldsymbol{0})$,
the first line becomes 
\[
\left(\psi_{1}(\boldsymbol{x}_{\psi_{1}})\wedge\psi_{2}(\boldsymbol{x}_{\psi_{2}})\wedge(\boldsymbol{y}=\boldsymbol{1})\right)\vee\left(\left(\overline{\psi_{1}(\overline{\boldsymbol{x}_{\psi_{1}}})\wedge\psi_{2}(\overline{\boldsymbol{x}_{\psi_{2}}})}\right)\wedge(\boldsymbol{y}=\boldsymbol{0})\right).
\]
For the second line, we observe 
\begin{eqnarray*}
\overline{T_{\psi}(\boldsymbol{x}_{\psi})} & \equiv & \left(\overline{\psi(\boldsymbol{x}_{\psi})}\vee\neg(\boldsymbol{y}=\boldsymbol{1})\right)\wedge\left(\psi(\boldsymbol{\overline{x}_{\psi}})\vee\neg(\boldsymbol{y}=\boldsymbol{0})\right)\\
 &  & \wedge\left(\boldsymbol{y}\notin\left\{ 100,010,001\right\} \vee\left((\boldsymbol{x}_{\psi}=\boldsymbol{0})\wedge(\boldsymbol{y}=001)\right)\right)\\
 &  & \wedge(\neg(\boldsymbol{x}_{\psi}=\boldsymbol{1})\vee\overline{(\boldsymbol{y}=110)}),
\end{eqnarray*}
thus $\overline{T_{\psi}(\boldsymbol{x}_{\psi})}\wedge(\boldsymbol{y}=001)\equiv(\boldsymbol{x}_{\psi}=\boldsymbol{0})\wedge(\boldsymbol{y}=001)$,
and the second line becomes
\[
\vee\left(\boldsymbol{y}\in\left\{ 100,010,001\right\} \wedge\neg\left((\boldsymbol{x}_{\psi_{1}}=\boldsymbol{0})\wedge(\boldsymbol{x}_{\psi_{2}}=\boldsymbol{0})\wedge(\boldsymbol{y}=001)\right)\right).
\]

Since $T_{\psi}(\boldsymbol{x}_{\psi})\wedge(\boldsymbol{y}=110)\equiv(\boldsymbol{x}_{\psi}=\boldsymbol{1})\wedge(\boldsymbol{y}=110)$
for any $\psi$, the third line becomes
\[
\vee\left((\boldsymbol{x}_{\psi_{1}}=\boldsymbol{1})\wedge(\boldsymbol{x}_{\psi_{2}}=\boldsymbol{1})\wedge(\boldsymbol{y}=110)\right).
\]
Now $T_{\wedge}(T_{\psi_{1}},T_{\psi_{2}})$ equals
\begin{eqnarray*}
T_{\psi_{1}\wedge\psi_{2}} & = & \left(\psi_{1}(\boldsymbol{x}_{\psi_{1}})\wedge\psi_{2}(\boldsymbol{x}_{\psi_{2}})\wedge(\boldsymbol{y}=\boldsymbol{1})\right)\vee\left(\overline{\psi_{1}(\overline{\boldsymbol{x}_{\psi_{1}}})\wedge\psi_{2}(\overline{\boldsymbol{x}_{\psi_{2}}})}\wedge(\boldsymbol{y}=\boldsymbol{0})\right)\\
 &  & \vee\left(\boldsymbol{y}\in\left\{ 100,010,001\right\} \wedge\neg\left((\boldsymbol{x}_{\psi_{1}}=\boldsymbol{0})\wedge(\boldsymbol{x}_{\psi_{2}}=\boldsymbol{0})\wedge(\boldsymbol{y}=001)\right)\right)\\
 &  & \vee\left((\boldsymbol{x}_{\psi_{1}}=\boldsymbol{1})\wedge(\boldsymbol{x}_{\psi_{2}}=\boldsymbol{1})\wedge(\boldsymbol{y}=110)\right).
\end{eqnarray*}

2. This follows from 1.\ by Proposition \ref{pro:trf}.

3. By \prettyref{lem:l37} there is an 1-reproducing $(n-3)$-ary
function $f$ with diameter of at least $2^{\left\lfloor \frac{n-3}{2}\right\rfloor }$.
Let $f$ be represented by a formula $\phi$; then, $T_{\phi}$ represents
an $n$-ary function of the same diameter in $\mathsf{D}_{1}$.\end{proof}
\begin{lem}
If $[B]\supseteq\mathsf{S}_{02}^{k}$ for any $k\geq2$,
\begin{enumerate}
\item \noun{st-BF-Conn(}$B$\emph{)} and \noun{BF-Conn(}$B$\emph{)} are
$\mathrm{PSPACE}$-complete,
\item \noun{st-Circ-Conn(}$B$\emph{)} and \noun{Circ-Conn(}$B$\emph{)}
are $\mathrm{PSPACE}$-complete,
\item for $n\geq k+4$, there is an $n$-ary function $f\in[B]$ with diameter
of at least $2^{\left\lfloor \frac{n-k-2}{2}\right\rfloor }$.
\end{enumerate}
\end{lem}
\begin{proof}
1. \emph{Step 1. }Since $\mathsf{S}_{02}^{k}=\mathsf{S}_{0}^{k}\cap\mathsf{R}_{0}\cap\mathsf{R}_{1}$,
$T_{\psi}$ must be 0-separating of degree $k$, 0-reproducing, and
1-reproducing. As in the previous proof, we construct an intermediate
formula $T_{\psi}^{\sim}$. For $\psi(\boldsymbol{x})$, let 
\[
T_{\psi}^{\sim}=\left(\psi\wedge y\wedge(\boldsymbol{z}=\boldsymbol{0})\right)\vee(|\boldsymbol{z}|>1),
\]
where $y$ and $\boldsymbol{z}=(z_{1},\ldots,z_{k+1})$ are new variables.

$T_{\psi}^{\sim}(\boldsymbol{x},y,\boldsymbol{z})$ is 0-separating
of degree $k$, since all vectors that are no solutions of $T_{\psi}^{\sim}$
have $|\boldsymbol{z}|\leq1$, i.e. $\boldsymbol{z}\in\left\{ 0\cdots0,10\cdots0,010\cdots0,\ldots,0\cdots01\right\} \subset\{0,1\}^{k+1}$,
and thus any $k$ of them have at least one common variable assigned
0. Also, $T_{\psi}^{\sim}$ is 0-reproducing and still 1-reproducing.

Further, for any two solutions $\boldsymbol{s}$ and $\boldsymbol{t}$
of $\psi(\boldsymbol{x})$, $\boldsymbol{s}'=\boldsymbol{s}\cdot1\cdot0\cdots0$
and $\boldsymbol{t}'=\boldsymbol{t}\cdot1\cdot0\cdots0$ are solutions
of $T_{\psi}^{\sim}(\boldsymbol{x},y,\boldsymbol{z})$ and are connected
in $G(T_{\psi}^{\sim})$ iff $\boldsymbol{s}$ and $\boldsymbol{t}$
are connected in $G(\psi)$.

But again, we have produced an additional connected component (consisting
of all solutions with $|\boldsymbol{z}|>1$). To connect it to a component
corresponding to one of $\psi$, we add $1\cdots1\cdot1\cdot10\cdots0$
as a solution,
\begin{eqnarray*}
T_{\psi} & = & \left(\psi\wedge y\wedge(\boldsymbol{z}=\boldsymbol{0})\right)\vee(|\boldsymbol{z}|>1)\vee\left((\boldsymbol{x}=\boldsymbol{1})\wedge y\wedge(\boldsymbol{z}=10\cdots0)\right).
\end{eqnarray*}
Now $G(T_{\psi})$ is connected iff $G(\psi)$ is connected.

\emph{Step 2.} Again we show that the algorithm \noun{Tr} works in
this case. Here,
\begin{eqnarray*}
T_{\wedge}(T_{\psi_{1}},T_{\psi_{2}}) & = & \left(T_{\psi_{1}}(\boldsymbol{x}_{\psi_{1}})\wedge T_{\psi_{2}}(\boldsymbol{x}_{\psi_{2}})\wedge y\wedge(\boldsymbol{z}=\boldsymbol{0})\right)\vee(|\boldsymbol{z}|>1)\\
 &  & \vee\left(T_{\psi_{1}}(\boldsymbol{x}_{\psi_{1}})\wedge T_{\psi_{2}}(\boldsymbol{x}_{\psi_{2}})\wedge y\wedge(\boldsymbol{z}=10\cdots0)\right).
\end{eqnarray*}
Since $T_{\psi}(\boldsymbol{x}_{\psi})\wedge y\wedge(\boldsymbol{z}=\boldsymbol{0})\equiv\psi(\boldsymbol{x}_{\psi})\wedge y\wedge(\boldsymbol{z}=\boldsymbol{0})$
and $T_{\psi}(\boldsymbol{x}_{\psi})\wedge y\wedge(\boldsymbol{z}=10\cdots0)\equiv(\boldsymbol{x}_{\psi}=1)\wedge y\wedge(\boldsymbol{z}=10\cdots0)$
for any $\psi$, this is equivalent to
\begin{eqnarray*}
T_{\psi_{1}\wedge\psi_{2}} & = & \left(\psi_{1}(\boldsymbol{x}_{\psi_{1}})\wedge\psi_{2}(\boldsymbol{x}_{\psi_{2}})\wedge y\wedge(\boldsymbol{z}=\boldsymbol{0})\right)\vee(|\boldsymbol{z}|>1)\\
 &  & \vee\left(\boldsymbol{x}_{\psi_{1}}\wedge\boldsymbol{x}_{\psi_{2}}\wedge y\wedge(\boldsymbol{z}=10\cdots0)\right).
\end{eqnarray*}

2. This follows from 1.\ by Proposition \ref{pro:trf}.

3. By \prettyref{lem:l37} there is an 1-reproducing $(n-k-2)$-ary
function $f$ with diameter of at least $2^{\left\lfloor \frac{n-k-2}{2}\right\rfloor }$.
Let $f$ be represented by a formula $\phi$; then, $T_{\phi}$ represents
an $n$-ary function of the same diameter in $\mathsf{S}_{02}^{k}$.
\end{proof}
This completes the proof of \prettyref{thm:func}.

\section{\label{sub:Quantified-case}Quantified Formulas}
\begin{defn}
A \emph{quantified $B$-formula} $\phi$ (in prenex normal form) is
an expression of the form 
\[
Q_{1}y_{1}\cdots Q_{m}y_{m}\varphi(y_{1},\ldots,y_{m},x_{1},\ldots,x_{n}),
\]
where $\varphi$ is a $B$-formula, and $Q_{1},\ldots,Q_{m}\in\{\exists,\forall\}$
are quantifiers. We denote the corresponding connectivity resp. $st$-connectivity
problems by \noun{QBF-Conn($B$) }resp.\noun{ st-QBF-Conn($B$).}\end{defn}
\begin{thm}
\label{thm:quan}Let $B$ be a finite set of Boolean functions.
\begin{enumerate}
\item If $B\subseteq\mathsf{M}$ or $B\subseteq\mathsf{L}$, then

\begin{enumerate}
\item \noun{st-QBF-Conn(}$B$\emph{)} and \noun{QBF-Conn(}$B$\emph{)} are
in \noun{P},
\item the diameter of every quantified $B$-formula is linear in the number
of free variables.
\end{enumerate}
\item Otherwise,

\begin{enumerate}
\item \noun{st-QBF-Conn(}$B$\emph{)} and \noun{QBF-Conn(}$B$\emph{)} are
$\mathrm{PSPACE}$-complete,
\item there are quantified $B$-formulas with at most one quantifier such
that their diameter is exponential in the number of free variables.
\end{enumerate}
\end{enumerate}
\end{thm}
\begin{proof}
1. For $B\subseteq\mathsf{M}$, any quantified $B$-formula $\phi$
represents a monotone function: Using $\exists y\psi(y,\boldsymbol{x})=\psi(0,\boldsymbol{x})\vee\psi(1,\boldsymbol{x})$
and $\forall y\psi(y,\boldsymbol{x})=\psi(0,\boldsymbol{x})\wedge\psi(1,\boldsymbol{x})$
recursively, we can transform $\phi$ into an equivalent $\mathsf{M}$-formula
since $\wedge$ and $\vee$ are monotone. Thus as in \prettyref{lem:M},
\noun{st-QBF-Conn(}$B$) and \noun{QBF-Conn(}$B$) are trivial, and
$d_{f}(\boldsymbol{a},\boldsymbol{b})=|\boldsymbol{a}-\boldsymbol{b}|$
for any two solutions $\boldsymbol{a}$ and $\boldsymbol{b}$.

For a quantified $B$-formula $\phi=Q_{1}y_{1}\cdots Q_{m}y_{m}\varphi$
with $B\subseteq\mathsf{L}$, we first remove the quantifications
over all fictive variables of $\varphi$ (and eliminate the fictive
variables if necessary). If quantifiers remain, $\phi$ is either
tautological (if the rightmost quantifier is $\exists$) or unsatisfiable
(if the rightmost quantifier is $\forall$), so the problems are trivial,
and $d_{f}(\boldsymbol{a},\boldsymbol{b})=|\boldsymbol{a}-\boldsymbol{b}|$
for any two solutions $\boldsymbol{a}$ and $\boldsymbol{b}$. Otherwise,
we have a quantifier-free formula and the statements follow from \prettyref{lem:L}.

2. Again as in \prettyref{lem:l36}, it follows that \noun{st-QBF-Conn(}$B$)
and \noun{QBF-Conn(}$B$) are in $\mathrm{PSPACE}$, since the evaluation
problem for quantified $B$-formulas is in $\mathrm{PSPACE}$ \citep{Schaefer:1978:CSP:800133.804350}.

An inspection of Post's lattice shows that if $B\nsubseteq\mathsf{M}$
and $B\nsubseteq\mathsf{L}$, then $[B]\supseteq\mathsf{S}_{12}$,
$[B]\supseteq\mathsf{D}_{1}$, or $[B]\supseteq\mathsf{S}_{02}$,
so we have to prove $\mathrm{PSPACE}$-completeness and show the existence
of $B$-formulas with an exponential diameter in these cases.

For $[B]\supseteq\mathsf{S}_{12}$ and $[B]\supseteq\mathsf{D}_{1}$,
the statements for the $\mathrm{PSPACE}$-hardness and the diameter
obviously carry over from \prettyref{thm:func}.

For $B\supseteq\mathsf{S}_{02}$, we give a reduction from the problems
for (unquantified) 3-CNF-formulas; we proceeded again similar as in
the proof of \prettyref{lem:s12}. We give a transformation $T_{\psi}$
s.t. $T_{\psi}\in\mathsf{S}_{02}$ for all formulas $\psi$. Since
$\mathsf{S}_{02}=\mathsf{S}_{0}\cap\mathsf{R}_{0}\cap\mathsf{R}_{1}$,
$T_{\psi}$ must be self-dual, 0-reproducing, and 1-reproducing. For
$\psi(\boldsymbol{x})$ let 
\[
T_{\psi}=(\psi\wedge y)\vee z,
\]
with the two new variables $y$ and $z$.

$T_{\psi}$ is 0-separating since all vectors that are no solutions
have $z=0$. Also, $T_{\psi}$ is 0-reproducing and 1-reproducing.
Again, we use the algorithm \noun{Tr} from the proof of \prettyref{lem:s12}
to transform any 3-CNF-formula $\phi$ into a $B$-formula $\varphi'$
equivalent to $T_{\phi}$. Again, we show

\begin{eqnarray*}
T_{\wedge}(T_{\psi_{1}},T_{\psi_{2}}) & = & \left(\left((\psi_{1}\wedge y)\vee z\right)\wedge\left((\psi_{2}\wedge y)\vee z\right)\wedge y\right)\vee z\\
 & \equiv & \left(\left(\psi_{1}\wedge y\right)\wedge\left(\psi_{2}\wedge y\right)\wedge y\right)\vee z\\
 & \equiv & \left(\psi_{1}\wedge\psi_{2}\wedge y\right)\vee z=T_{\psi_{1}\wedge\psi_{2}}.
\end{eqnarray*}

Now let 
\[
\phi'=\forall z\varphi'.
\]
Then, for any two solutions $\boldsymbol{s}$ and $\boldsymbol{t}$
of $\phi(\boldsymbol{x})$, $\boldsymbol{s}'=\boldsymbol{s}\cdot1$
and $\boldsymbol{t}'=\boldsymbol{t}\cdot1$ are solutions of $\phi'(\boldsymbol{x},y)$,
and they are connected in $G(\phi')$ iff $\boldsymbol{s}$ and $\boldsymbol{t}$
are connected in $G(\phi)$, and $G(\phi')$ is connected iff $G(\phi)$
is connected.

The proof of \prettyref{lem:l37} shows that there is an $(n-1)$-ary
function $f$ with diameter of at least $2^{\left\lfloor \frac{n-1}{2}\right\rfloor }$.
Let $f$ be represented by a formula $\phi$; then $\phi'$ as defined
above is a quantified $B$-formula with $n$ free variables and one
quantifier with the same diameter.\end{proof}
\begin{rem}
An analog to \prettyref{thm:quan} also holds for quantified circuits
as defined in \citep[Section 7]{reith2000}.
\end{rem}

\chapter{Future Directions}

We have proved classifications in two quite different settings: In
Schaefer's framework for constraint satisfaction problems, specifically
CNF\textsubscript{C}($\mathcal{S}$)-formulas, CNF($\mathcal{S}$)-formulas
and Q-CNF\textsubscript{C}($\mathcal{S}$)-formulas, and in Post's
framework for nested formulas and circuits, specifically $B$-formulas,
$B$-circuits, and quantified $B$-formulas.

While we now have a quite complete picture for nested formulas and
circuits, and also for CNF\textsubscript{C}($\mathcal{S}$)-formulas
and Q-CNF\textsubscript{C}($\mathcal{S}$)-formulas, the complexity
of the connectivity problem for CNF($\mathcal{S}$)-formulas without
constants is still open for sets $\mathcal{S}$ that are\noun{ }0-valid,
1-valid, or complementive, but not Schaefer, nor nc-CPSS, nor quasi
disconnecting.

Recently, Scharpfenecker refined some of our complexity results for
CNF\textsubscript{C}($\mathcal{S}$)-formulas up to logarithmic-space
isomorphisms and investigated the realizable solution-graphs in more
detail \citep{scharpfenecker2015structure}.

It seems very likely that connectivity is not for all these sets in
P, since there are even Schaefer sets that are 0-valid or 1-valid,
but have a coNP-complete connectivity problem.  On the other hand,
it seems promising to search for more nc-CPS sets, for which we have
a P or $\mathrm{P^{NP}}$ algorithm.

Also, one might further explore the connection between the connectivity
of CNF-formulas and properties of the corresponding hypergraphs (see
Remark \ref{hyp}); while we used this connection only for Horn formulas
and thus dealt with hypergraphs of head-size 1, there might be useful
reductions between the connectivity problems for more general CNF-formulas
and problems for the corresponding hypergraphs.\\

There is a multitude of interesting variations of the problems investigated
in this thesis, in different directions; we close with a quick survey
of some such variations.
\begin{itemize}
\item \textbf{Other representations of Boolean relations} ~Disjunctive
normal forms with special connectivity properties were studied by
Ekin et al.\ already in 1997 for their ``important role in problems
appearing in various areas including in particular discrete optimization,
machine learning, automated reasoning, etc.'' \citep{ekin1999connected}.\\
{\small{}\hspace*{3ex}}There are yet more kinds of representations
for Boolean relations, such as binary decision diagrams and Boolean
neural networks, and investigating the connectivity in these settings
might be worthwhile as well.
\item \textbf{Related problems} ~Other connectivity-related problems already
mentioned by Gopalan et al.\ are counting the number of components
and approximating the diameter. For counting the number of components,
Gopalan et al.\ mentioned in \citep{Gopalan:2006} that they could
show that the problem is in P for affine, monotone and dual-monotone
relations, and \#P-complete otherwise.\\
{\small{}\hspace*{3ex}}Further, especially with regard to reconfiguration
problems, it is interesting to find the shortest path between two
solutions; this was recently investigated by Mouawad et al.\ \citep{mouawad2014shortest},
who proved a computational trichotomy for this problem. In this direction,
one could also consider the optimal path according to some other measure.
\item \textbf{Other definitions of connectivity} ~Our definition of connectivity
is not the only sensible one: One could regard two solutions connected
whenever their Hamming distance is at most $d$, for some $d\geq1$;
this was already considered related to random satisfiability, see
\citep{achlioptas2006solution}. This generalization seems meaningful
as well as challenging.
\item \textbf{Higher domains} ~Finally, a most interesting subject are
CSPs over larger domains; in 1993, Feder and Vardi conjectured a dichotomy
for the satisfiability problem over arbitrary finite domains \citep{feder1998computational},
and while the conjecture was proved for domains of size three in 2002
by Bulatov \citep{bulatov2002dichotomy}, it remains open to date
for the general case. Close investigation of the solution space might
lead to valuable insights here.\\
{\small{}\hspace*{3ex}}For $k$-colorability, which is a special
case of the general CSP over a $k$-element set, the connectivity
problems and the diameter were already studied by Bonsma and Cereceda
\citep{bonsma2009finding}, and Cereceda, van den Heuvel, and Johnson
\citep{cereceda2011finding}. They showed that for $k=3$ the diameter
is at most quadratic in the number of vertices and the $st$-connectivity
problem is in P, while for $k\geq4$, the diameter can be exponential
and $st$-connectivity is PSPACE-complete in general.
\end{itemize}
\bibliographystyle{amsalpha}
\phantomsection\addcontentsline{toc}{chapter}{\bibname}\bibliography{jib}

\end{document}